\documentclass{article}

\usepackage{arxiv}

\usepackage[utf8]{inputenc} 
\usepackage[T1]{fontenc}    
\usepackage{hyperref}       
\usepackage{url}            
\usepackage{booktabs}       
\usepackage{amsfonts}       
\usepackage{nicefrac}       
\usepackage{microtype}      
\usepackage{lipsum}
\usepackage{graphicx}
\usepackage{graphicx}
\usepackage{balance}  

\usepackage{comment}
\usepackage{tabularx}
\usepackage{paralist}
\usepackage{booktabs}

\usepackage{amssymb}
\usepackage{amsthm}

\newcommand{\eat}[1]{\ignorespaces}

\usepackage{array}
\newcolumntype{P}[1]{>{\centering\arraybackslash}p{#1}}
\newcolumntype{M}[1]{>{\centering\arraybackslash}m{#1}}
\newcolumntype{R}[1]{>{\arraybackslash}m{#1}}

\usepackage{enumerate}

\usepackage{comment}
\usepackage{paralist}
\usepackage{nicefrac}

\usepackage{nicefrac}

\usepackage{mathtools}
\usepackage{blkarray, bigstrut} 

\usepackage{tikz}
\usepackage{verbatim}
\usetikzlibrary{arrows}
\usetikzlibrary{shapes,snakes}
\usetikzlibrary{decorations.pathmorphing} 
\usetikzlibrary{fit}					
\usetikzlibrary{backgrounds}	

\makeatletter
\global\let\tikz@ensure@dollar@catcode=\relax
\makeatother

\usepackage{subfigure}
\usepackage{xcolor}

\definecolor{thelightblue}{RGB}{0,191,255}
\definecolor{theblue}{RGB}{0,0,180}
\definecolor{mygrey}{gray}{0.6}

\usepackage{blkarray}
\usepackage{bigstrut, booktabs}

\makeatletter
\renewcommand*\env@matrix[1][*\c@MaxMatrixCols c]{
\hskip -\arraycolsep
\let\@ifnextchar\new@ifnextchar
\array{#1}}
\makeatother

\usepackage{booktabs} 
\usepackage{enumitem}

\usepackage{subfigure}
\usepackage{rotating}

\usepackage{multirow}

\usepackage{color}
\usepackage{xcolor}
\definecolor{mydarkblue}{RGB}{0, 20, 159} 
\definecolor{mydarkblue}{rgb}{0,0.08,0.45} 
\definecolor{mydarkblue}{rgb}{0,0.08,0.45} 

\usepackage{hyperref}
\hypersetup{%
    colorlinks=true,
    linkcolor=mydarkblue,
    citecolor=mydarkblue,
    filecolor=mydarkblue,
    urlcolor=mydarkblue}

\DeclareSymbolFont{cmbrightop}{OT1}{cmbr}{m}{n}
\DeclareMathSymbol{\sfPsi}{\mathalpha}{cmbrightop}{9}
\let\hat\widehat

\usepackage{tikz}
\usepackage{verbatim}
\usetikzlibrary{arrows}
\usetikzlibrary{shapes,snakes}
\usetikzlibrary{decorations.pathmorphing} 
\usetikzlibrary{fit}					
\usetikzlibrary{backgrounds}	

\definecolor{gray}{RGB}{150,150,150}
\definecolor{theblue}{RGB}{0, 20, 159} 

\definecolor{myyellow}{RGB}{255,255,204}
\definecolor{myred}{RGB}{255,204,204}
\definecolor{myblue}{RGB}{0,200,255}
\definecolor{mygreen}{RGB}{80,220,80}

\newcommand{\eg}{\emph{e.g.}}
\newcommand{\ie}{\emph{i.e.}}

\newtheorem{Definition}{Definition}
  
\theoremstyle{definition}

\usepackage{balance}

\usepackage{epsfig}

\usepackage{graphicx}

\usepackage{amssymb}
\usepackage{amsmath}
\usepackage{amsfonts}

\usepackage{tabularx}
\usepackage{booktabs}
\usepackage{relsize}

\usepackage{numprint} 

\usepackage{setspace}
\graphicspath{{./}{./graphics/}}
\newcolumntype{H}{>{\setbox0=\hbox\bgroup}c<{\egroup}@{}}

\newcommand{\norm}[1]{\|#1\|}

\newcommand{\figsz}{0.24}

\newtheorem{theorem}{Theorem}

\newcommand{\be}{\begin{equation}}
\newcommand{\ee}{\end{equation}}
\newcommand{\bea}{\begin{eqnarray}}
\newcommand{\eea}{\end{eqnarray}}

\relax
\def\up#1{^{(#1)}}

\let\hat\widehat
\let\tilde\widetilde

\def\E{{\mathbb E}}
\def\Pr{{\mathbb P}}

\def\cF{{\mathcal F}}
\def\cH{{\mathcal H}}

\def\cZ{{\mathcal Z}}

\def\cM{{\mathcal M}}

\def\var{\mathop{\mathrm{Var}}}

\def\argmin{\mathop{\mathrm{arg\,min}}}

\newcommand{\parab}[1]{\paraspace\noindent{\textbf{#1}}}

\newcommand{\paraspace}{\vspace{0.05in}}

\usepackage{algorithm}
\usepackage{algorithmicx}
\usepackage{algpseudocode}

\usepackage{xpatch}

\algblockdefx[parallel]{ParFor}{EndPar}[1][]{$\textbf{parallel for}$ #1 $\textbf{do}$}{$\textbf{end parallel}$}
\algrenewcommand{\alglinenumber}[1]{\fontsize{6.5}{7}\selectfont#1}
\algtext*{EndPar}

\algblockdefx[parallel]{parfor}{endpar}[1][]{$\textbf{parallel for}$ #1 $\textbf{do}$}{$\textbf{end parallel}$}
\algrenewcommand{\alglinenumber}[1]{\scriptsize#1:}

\algblockdefx[parallel]{parallelfor}{parallelend}
[1][]{\textbf{parallel for} #1}
{\textbf{end parallel}}

\xpatchcmd{\algorithmic}{\setcounter}{\algorithmicfont\setcounter}{}{}
\providecommand{\algorithmicfont}{}

\usepackage{nicefrac}

\algnotext{EndFor}
\algnotext{EndIf}
\algnotext{EndWhile}
\algnotext{EndProcedure}
\algnotext{EndParFor}
\algnotext{EndFunction}

\newcommand{\algrule}[1][.2pt]{\par\vskip.5\baselineskip\hrule height #1\par\vskip.5\baselineskip}

\title{Temporal Network Sampling}

\author{
 Nesreen K. Ahmed \\
  Intel Labs\\
  Santa Clara, CA 95054 \\
  \texttt{nesreen.k.ahmed2@intel.com} \\
   \And
 Nick Duffield \\
  Texas A\&M University\\
  College Station, TX 77843 \\
  \texttt{duffieldng@tamu.edu} \\
  \And
 Ryan A. Rossi \\
  Adobe Research\\
  San Jose, CA 95110 \\
  \texttt{rrossi@adobe.com} \\
}

\begin{document}

\maketitle

\begin{abstract}
Temporal networks representing a stream of timestamped edges are seemingly ubiquitous in the real-world.
However, the massive size and continuous nature of these networks make them fundamentally challenging to analyze and leverage for descriptive and predictive modeling tasks. In this work, we propose a general framework for temporal network sampling with unbiased estimation. We develop online, single-pass sampling algorithms and unbiased estimators for temporal network sampling. The proposed algorithms enable fast, accurate, and memory-efficient statistical estimation of temporal network patterns and properties. In addition, we propose a temporally decaying sampling algorithm with unbiased estimators for studying networks that evolve in continuous time, where the strength of links is a function of time, and the motif patterns are temporally-weighted. In contrast to the prior notion of a $\bigtriangleup t$-temporal motif, the proposed formulation and algorithms for counting temporally weighted motifs are useful for forecasting tasks in networks such as predicting future links, or a future time-series variable of nodes and links. Finally, extensive experiments on a variety of temporal networks from different domains demonstrate the effectiveness of the proposed algorithms. A detailed ablation study is provided to understand the impact of the various components of the proposed framework.
\end{abstract}

\section{Introduction}
\label{sec:intro}

Networks provide a natural framework to model and analyze complex systems of interacting entities in various domains (\eg, social, neural, communication, and technological domains)~\cite{newman2018networks,newman2006structure}. 
Most complex networked systems of scientific interest are continuously evolving in time, while entities interact continuously, and different entities may enter or exit the system at different times. 
The accurate modeling and analysis of these complex systems largely depend on the network representation~\cite{holme2015modern}. Therefore, it is crucial to incorporate both the heterogeneous \emph{structural} and \emph{temporal} information into network representations~\cite{nguyen2018continuous,rossi2012time,sharan2008temporal}. By incorporating the temporal information alongside the structural information, we obtain time-varying networks, also called \emph{temporal networks}~\cite{holme2012temporal}. 

In temporal networks, the nodes represent the entities in the system, and the links represent the interactions among these entities across time. Unlike static networks, nodes and links in temporal networks become active at certain times, leading to changes in the network structure over time~\cite{li2017fundamental}. Temporal networks have been recently used to model and analyze dynamic and streaming network data, \eg, to analyze and model information propagation~\cite{eckmann2004entropy,rocha2017dynamics}, epidemics~\cite{peixoto2018change}, infections~\cite{masuda2013predicting}, user influence~\cite{chen2013information,goyal2010learning}, among other applications~\cite{nguyen2018continuous,sharan2008temporal}. However, there are fundamental challenges to the analysis of temporal networks in real-world applications. One major challenge is the massive size and streaming characteristics of temporal network data that are generated by interconnected systems, since all interactions must be stored at any given time (\eg, email communications)~\cite{ahmed2014network}. As a result, several algorithms that were studied and designed for static networks that can fit in memory are becoming computationally intensive~\cite{ahmed2014graph}, due to their struggle to deal with the size and streaming properties of temporal networks. 

One common practice is to aggregate interactions in discrete time windows (time bins) (\eg, aggregate all interactions that appear in $1$-day or $1$-month), these are often called static \emph{graph snapshots}~\cite{soundarajan2016generating}. Given a graph snapshot, traditional techniques can be used to study and analyze the network (\eg, community detection, model learning, node ranking). Unfortunately, there are multiple challenges with employing these static aggregations. First, the choice of the size and placement of these time windows may alter the properties of the network and/or introduce a bias in the description of network dynamics~\cite{valdano2018epidemic,holme2019impact,sulo2010meaningful,caceres2011temporal}. For example, a small window size will likely miss important network sub-structures that span multiple windows (\eg, multi-node interactions such as motifs)~\cite{nguyen2018continuous}. On the other hand, a large window size will likely lose the temporal patterns in the data~\cite{fenn2012dynamical}. Second, modeling and analyzing bursty network traffic will likely be impacted by the placement of time windows. Finally, it is costly to consistently and reliably maintain these static aggregates for real-time applications~\cite{ahmed2014graph,ahmed2014network}. For example, it is often difficult to consistently gather these snapshots of graphs in one place, at one time, in an appropriate format for analysis. Thus, aggregates of network interactions in discrete time bins may not be an appropriate representation of temporal networks that evolve on a \emph{continuous-time} scale~\cite{flores2018eigenvector,aggarwal2006biased}, and can often lead to errors and bias the results~\cite{zino2016continuous,valdano2018epidemic,yang2018influential,nguyen2018continuous,zino2017analytical,soundarajan2016generating}.  

Statistical sampling is also common in studying networks, where the goal is to select a \emph{representative} sample (\ie, subnetwork) that serves as a proxy for the full network~\cite{lohr2019sampling}. Sampling algorithms are fundamental in studying and understanding networks~\cite{newman2018networks,kolaczyk2014statistical,ahmed2014network}. A sampled network is called \emph{representative}, if the characteristics of interest in the full network can be accurately estimated from the sample. Statistical sampling can provide a versatile framework to model and analyze network data. For example, when handling big data that cannot fit in memory, collecting data using limited storage/power electronic devices (\eg, mobile devices, RFID), or when the measurements required to observe the entire network are costly (\eg, protein interaction networks~\cite{stumpf2005subnets}).

While many network sampling techniques are studied in the context of small static networks that can fit entirely in memory~\cite{kolaczyk2014statistical} (\eg, uniform node sampling~\cite{stumpf2005subnets}, random walk sampling\cite{leskovec2006sampling}), recently there has been a growing interest in sampling techniques for streaming network data in which temporal networks evolve continuously in time~\cite{cormode2014second,jha2013space,ahmed2014graph,stefani2017triest,ahmed2017sampling,ahmedijcai18sampling,simpson2015catching,jha2015counting,pavan2013counting,lim2015mascot,ray2019efficient,choudhury2013streamworks} (see~\cite{ahmed2014network,mcgregor2014graph} for a survey). Most existing methods for sampling streaming network data have focused on the primary objective of selecting a sample to estimate static network properties (\eg, point statistics such as global triangle count or clustering coefficient). This poses an interesting and important question of how representative these samples of the characteristics of temporal networks that evolve on a \emph{continuous-time} scale~\cite{aggarwal2018extracting}, such as the link strength~\cite{xiang2010modeling}, link persistence~\cite{clauset2012persistence}, burstiness~\cite{barabasi2005origin}, temporal motifs~\cite{kovanen2011temporal}, among others~\cite{holme2012temporal,AggarwalSDM2020,larock2020understanding,fond2018designing,la2017ensemble}. Although this question is important, it has thus far not been addressed in the context of streaming and online methods.   

In this paper, we introduce an online importance sampling framework that extracts continuous-time dynamic network samples, in which the strength of a link (\ie, edge between two nodes) can evolve continuously as a function of time. Our proposed framework samples interactions to include in the sample based on their importance weight relative to the variable of interest (\ie, link strength), this enables sampling algorithms to adapt to the topological changes of temporal networks. Also, our proposed framework allows online and incremental updates, and can run efficiently in a single-pass over the data stream, where each interaction is observed and processed once upon arrival. We present an unbiased estimator of the link strength, and extend our formulation to unbiased estimators of general subgraphs in temporal networks. We also introduce the notion of \emph{link-decay} network sampling, in which the strength of a sampled link is allowed to decay exponentially after the most recent update (\ie, recent interaction). We show unbiased estimators of link strength and general subgraphs under the link-decay model. 

\smallskip
\noindent\textbf{Summary of Contributions}:
This work makes the following key contributions:
\begin{itemize}
    \item We propose a general temporal network sampling framework for unbiased estimation of temporal network statistics. We develop online, single-pass, memory-efficient sampling algorithms and unbiased estimators.
    
    \item We propose a temporally decaying sampling algorithm with unbiased estimators for studying networks that evolve in continuous time, where the strength of links is a function of time, and the motif patterns and temporal statistics are temporally weighted accordingly.
    This temporal decay model is more useful for real-world applications such as prediction and forecasting in temporal networks.
    
    \item The proposed algorithms enable fast, accurate, and memory-efficient statistical estimation of temporal network patterns and statistics.
    
    \item Experiments on a wide variety of temporal networks demonstrate the effectiveness of the framework.
\end{itemize}

\begin{table}[t!]
\centering
\caption{Summary of notation. 
}
\renewcommand{\arraystretch}{1.15} 
\scalebox{1.0}{
\centering 
\small
\setlength{\tabcolsep}{20pt} 
\label{table:notation}
\hspace*{-2.5mm}
\begin{tabularx}{0.6\linewidth}{@{}r X@{}} 
\toprule
$G$ & Temporal network \\ 
$E$ & Set of interaction events \\
$K$ & Set of Unique Edges (links) \\
$E_t$ & Set of interactions $\{e_s: s\le t\}$ \\
$K_t$ & Set of unique interactions in $E_t$ arriving by time $t$ \\
$V_t$ & Set of vertices that appeared in $K_t$ \\
$\tilde G$ & Graph induced by unique edges \\
$N, M$ & number of nodes $N = |V|$ and edges $M = |K|$ in $\tilde G$\\
$C_{e,t}$ & Multiplicity (weight) of edge $e$ at time $t$ \\
$\hat{C}_{e,t}$ & Estimated multiplicity of edge $e$ at time $t$ \\
$\mathbf{C}_t$ & time-dependent adjacency matrix of link strength at time $t$ \\
$\hat{K}$ & Reservoir of sampled edges \\
$m$ & Number of sampled edges (Sample Size), $m = |\hat{K}|$ \\
$\delta$ & Link decay rate \\
$\phi$ & Initial weight\\
${C}_\mathcal{M}$ & Weighted count of motif pattern $\mathcal{M}$ \\
$\hat{C}_\mathcal{M}$ & Estimated weighted count of motif pattern $\mathcal{M}$ \\
$V(e)$ & Unbiased estimator of variance of edge $e$\\
$w(e)$ & Sampling weight of edge $e$ \\
$r(e)$ & Rank of edge $e$ in the sample\\
\bottomrule
\end{tabularx}
}
\end{table}

\section{Online Sampling Framework}\label{sec:framework}

Here, we introduce our proposed online importance sampling framework that extracts continuous-time dynamic network samples from temporal networks. See Table~\ref{table:notation} for a summary of notations.

\begin{figure*}
\centering
\includegraphics[width=0.7\linewidth]{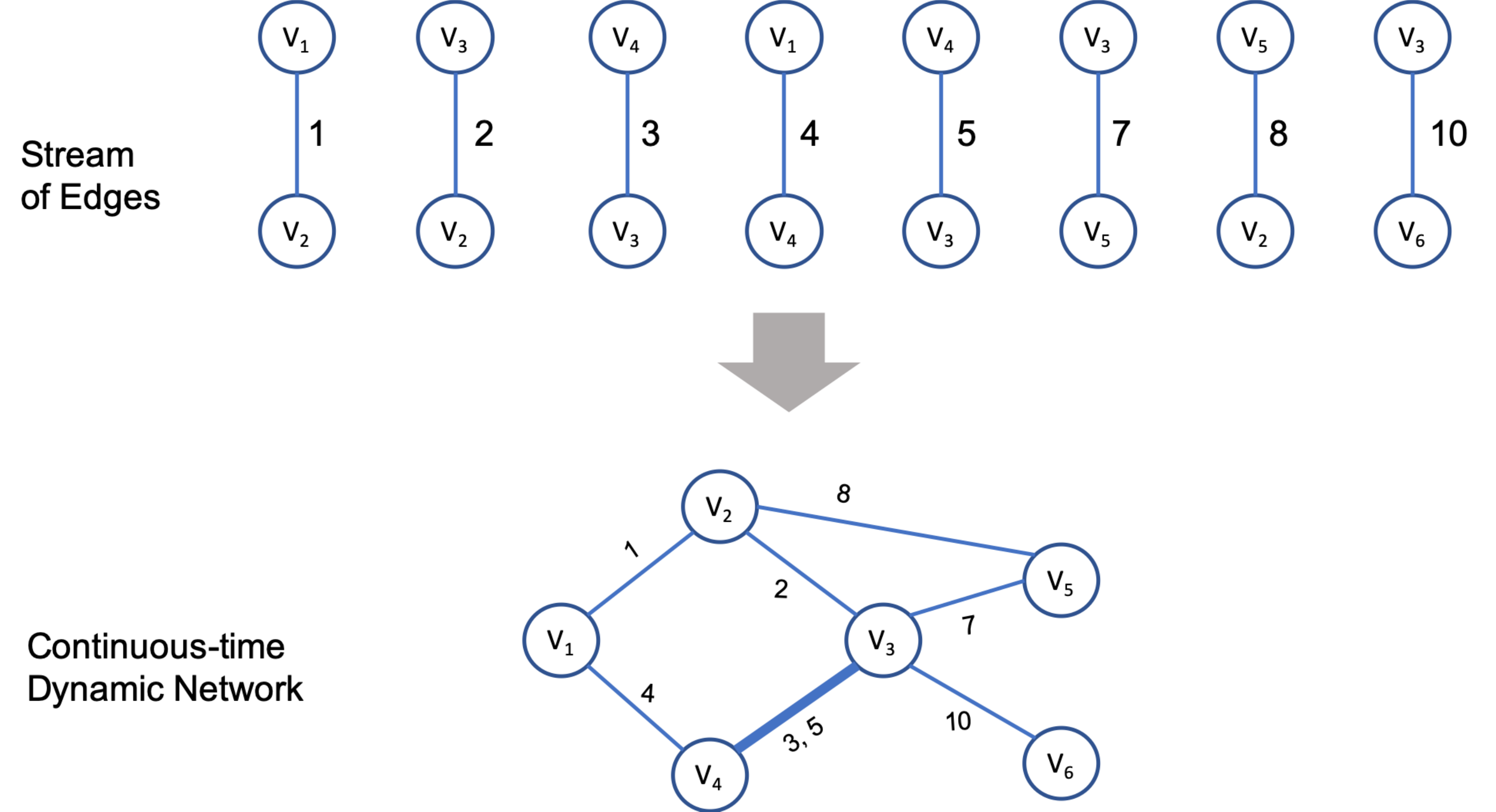}
\caption{
An illustrative example of streaming temporal networks.
}
\label{fig:ct_net}
\end{figure*}

\parab{Overview.} We propose an online sampling framework for temporal streaming networks which seeks to construct continuous-time, fixed-size, dynamic sampled network that can capture the evolution of the full network as it evolves in time as a stream of edges. Our framework assumes an input temporal network represented a stream of interactions links at certain times, and each interaction can be observed and/or processed only once. Sampling algorithms are allowed to store only $m$ sampled edges, and can process the stream in a single-pass. If any two vertices interact at time $t = \tau$, their edge strength increases by $1$. Figure~\ref{fig:ct_net} shows an illustrative example of how a continuous-time dynamic network can be formed from a stream of edges, where the edge strength is a function of the interactions among vertices over time. 

\subsection{Notation \& Problem Definition}\label{sec:prob_def}

\parab{Edges, Interactions, and Streaming Temporal Networks.} Our framework seeks to construct a continuous-time \emph{sampled} network that can capture the characteristics and serve as a proxy of an input temporal network as it evolves continuously in time. In this paper, we draw an important distinction between \emph{interactions} and \emph{edges}. An interaction (contact) between two entities is an event that occurred at a certain point in time (\eg, an email, text message, physical contact). On the other hand, an edge between two entities represents the link or the relationship between them, and the weight of this edge represents the strength of the relationship (\eg, strength of friendship in social network~\cite{xiang2010modeling}). We use $G$ to denote an input temporal network, where a set of vertices $V$ (\eg, users or entities) are interacting at certain times. Let $(i,j,t) \in E$ denote the interaction event that takes place at time $t$, where $i,j \in V$, $E$ is the set of interactions, $E_t$ is the set of interactions up to time $t$, and $K$ is the set of \emph{unique} edges ($e = (i,j) \in K$) in the temporal network $G$. We assume these interactions are instantaneous (\ie, the duration of the interaction is negligible), \eg, email, tweet, text message, etc. Let $C_e$ denote the multiplicity (weight) of an edge $e = (i,j)$, with $C_{e,t}$ being the multiplicity of the edge at time $t$, \ie, the number of times the edge appears in interactions up to time $t$. Finally, we define a streaming temporal network $G$ as a \emph{stream of interactions} $e_1,\dots, e_t,\dots,e_T$, with $e_t = (i,j,t)$ is the interaction between $i,j \in V$ at time $t$.

We note that the term unique edges refers to the set of relationships that exist among the vertices. On the other hand, the term interaction refers to an event that occurred at some point in time between two vertices. As such, interactions can happen more than once between two vertices, while unique edges represent the existing relationship between two vertices. For example, in Figure~\ref{fig:ct_net}, the edge $e = (v_3,v_4)$ has two interactions that happen at times $t= 3, 5$ respectively. Hence, the strength of $e = (v_3,v_4)$ is higher compared to other edges. 

\parab{Continuous-time Dynamic Network Samples.} Consider a set of $N = |V|$ interacting vertices, with their interactions represented as a streaming temporal network $G$, \ie, $e_1,\dots, e_t,\dots,e_T$. Let $\mathbf{C}_t$ be the time-dependent adjacency matrix, whose entries $C_{ij,t} \geq 0$ represent the relationship strength between vertices $i,j \in V$ at time $t$. The relationship strength is a function of the edge multiplicity and time. Our framework seeks to construct, maintain, and adapt a continuous-time dynamic sampled network, represented by the matrix $\hat{\mathbf{C}}_t$ that serves as unbiased estimator of $\mathbf{C}_t$ at any time point $t$, where the expected number of non-zero entries in $\hat{\mathbf{C}}_t$ is at most $m$, and $m$ is the sample size (\ie, maximum number of sampled edges). Our framework makes the following assumptions: 
\begin{itemize}
    \item We assume an input temporal network represented a stream of interactions at certain times, and each interaction can be processed and observed only once.
    \item Any algorithm can only store $m$ sampled edges, and is allowed a single-pass over the stream.
    \item If two vertices interact at time $t = \tau$, their edge strength increases by $1$. 
\end{itemize}

\subsection{Link-Decay Network Sampling}\label{sec:decay}

Here, we introduce a novel online sampling framework that seeks to construct and maintain a \emph{sampled} temporal network in which the strength of a link (\ie, relationship between two friends) can evolve continuously in time. Since the sampled network serves as a proxy of the full temporal network, the sampled network is expected to capture both the structural and temporal characteristics of the full temporal network.

\parab{Temporal Link-Decay.} Assume an input stream of interactions, where interactions are instantaneous (\eg, email, text message, and so on). For any pair of vertices $i,j \in V$, with a set of interaction times $\tau^{(1)}, \tau^{(2)}, \dots, \tau^{(T)}$, where $0 < \tau^{(1)} < \dots <\tau^{(T)}$, and their first interaction time is $\tau^{(1)} > 0$. Our goal is to estimate the strength of the link $e = (i,j)$ as a function of time, in which the link strength may increase or decrease based on the frequency and timings of the interactions. 
Consider two models of constructing an adaptive sampled network represented as a time-dependent adjacency matrix $\mathbf{C}_t$, whose entries represent the link strength $C_{ij,t}$. 

The first model is the \emph{no-decay} model, in which the link strength does not decrease over time, \ie, $C_{ij,\tau^{(2)}} = C_{ij,\tau^{(1)}} + 1$. Thus, $C_{ij,t}$ is the multiplicity or a function of the frequency of an edge, and we provide an unbiased estimator for this in Theorem~\ref{thm:nonu:C}. However, the no decay model assumes the interactions are fixed once happened, taking only the frequency of interactions as the primary factor in modeling link strength, which could be particularly useful for certain applications, such as proximity interactions (\eg, link strength for people attending a conference).  

The second model is the link-decay model, in which the strength of the link decays exponentially after the most recent interaction, to capture the temporal evolution of the relationship between $i$ and $j$ at any time $t$. Let the initial condition of the strength of link $(i,j)$ be $C_{ij,t_0} = 0$. Then, $C_{ij,t} = \sum_{s = 0}^{T} \theta(t - \tau^{(s)})\;{\rm e}^{-(t - \tau^{(s)})/\delta}$, where $\theta(t)$ is the unit step function, and the decay factor $\delta > 0$. We formulate the link strength as a stream of events (\eg, signals or pulses), that can be adapted incrementally in an online fashion, so the strength of link $e = (i,j)$ at time $t$ follows the equation,
\vspace{-1mm}
\begin{equation}
    C_{e,t} = C_{e,t-1} * {\rm e}^{-1/\delta}
\end{equation}
\noindent
And if a new interaction occurred at time $t$, the link strength follows, 
\vspace{-2mm}
\begin{equation}
    C_{e,t} = C_{e,t-1} * {\rm e}^{-1/\delta} + 1
\end{equation}

Our approach discounts the contributions of interactions to the time-dependent link strength as a function of the interaction age, while adapting the sampling weight of the link to its non-discounted multiplicity. This allows us to preferentially retain the relatively small proportion of highly active links, while the capability to temporally weight motif and subgraphs resides in the estimator. This is distinct from previous approaches for temporal sampling in which the retention sampling probability for single items were, \eg,  exponentially discounted according to age, without regard to item frequencies as a criterion for retention; see \cite{10.1109/ICDE.2009.65}.
    
We formulate an unbiased estimator for the link-decayed strength as a function of the link multiplicity in Section ~\ref{sec:topology-adaptive-est} (see Theorem~\ref{thm:decay}). All the proposed estimators can be computed and updated efficiently in a single-pass streaming fashion using Algorithm~\ref{alg-TNS}. In addition to exponential decay, the unbiased estimators generalize and can be easily extended to other decay functions, such as polynomial decay. Note that in this paper, we use the term link-decay to refer to exponential link-decay. 

Link decay has major advantages in network modeling that we discuss next. First, it allows us to utilize both the frequency and timings of interactions in modeling link strength. Second, it is more realistic, allowing us to avoid any potential bias that may result from partitioning interactions into time windows. Link decay is also flexible, by tuning the decay factor $\delta$, we can determine the degree at which the strength of the link ages (\ie, the half-life of a link $t_{1/2} = \delta \ln 2$). We also note that the link decay model and the unbiased estimator in Theorem~\ref{thm:decay} can generalize to allow more flexibility, by tuning the decay parameter on the network-level, the node-level, or the link-level, to allow different temporal scales at different levels of granularity.   

\parab{Temporally Weighted Motifs.} We showcase our formulation of estimated link strength by estimating the counts of motif frequencies in continuous time. We introduce the notion of \emph{temporally weighted motifs} in Definition~\ref{def:temp_motif}. Temporally weighted motifs are more meaningful and useful for practical applications especially related to prediction and forecasting where links and motifs that occur more recently as well as more frequently are more important than those occurring in the distant past. 

\begin{Definition}{\sc (Temporally Weighted Motif)}\label{def:temp_motif}
A temporally weighted network motif $\mathcal{M}$ is a small induced subgraph pattern with $n$ vertices, and $m$ edges, such that $C_\mathcal{M}$ is the time-dependent frequency of $\mathcal{M}$ and is subject to temporal decay, and $C_{\mathcal{M},t} = \sum_{h \in H_t} \prod_{e \in h} C_{e,t}$, where $H_t$ is the set of observed subgraphs isomorphic to $\mathcal{M}$ at time $t$, and $C_{e,t}$ is the link strength. 
\end{Definition}

In general, motifs represent small subgraph patterns and the motif counts were shown to reveal fundamental characteristics and design principles of complex networked systems~\cite{milo2002network,ahmed2015efficient,benson2016higher,ahmed2017graphlet}, as well as improve the accuracy of machine learning models~\cite{ahmed2018learning,HONE}. While prior work focused on aggregating interactions in time windows and analyze the aggregated graph snapshots~\cite{rossi2012time,sharan2008temporal,Taheri-www2019}, others have focused on aggregating motifs in $\bigtriangleup t$ time bins, and defined motif duration~\cite{liu2019sampling}. These approaches rely on judicious partitioning of interactions in time bins, and would certainly suffer from the limitations discussed earlier in Section~\ref{sec:intro}. Time partitioning may obfuscate or dilute temporal and structural information, leading to biased results. Here, we define instead a temporal weight or strength for any observed motif, which is a function of the strength of the links participating in the motif itself. Similar to the link strength, the motif weight is subject to time decay. This formulation can also generalize to models of higher-order link decay (\ie, decaying hyperedges in hypergraphs), we defer this to future work. The definition in~\ref{def:temp_motif} can be computed incrementally in an online fashion, and subject to approximation via sampling and unbiased estimators. We establish our sampling methodology in Algorithm~\ref{alg-TNS} (see line~\ref{func:subgraph}) and unbiased estimators of subgraphs in Section~\ref{sec:topology-adaptive-est} (see Theorem~\ref{thm:nonu:C}).

\section{Proposed Algorithm}\label{sec:alg-desc}

{
\algrenewcommand{\alglinenumber}[1]{\fontsize{8}{9}\selectfont#1}
\algnotext{endpar}
\begin{figure}[t!]
\begin{center}
\begin{algorithm}[H]
\caption{Online Temporal Network Sampling (Online-TNS)}
  \label{alg-TNS}
\begin{spacing}{1.0}
\fontsize{8}{9}\selectfont
\algrenewcommand{\alglinenumber}[1]{\fontsize{6.5}{7}\selectfont#1}
\begin{flushleft}
        \textbf{INPUT:} Sample size $m$, Motif pattern $\mathcal{M}$, initial weight $\phi$ \\
        \textbf{OUTPUT:} Estimated network $\mathbf{C}_t$, Estimated motif count $\hat{C}_\mathcal{M}$ 
\end{flushleft}
\begin{algorithmic}[1]
 \Procedure{\textsc{$\text{Online-TNS}$}}{$m$} 
 \State $\hat{K}=\emptyset$;
  $z^*\!= 0$; $\hat{C}_\mathcal{M} = 0$ \textcolor{theblue}{\Comment{Initialize edge sample \& threshold}}
  \While{(new interaction $e_t = (i,j,t)$)} \label{line:new_inter} 
        \State $e = (i,j)$
        \State \Call{\textsc{$\text{Subgraph-Estimation}$}}{$e$} \textcolor{theblue}{\Comment{Update Estimated Motif}}
        \If{($e \in \hat{K}$)} \textcolor{theblue}{\Comment{Edge exists in $\hat{K}$}} 
        \State \Call{\textsc{$\text{update-edge-strength}$}}{$e$} \textcolor{theblue}{\Comment{Update
          Edge Strength}} \label{line:update_strength}
       \State $\hat{C}(e)= \hat{C}(e)+1$ \textcolor{theblue}{\Comment{Increment edge multiplicity}} \label{line:update_edge_st}
        \State $w(e) = w(e) + 1$ \textcolor{theblue}{\Comment{Adapt importance weight}} \label{line:weight_adapt}
        \State $r(e)= w(e)/u(e)$ \textcolor{theblue}{\Comment{Adapt edge rank}} \label{line:rank_update}
        \State $\tau(e) = t$ \textcolor{theblue}{\Comment{Last Interaction Time}} \label{line:update_edge_en}
     \Else 
     \State \textcolor{theblue}{//Initialize parameters for new edge}
     \State $p(e)= 1;\ \hat{C}(e)= 1$; $V(e)= 0$; 
      \State $u(e)= \text{Uniform}\;(0,1]$ \textcolor{theblue}{\Comment{Initialize Uniform r.v.}} 
      \State $w(e)= \phi$
      \textcolor{theblue}{\Comment{Initialize edge weight}}
      \State $r(e)= w(e)/u(e)$ \textcolor{theblue}{\Comment{Compute edge rank}} \label{line:rank_new}
      \State $\tau(e) = t$
      \State $\hat{K} = \hat{K} \cup \{e\}$  \textcolor{theblue}{\Comment{Provisionally include $e$ in sample}} \label{line:addto_sample}
      \If{($|\hat{K}| > m$)} 
      \State $e^* = \argmin_{e'\in \hat{K}} r(e')$ \textcolor{theblue}{\Comment{Find edge with min rank}} \label{line:findmin-edge}
        \State $z^*=\max\{z^*, r(e^*)\}$ \textcolor{theblue}{\Comment{Update threshold}} \label{line:thresh}
        \State remove $e^*$ from $\hat{K}$ \label{line:removemin-edge}
        \State delete $\{w(e^*),u(e^*),p(e^*),\hat{C}(e^*),V(e^*)\}$
        \EndIf
        \EndIf
        \EndWhile
        \EndProcedure
\algrule
\Procedure{\textsc{$\text{update-edge-strength}$}}{$\tilde e$} \label{func:edge_strength}
    \State \textcolor{theblue}{// Function to estimate edge strength (No-decay)} 
    \If{($z^*>0$)}
    \State $q = \min\{1, w(\tilde e)/(z^* p(\tilde e))\}$
    \State $\hat{C}(\tilde e) = \hat{C}(\tilde e)/q$ \textcolor{theblue}{\Comment{Estimate edge strength}}
    \State $V(\tilde e)= V(\tilde e)/q + (1-q)*\hat{C}(\tilde e)^2$
    \State $p(\tilde e)= p(\tilde e)*q$
    \EndIf
\EndProcedure
\algrule
\Procedure{\textsc{$\text{update-edge-decay}$}}{$\tilde e$} \label{func:edge_strength_decay}
\State \textcolor{theblue}{// Function to estimate edge strength (Link-decay)}
    \If{($z^*>0$)}
    \State $q = \min\{1, w(\tilde e)/(z^* p(\tilde e))\}$
    \State $\hat{C}(\tilde e) = {\rm e}^{-\delta(t - \tau(\tilde e))} * \hat{C}(\tilde e)/q$ \textcolor{theblue}{\Comment{Estimate link strength}}
    \State $V(\tilde e)= V(\tilde e)/q + (1-q)*\hat{C}(\tilde e)^2$
    \State $p(\tilde e)= p(\tilde e)*q$
    \EndIf
\EndProcedure
\algrule
\Procedure{\textsc{$\text{Subgraph-Estimation}$}}{$\tilde e$} \label{func:subgraph}
\State \textcolor{theblue}{//Set of Subgraphs isomorphic to $\mathcal{M}$ and completed by $\tilde e$}
\State $H = \{h\subset \hat{K}\cup\{\tilde e\}: h \owns {\tilde e}, h \cong \mathcal{M}\}$ \label{line:tris}
\For{$h\in H$} 
    \For{$j\in h\setminus{\tilde e}$}  
      \State \Call{\textsc{$\text{update-edge-strength}$}}{$j$} \textcolor{theblue}{\Comment{Update
        other edges}}
    \EndFor
    \State \textcolor{theblue}{//Increment estimated count of motif $\mathcal{M}$}
    \State $\hat{C}_\mathcal{M} = \hat{C}_\mathcal{M} + \prod_{j \in h\setminus\{\tilde e\}}
        \hat{C}(j)$ \label{line7} 
\EndFor  
\EndProcedure
\end{algorithmic}
\end{spacing}
\end{algorithm}
\end{center}
\end{figure}
}

In this paper, we propose an online sampling framework for temporal streaming networks which seeks to construct \emph{continuous-time}, fixed-size, dynamic sampled network that can capture the evolution of the full network as it evolves in time. Our proposed framework establishes a number of properties that we discuss next. We formally state our algorithm, called \textsc{Online-TNS}, in Algorithm~\ref{alg-TNS}. 

\parab{Setup and Key Intuition.} The general intuition of the proposed algorithm in Algorithm~\ref{alg-TNS}, is to maintain a dynamic rank-based reservoir sample $\hat{K}$ of a fixed-size $m$~\cite{vitter1985random,duffield2007priority,ahmed2017sampling}, from a temporal network represented as stream of interactions, where edges can appear repeatedly. And, $m = |\hat{K}|$ is the maximum possible number of sampled edges. When a new interaction $e_t = (i,j,t)$ arrives (line~\ref{line:new_inter}), if the edge $e = (i,j)$ has been sampled before (\ie, $e =(i,j) \in \hat{K}$), then we only need to update the edge sampling parameters (in lines~\ref{line:update_edge_st}--\ref{line:update_edge_en}) and the edge strength (line~\ref{line:update_strength}). However, if the edge is new (\ie, $e =(i,j) \notin \hat{K}$), then the new edge is added provisionally to the sample (line~\ref{line:addto_sample}), and one of the $m+1$ edges in $\hat{K}$ gets discarded (lines~\ref{line:findmin-edge}~and~\ref{line:removemin-edge}).   

\parab{Importance sampling weights and rank variables.} Algorithm~\ref{alg-TNS} preferentially selects edges to include in the sample based on their importance weight relative to the variable of interest
(\eg, relationship strength, topological features), then adapts their weights to allow edges to gain importance during stream processing. To achieve this, each arriving edge $e$ is assigned an initial weight $w(e)$ on arrival and an iid uniform $U(0,1]$ random variable $u(e)$. Then, Algorithm~\ref{alg-TNS} computes and continuously updates a rank variable for each sampled edge $r(e) = w(e)/u(e)$ (see line~\ref{line:rank_new} and line~\ref{line:rank_update}). This rank variable quantifies the importance/priority of the edge to remain in the sample. To keep a fixed sample size, the $m+1$ edge with minimum rank is always discarded (lines~\ref{line:findmin-edge}~and~\ref{line:removemin-edge}). 
The algorithm also maintains a sample threshold $z^*$ which is the maximum discarded rank (line~\ref{line:thresh}). Thus, the inclusion probability of an edge $e$ in the sample is: $\mathbb{P}(e \in \hat K) = \mathbb{P}(r(e) > z^*) = \mathbb{P}(u(e) < w(e)/z^*) = \min\{1,w(e)/z^*\}$. 
Our mathematical formulation in Section~\ref{sec:topology-adaptive-est} allows the edge sampling weight to increase when more interactions are observed (line~\ref{line:weight_adapt}). Thus, edges can gain more importance or rank that reflects the relationship strength as it evolves continuously in time. This setup will support network models that focus on capturing the relationship strength in temporal networks~\cite{miritello2011dynamical,holme2012temporal}.    

\parab{Unbiased estimation of link strength.} We use a procedure called \textsc{update-edge-strength} (line~\ref{func:edge_strength} of Algorithm~\ref{alg-TNS}) to dynamically maintain an unbiased estimate (see Theorem~\ref{thm:nonu:C}) of the edge strength as it evolves continuously in time. The procedure in line~\ref{func:edge_strength} of Algorithm~\ref{alg-TNS} also maintains an unbiased estimate of the variance of the edge strength following Theorem~\ref{thm:var-c}. Note the strength of an edge $e$ is a function of the edge multiplicity $C_{e}$ (the number of interactions $e_t$ where $e_t = e$). If a link-decaying model is required, the procedure called \textsc{update-edge-decay} can be used instead of \textsc{update-edge-strength} to estimate the link-decayed strength (see line~\ref{func:edge_strength_decay} of Algorithm~\ref{alg-TNS}). We prove that our estimated link-decaying weight is unbiased in Theorem~\ref{thm:decay}.      

\parab{Unbiased estimation of subgraph counts.} Given a motif pattern $\mathcal{M}$ of interest (\eg, triangles, or small cliques), the procedure called \textsc{Subgraph-Estimation} in line~\ref{func:subgraph} of Algorithm~\ref{alg-TNS} is used to update an unbiased estimate of the count of all occurrences of the motif $\mathcal{M}$ at any time $t$. Theorem~\ref{thm:nonu:C} is used to establish the unbiased estimator of the count of general subgraphs. The unbiased estimator of subgraph counts also applies in the case of link decay, and gives rise to temporally decayed (weighted) motifs. 

\parab{Computational Efficiency and complexity.} All the algorithms and estimators can run in a \emph{single-pass} on the stream of interactions, where each interaction can be observed and processed once (see Alg~\ref{alg-TNS}). The main reservoir sample is implemented as a heap data structure (min-heap) with a hash table to allow efficient updates. The estimator of edge strength can be updated in constant time $O(1)$. Also, retrieving the edge with minimum rank can be done in constant time $O(1)$. Any updates to the sampling weights and rank variables can be executed in a worst-case time of $O(\log(m))$ (\ie, since it will trigger a bubble-up or bubble-down heap operations). For any incoming edge $e = (i,j)$, subgraph estimators can be efficiently computed if a hash table or bloom filter is used for storing and looping over the sampled neighborhood of the sampled vertex with minimum degree and querying the hash table of the other sampled vertex. For example, if we seek to estimate triangle counts, then line~\ref{line:tris} in Algorithm~\ref{alg-TNS} can be implemented in $O(min\{\text{deg}(i),\text{deg}(j)\})$.

\section{Adaptive Unbiased Estimation}
\label{sec:topology-adaptive-est}

In this section, we theoretically show and discuss our formulation of unbiased estimators for temporal networks, that we use in Algorithm~\ref{alg-TNS}.

\parab{Edge Multiplicities.}
We consider a temporal network $G=(V,E)$ comprising interactions $E$ between vertex pairs of $V$. Each interaction can be viewed as a representative of an edge set $K$ comprising the unique elements of $E$. We will write $\tilde G=(V,K)$ as the graph induced by $K$. Thus the stream of interactions can also be regarded as a stream $\{e_t: t\in[|E|]\}$ of non-unique edges from $K$.
Let $K_t$ denote the unique edges in $\{e_s: s\le t\}$ and have arrived by time $t$. Also, let $\tilde G_t=(V_t,K_t)$ be the induced graph, where $V_t$ is the subset of vertices that appeared in $K_t$. The multiplicity $C_{e,t}$ of an edge $e\in K_t$ is the number of times it occurs in $E_t=\{e_s: s\le t\}$, i.e., $C_{e,t}=|\{s\le t: e_s = e\}|$. 
The multiplicity $C_{J,t}$ of $J\subset K_t$ 
is the number of distinct ordered interaction subsets $\tilde J = \{e_{i_1},\ldots,e_{i_{|J|}}\}$ with $i_j\le t$, such that $\tilde J$ is a permutation of $J$. Hence $C_{J,t}=\prod_{e\in J}C_{e,t}$. Given a class $\cH$ of subgraphs of  $\tilde G$, we wish to estimate  for each $t$ the total multiplicity $H_t=\sum_{J\in \cH}C_{J,t}$ of subgraphs from $\cH$ that are present in the first $t$ arrivals. 

\parab{Sampling Edges and Estimating Edge Multiplicities.}
We record edge arrivals by the indicators $c_{e,t}=1$ if $e_t=e$ and zero otherwise, and hence $C_{e,t}=\sum_{t\ge 1}c_{e,t}$. 
$\hat K_t$ will denote the sample set of unique edges after arrival $t$ has been processed. We maintain an estimator $\hat C_{e,t}$ of $C_{e,t}$ for each $e\in \hat K_t$. Implicitly $\hat C_{e,t}=0$ if $e \notin \hat K_t$.

The algorithm proceeds as follows. If the arriving edge $e_t\notin \hat
K_{t-1}$ then $e_t$ is provisionally included in the sample, forming $\hat K'_t=\hat K_t\cup\{e_t\}$, and we set $\hat C_{e_t,t}=c_{e_t,t}= 1$. The new edge is assigned a random variable $u_{e_t}$ distributed IID in $(0,1]$. A weight $w_{i,t}$ is specified for each edge $i\in \hat K'_t$ as described below, from which the edge \textsl{time-dependent priority} at time $t$ is $r_{i,t}=w_{i,t}/u_i$. If $|\hat K'_t| > m$, the edge $d_t=\argmin_{i\in \hat K'_t} r_{i,t}$ of minimum priority is discarded, and the estimates $\hat C_{i,t}$ of the surviving edges $i \in \hat K_t = \hat K'_t\setminus\{d_t\}$ undergo inverse probability normalization through division by the conditional probability $q_{i,t}$ of retention in $\hat K_t$; see (Equation~\ref{eq:22}). If the arriving edge is already in the reservoir $e_t\in \hat K_{t-1}$ then we increment its multiplicity $\hat C_{e_t,t}=\hat C_{e_t,t-1} + 1$ and no sampling is needed, i.e., $\hat K_t = \hat K_{t-1}$.

\parab{Unbiased Estimation of Edge Multiplicities.} 
Let $\Omega$ denote the (random) set of times at which the sampling step takes place, i.e., such that the arriving edge $e_t$ is not currently in the reservoir $e_t \notin \hat K_{t-1}$ and $|\hat K_{t-1}| = m$. For $t\in\Omega'=\Omega\setminus\{\min\Omega\}$, let $\omega(t)=\max \{ [0,t)\cap\Omega \}$ denote the next most recent time at which the sampling step took place. For $t\in\Omega'$, the sample counts present in the reservoir accrue unit increments from arrivals $e_{\omega(t)+1},\ldots,e_{t-1}$ until the sampling step takes place at time $t$. For $t\in\Omega$,  an edge $i\in \hat K_{\omega(t)}$ is selected into $\hat K_t$ if and only if $r_{i,t}$ exceeds the smallest priority of all other elements of 
$\hat K'_t$, i.e., 
\be
r_{i,t}>z_{i,t}:=\min_{j\in \hat K'_t\setminus \{i\}}r_{j,t}
\ee 
Hence by recurrence, $i\in \hat K_t$ only if $u_i<\min_s\{w_{i,s}/z_{i,s}\}$ where $s$ takes values over 
\be
\{\alpha_i(t),\ldots,\omega(\omega(t)),\omega(t),t\}
\ee 
where $\alpha_i(t)$ is the most recent time at which edge $i$ was sampled into the reservoir.

This motivates the definition below where $p_{e,t}$ is the edge selection probability conditional on the thresholds $z_t$, and $q_{e,t}$ is the conditional probability for sampling for each increment of time. 
Let $t_e$ denote the time of first arrival of edge $e$. For $t\in\Omega$, define $p_{e,t}$ through the iteration

\be\label{eq:21}
p_{e,t} =
\left\{
\begin{array}{ll}
\min\{1,w_{e,t}/z_t\} & \mbox{ if } t=\min\Omega \\
\min\{p_{e,\omega(t)},w_{e,t}/z_t\} & \mbox{ otherwise }\\
\end{array}
\right.
\ee

where $z_t =\min_{e\in \hat K'_t} r_{e,t}$ for $t\in\Omega_t$ in the unrestricted minimum priority over edges in $\hat K'_t$. Note that $z_{i,t}=z_t$ if $i\in \hat K_t$. 
Then $\hat C_{e,t}$ is defined by the iteration $\hat C_{e,t}=0$ for $t<t_e$ and
\be\label{eq:22}
\hat C_{e,t}=\left(\hat C_{e,t-1} + c_{e,t}\right)\frac{I(u_i< w_{i,t}/z_{i,t})}{q_{e,t}}
\ee
where 
\be\label{eq:23}
q_{e,t}=
\left\{
\begin{array}{ll}
1&\mbox{ if } t\notin\Omega\\
p_{e,t}&\mbox{ if } t\in\Omega \mbox{ and } e=e_t\\
p_{e,t}/p_{e,\omega(t)}& \mbox{ otherwise}
\end{array}
\right.
\ee

For $J\subset V$, let $t_J=\min_{j\in J}t_j$, i.e., the earliest time at which any instance of an edge in $J$ has arrived. Let $J_t=\{j\in J: t_j\le t\}$, i.e., the edges in $J$ whose first instance has arrived by $t$. Note in our model these are deterministic. The proof of the following Theorem and others in this paper are detailed in Section~\ref{sec:proofs}. 

\begin{theorem}[\textbf{Unbiased Estimation}]
\label{thm:nonu:C} 
\hfill
\setlength{\itemsep}{0.5pt}
\setlength{\parskip}{2pt}
\begin{itemize}
\item[(i)] $\E[\hat C_{e,t}]=C_{e,t}$ for all $t\ge 0$.
\item[(ii)] For each $J\subset V$ and $t\ge t_J$ then $\prod_{e\in
    J_t}(\hat C_{e,t}-C_{e,t}): t\ge t_J\}$ has expectation $0$.
\item[(iii)] $\E[\prod_{e\in J}\hat C_{e,t}]=\prod_{e\in J} C_{e,t}$
  for $t\ge \max_{j\in J} t_j$.
\end{itemize}
\end{theorem}

\parab{Estimating Subgraph Multiplicities.}
Theorem~\ref{thm:nonu:C} tells us for a subgraph $J\subset K_t$,
that $\prod_{e\in J}\hat C_{e,t}$ is an unbiased estimator of the
multiplicity $\prod_{e\in J} C_{e,t}$ of subgraphs formed by distinct set of interactions isomorphic to $J$. Now let $h\in H_t$, the set of subgraphs of $\tilde G_t$ that are isomorphic to $\cM$ at time $t$. We partition the set of interactions in $E_t$ that represent $h$ according to the time of last arrival. Thus it is evident that
$C_{\cM,t}=\sum_{s\le t}C\up0_{\cM,s}$ where 
$C\up0_{\cM,s}=\sum_{h\in H\up 0_s}C_{h\setminus\{e_s\},s}$ where
$H\up0_s=\{h\in K_s: h\owns e_s: h\cong M\}$, 
meaning, for each interaction $e_s$, we consider subgraphs $h$ of $K_s$ congruent to $\cM$ and containing $e_s$, and compute the multiplicity of the $h$ with $e_s$ removed, i.e., not counting any isomorphic sets of interactions in which the  $e=e_s$ arrived previously, thus avoiding over-counting.
It follows by linearity that $\hat C_{\cM,t}=\sum_{s\le t}\hat C\up 0_{\cM,s-1}$ is an unbiased estimator of $C_{\cM,t}$ where
$\hat C\up 0_{\cM,s-1}=
\sum_{h\in H\up0_s}\hat C_{h\setminus\{e_s\},s-1}$.
Thus for each arrival $e_t$ we estimate $C_{J,t}$ just prior to sampling of $e_t$ by 
$\hat C_{J,t}= \prod_{j\in J\setminus \{e_t\}} \hat
C_{j,t-1}$. For each $J\subset H_t$ we increment a running total of $\hat M_t$ by this amount; see line ~\ref{line7} in Algorithm~\ref{alg-TNS}.

\parab{Edge Multiplicity Variance Estimation.} We now discuss the unbiased estimator of the variance $\var(\hat C_{e,t})$.

\begin{theorem}[\textbf{Unbiased Variance Estimator}]
\label{thm:var-c}
\hfill
\setlength{\itemsep}{0.5pt}
\setlength{\parskip}{2pt}

\noindent
Suppose $\hat V_{e,t-1}$ is an 
unbiased estimator of $\var(\hat C_{e,t-1})$ that can be computed from information on the first $t-1$ arrivals. Then
\be\label{eq:var-c}
\hat V_{e,t} = \hat C_{e,t}^2(1-q_{e,t}) + 
I(u_e < w_{e,t}/z_t)\hat V_{e,t-1}/q_{e,t}
\ee 
is a unbiased estimator of $\var(\hat C_{e,t})$ that can be computed from information on the first $t$ arrivals.
\end{theorem}

The computational condition expresses the property that $\hat V_{e,t}$ can be computed immediately when $e\in \hat K_t$. The relation (Equation~\ref{eq:var-c}) defines an iteration for estimating the variance $\var(\hat C_t)$ for any $t$ following a time $s\in\Omega$ at which edge $e$ was sampled into $\hat K_s$, such that $e$ remained in the reservoir at least until $t$.
The unbiased variance estimate $\hat V_{e,s}$ takes the value $1/p_{e,s} - 1$ at time $s$ of selection into the reservoir. In practice $\hat V_{e,t}$ only needs to be updated at $t\in\Omega$, i.e., when some edge is sampled into the reservoir, since $q_{e,t}=1$ when $t\notin\Omega$.

\parab{Estimation and Variance for Link-Decay Model.}

The link-delay model adapts (Sec.\ref{sec:decay}) through
\be
\hat C^\delta_{k,t}=\left( \hat C^\delta_{k,t-1}e^{-1/\delta} + c_{k,t}\right)\frac{I(u_k<w_{k,t}/z_{t})}{q_{k,t}}
\ee
which exponentially discounts the contribution from the previous time slot.

\begin{theorem}[\textbf{Unbiased Estimation with Link Decay}]
\label{thm:decay}
\hfill
\setlength{\itemsep}{0.5pt}
\setlength{\parskip}{2pt}
\begin{itemize}
    \item [(i)] $\hat C^\delta_{k,t}$ is an unbiased estimator of $C^\delta_{k,t}$
    \item[(ii)] Replacing $\hat C_{k,t}$ with $\hat C^\delta_{k,t}$ in the iteration yields an unbiased estimator $V^\delta_{k,t}$ of $\var(\hat C^\delta_{k,t})$ 
\end{itemize}
\end{theorem}

\renewcommand{\subsubsection}[1]{\medskip\noindent\textbf{#1.}}

\section{Experiments} 
\label{sec:exp}

\begin{table}[t!]
\centering
\caption{
Temporal network data~\cite{nr}.
Note $|K|$ is the number of static edges (not including multiplicities); 
$|E|$= number of temporal edges;
and $C_{\max}$= maximum edge weight. 
}
\label{table:network-stats}
\renewcommand{\arraystretch}{1.10} 
\small
\setlength{\tabcolsep}{4.7pt} 
\begin{tabularx}{0.6\linewidth}{@{}
r r rHr r r 
HHHHH H
@{}
}
\toprule
\textsc{Temporal Network} & $|V|$ & $|K|$ & & $|E|$ & \textbf{days} & $C_{\max}$ &
\\
\midrule
\textsf{sx-stackoverflow} & 2.6M & 28.1M & 48M & 47.9M & 2774.3 & 1.04k \\
\textsf{ia-facebook-wall-wosn} & 46k & 183k & 856k & 877k & 1591.0 & 1.3k \\
\textsf{wiki-talk} & 1.1M & 2.8M & 6.1M & 7.8M & 2320.4 & 1.6k \\
\textsf{bitcoin} & 24.5M & 86.1M & 122.9M & 129.2M & 1811.7 & 72.6k \\
\textsf{CollegeMsg} & 1.9k & 14k & 60k & 60k & 193.7 & 184 \\
\textsf{ia-retweet-pol} & 18k & 48k & 61k & 61k & 48.8 & 79 \\
\textsf{ia-prosper-loans} & 89k & 3.3M & 3.4M & 3.4M & 2142.0 & 15 \\
\textsf{comm-linux-reply} & 26k & 155k & 984k & 1.0M & 2921.6 & 1.9k \\
\textsf{email-dnc} & 1.9k & 4.4k & 37k & 39k & 982.3 & 634 \\
\textsf{ia-enron-email} & 87k & 297k & 1.1M & 1.1M & 16217.5 & 1.4k \\
\textsf{ia-contacts-dublin} & 11k & 45k & 416k & 416k & 80.4 & 345 \\
\textsf{fb-forum} & 899 & 7.0k & 34k & 34k & 164.5 & 171 \\
\textsf{ia-contacts-hyper09} & 113 & 2.2k & 21k & 21k & 2.5 & 1.3k \\
\textsf{SFHH-conf-sensor} & 403 & 9.6k & 70k & 70k & 1.3 & 1.2k \\
\textsf{sx-superuser} & 192k & 715k & 1.1M & 1.4M & 2773.3 & 139 \\
\textsf{sx-askubuntu} & 157k & 456k & 727k & 964k & 2613.8 & 215 \\
\textsf{sx-mathoverflow} & 25k & 188k & 390k & 507k & 2350.3 & 325 \\
\bottomrule
\end{tabularx}
\end{table}

We perform extensive experiments on a wide variety of temporal networks from different domains. The temporal network data used in our experiments is shown in Table~\ref{table:network-stats}. We discuss baseline comparisons in Section~\ref{sec:baseline}, and perform a detailed ablation study that shows the contributions of the different components and design choices of Algorithm~\ref{alg-TNS} in Section~\ref{sec:ablation}.
The experiments systematically investigate the effectiveness of the framework for estimating \emph{temporal link strength}  (Section~\ref{sec:exp-temporal-link-strength}), \emph{temporally weighted motifs}  
using the decay model  (Section~\ref{sec:exp-temporal-motifs}), 
and \emph{temporal network statistics} (Section~\ref{sec:exp-temporal-network-statistics}). 
We use sample fractions $p = \{0.10,0.20,0.30,0.40,0.50\}$ and all experiments are the average of five different runs, similar to the setup in prior work~\cite{lim2018memory}. 

\begin{table}[h!]
\centering
\caption{
Baseline Comparison:  Relative spectral norm (\ie, $\norm{\mathbf{C} - \hat{\mathbf{C}}}_2/\norm{\mathbf{C}}_2$) for sampling fraction $p=0.1$, comparison between Online-TNS (Alg~\ref{alg-TNS}), Triest sampling~\cite{stefani2017triest}, Reservoir sampling~\cite{vitter1985random}, and MultiWMascot sampling~\cite{lim2018memory}.
}
\label{table:baseline-spectral-error}
\small 
\setlength{\tabcolsep}{15.0pt} 
\begin{tabularx}{0.8\linewidth}{@{}
r cccc
HHHHH H
@{}
}
\toprule

\textsc{Temporal Network} & \textbf{Online-TNS} & \textbf{Triest} & \textbf{Reservoir} & \textbf{MultiWMascot}  \\
\midrule
\textsf{CollegeMsg} & 0.0558 & 0.2304 & 0.2212 & 0.1990  & \\
\textsf{ia-retweet-pol}  & 0.1800 & 0.4103 & 0.4091 & 0.3973 & \\ 
\textsf{ia-contacts-dublin} & 0.0215 & 0.1926 & 0.1937 & 0.1855 & \\
\textsf{wiki-talk} & 0.1020 & 0.2554 & 0.2347 & 0.2343 & \\
\textsf{fb-forum} & 0.0390 & 0.1900 & 0.1912 & 0.2052 & \\
\textsf{sx-mathoverflow} & 0.0668 & 0.1767 & 0.1764 & 0.1584 & \\
\textsf{sx-stackoverflow} & 0.0992 & 0.2114 & 0.2036 & 0.2045 & \\

\bottomrule
\end{tabularx}
\end{table}

\renewcommand{\figsz}{0.22}
\begin{figure*}[h!]
\centering
\subfigure
{\includegraphics[width=\figsz\linewidth]{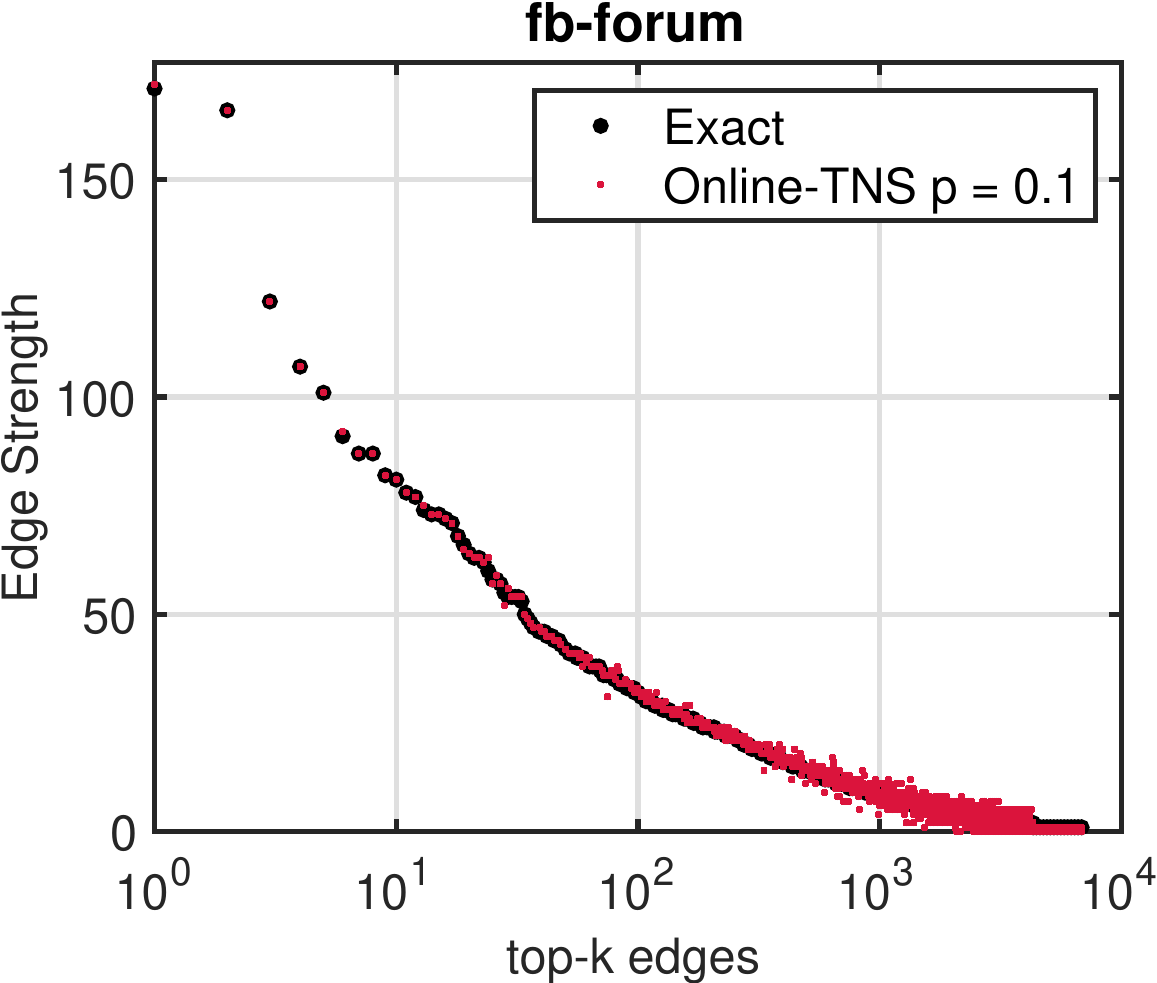}}
\hspace{3mm}
\subfigure
{\includegraphics[width=\figsz\linewidth]{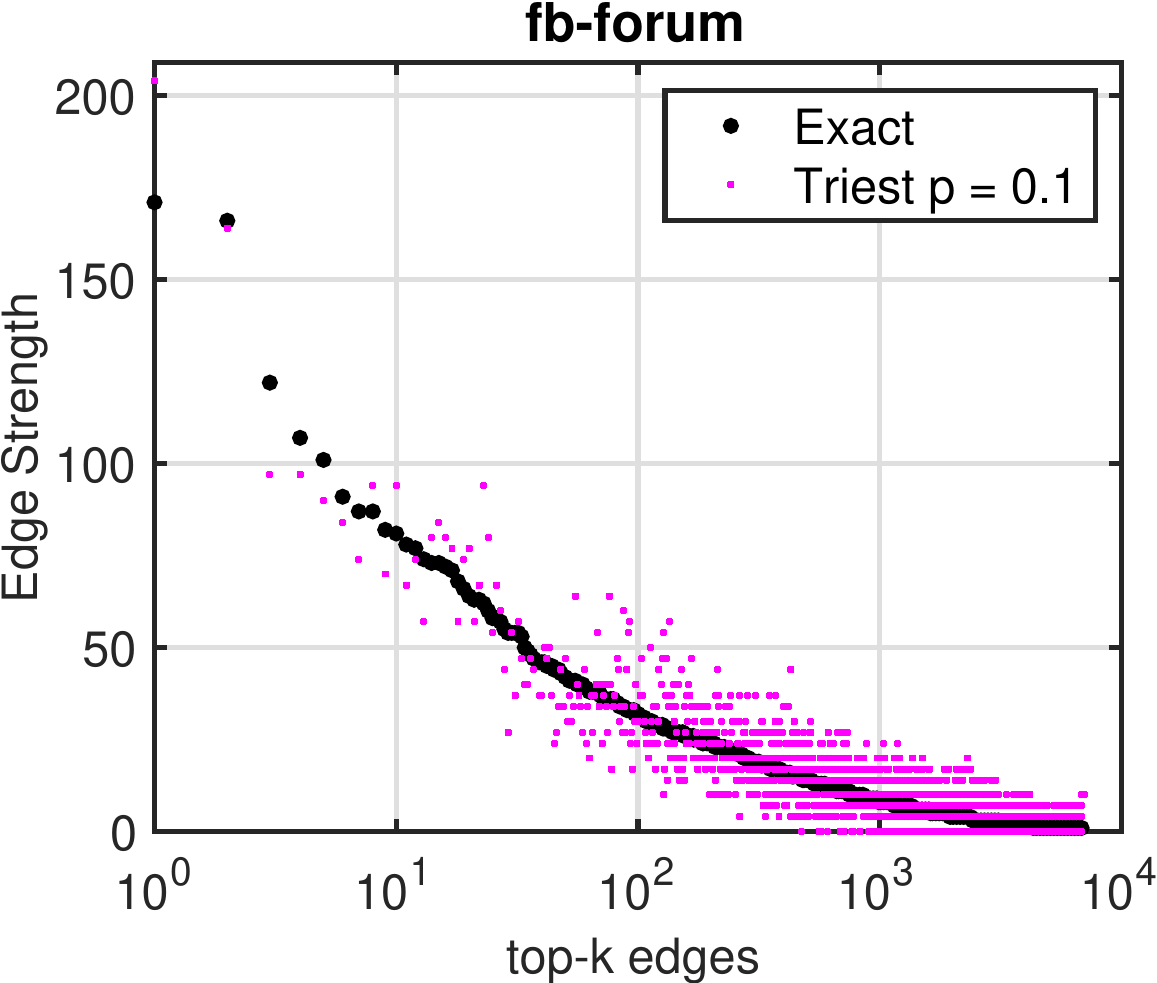}}
\hspace{3mm}
\subfigure
{\includegraphics[width=\figsz\linewidth]{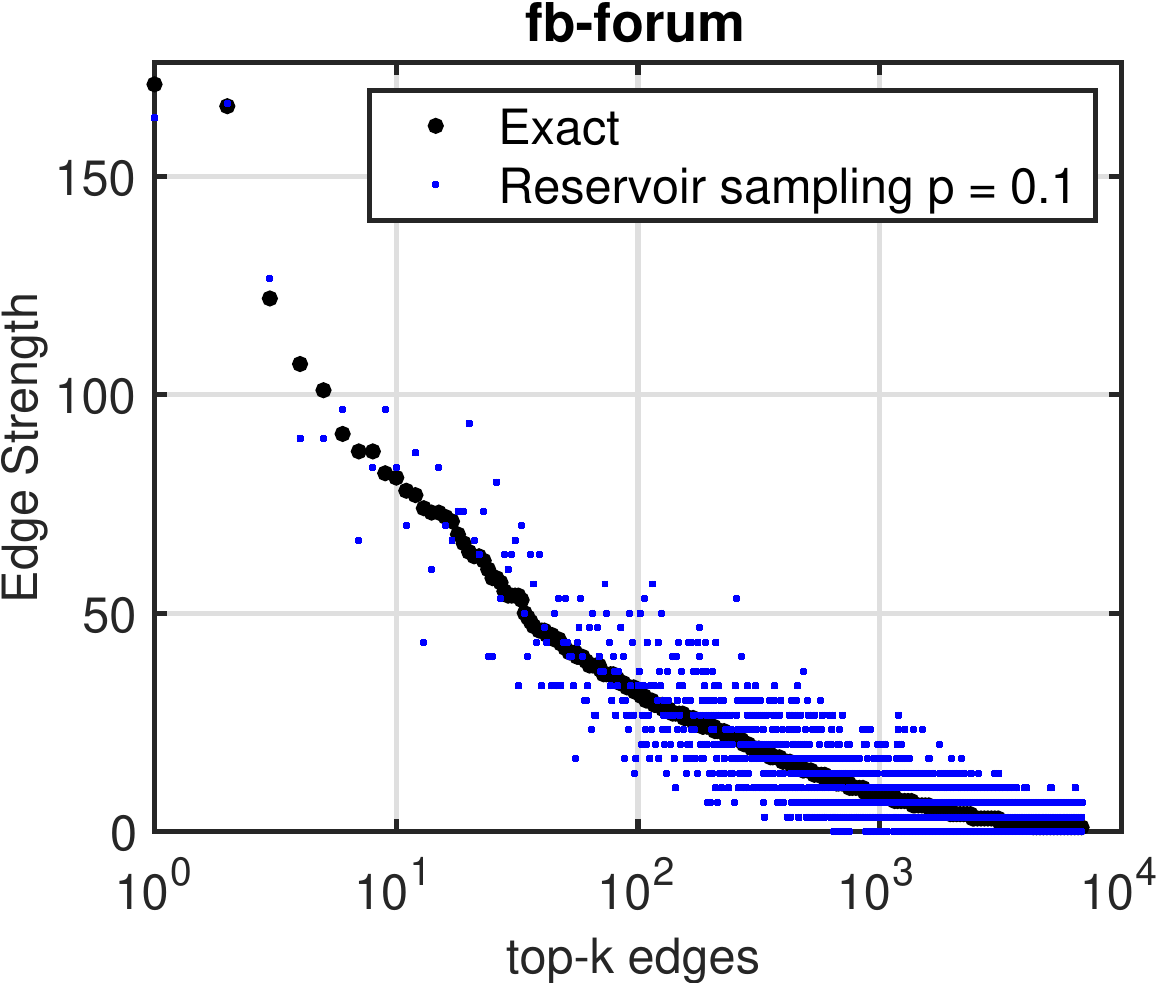}}
\hspace{3mm}
\subfigure
{\includegraphics[width=\figsz\linewidth]{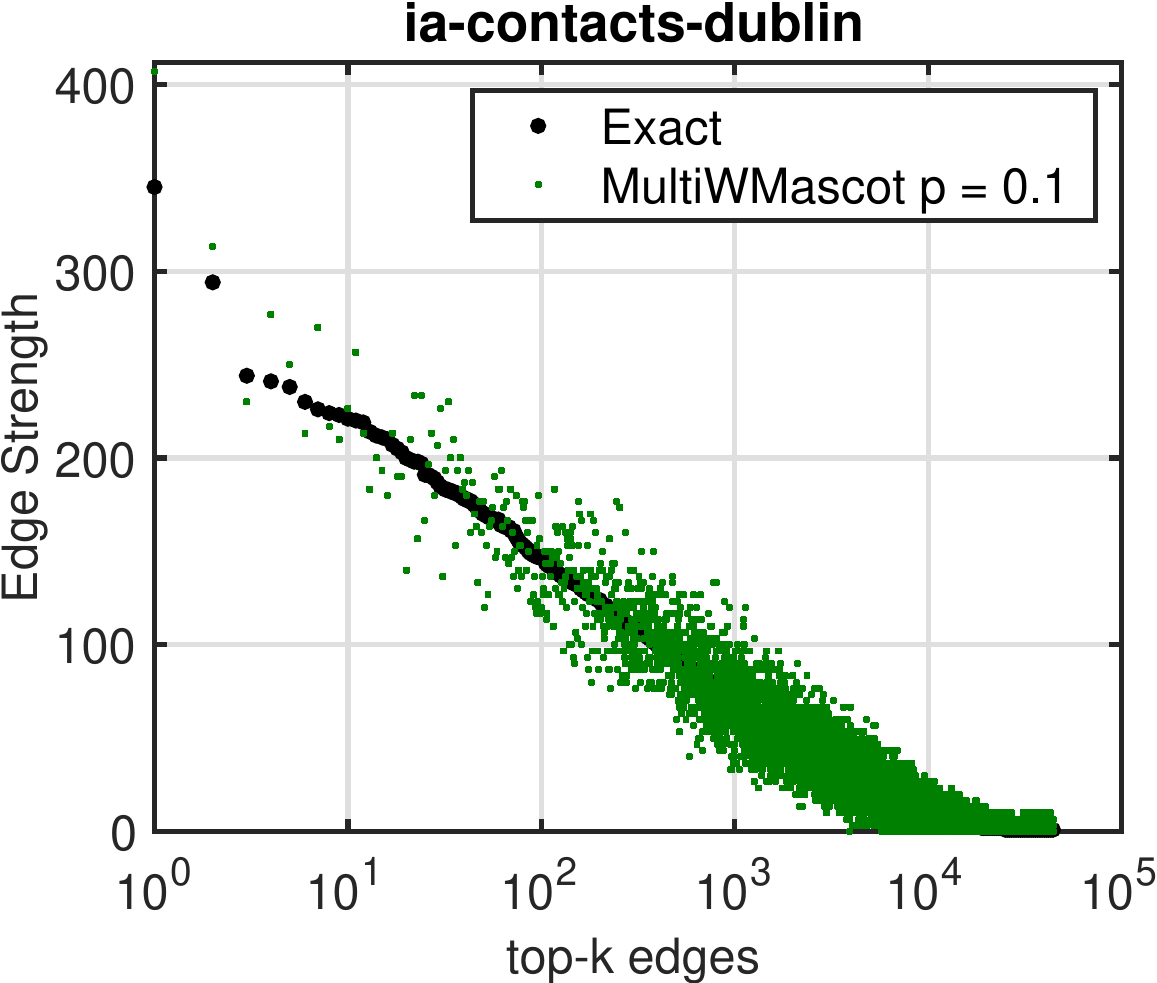}}

\subfigure
{\includegraphics[width=\figsz\linewidth]{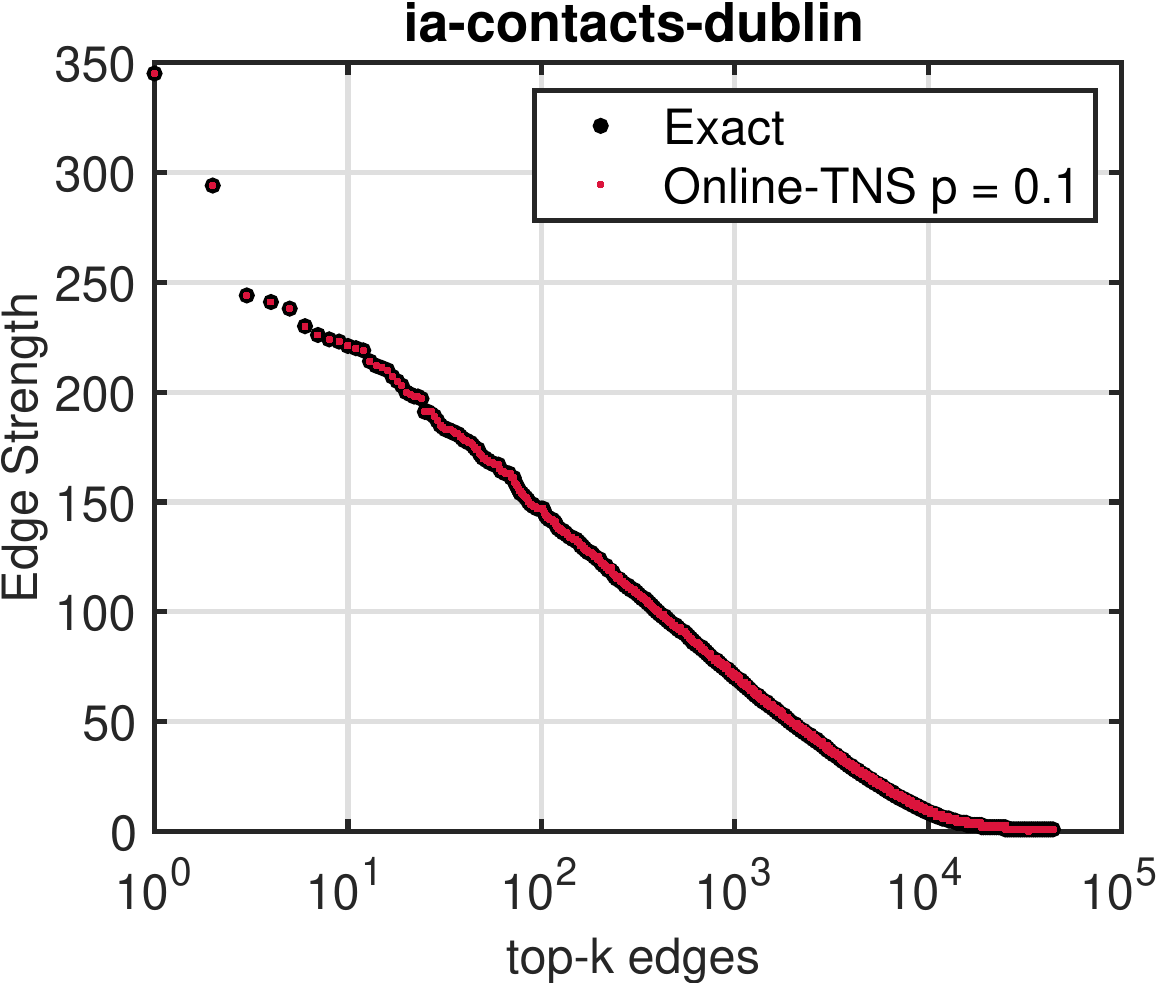}}
\hspace{3mm}
\subfigure
{\includegraphics[width=\figsz\linewidth]{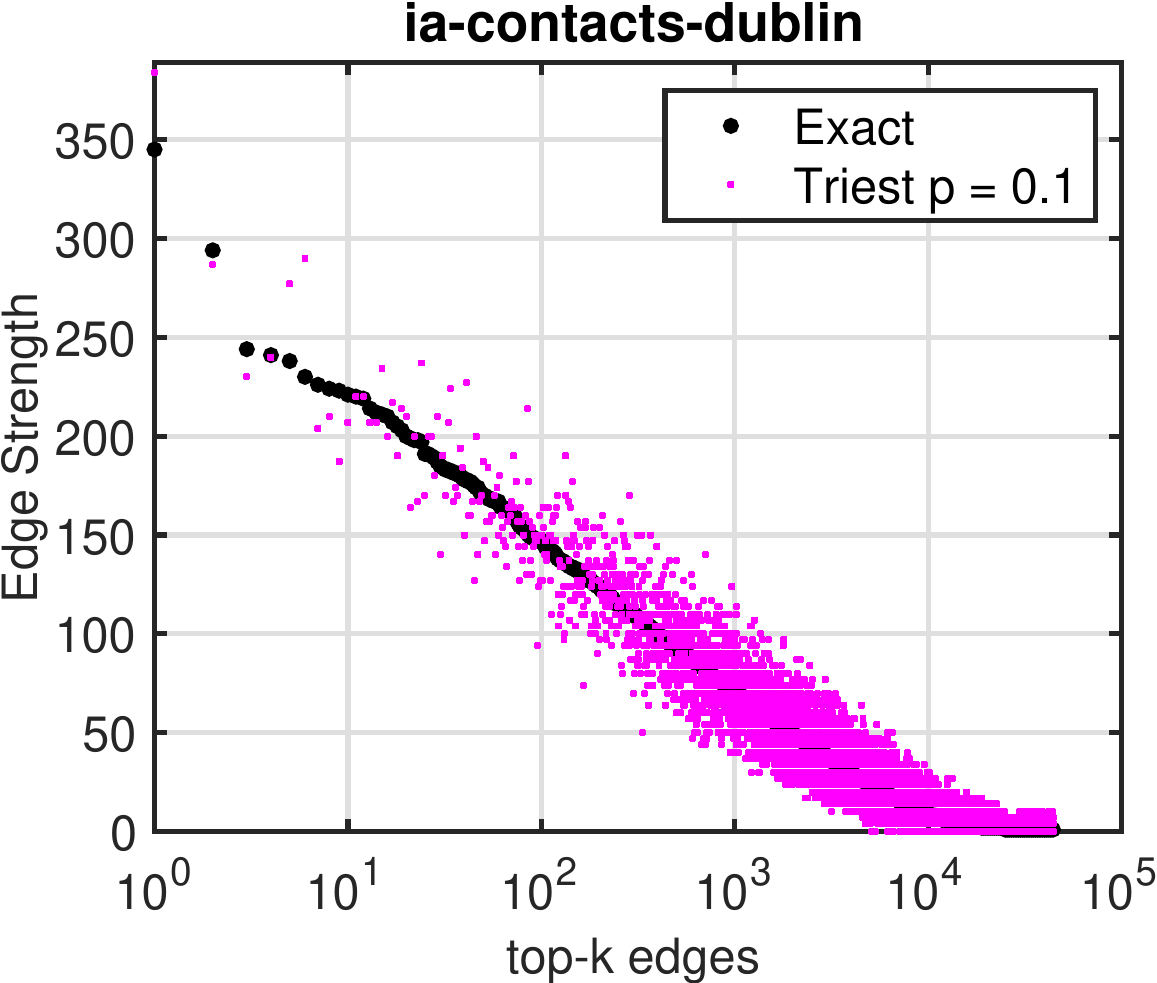}}
\hspace{3mm}
\subfigure
{\includegraphics[width=\figsz\linewidth]{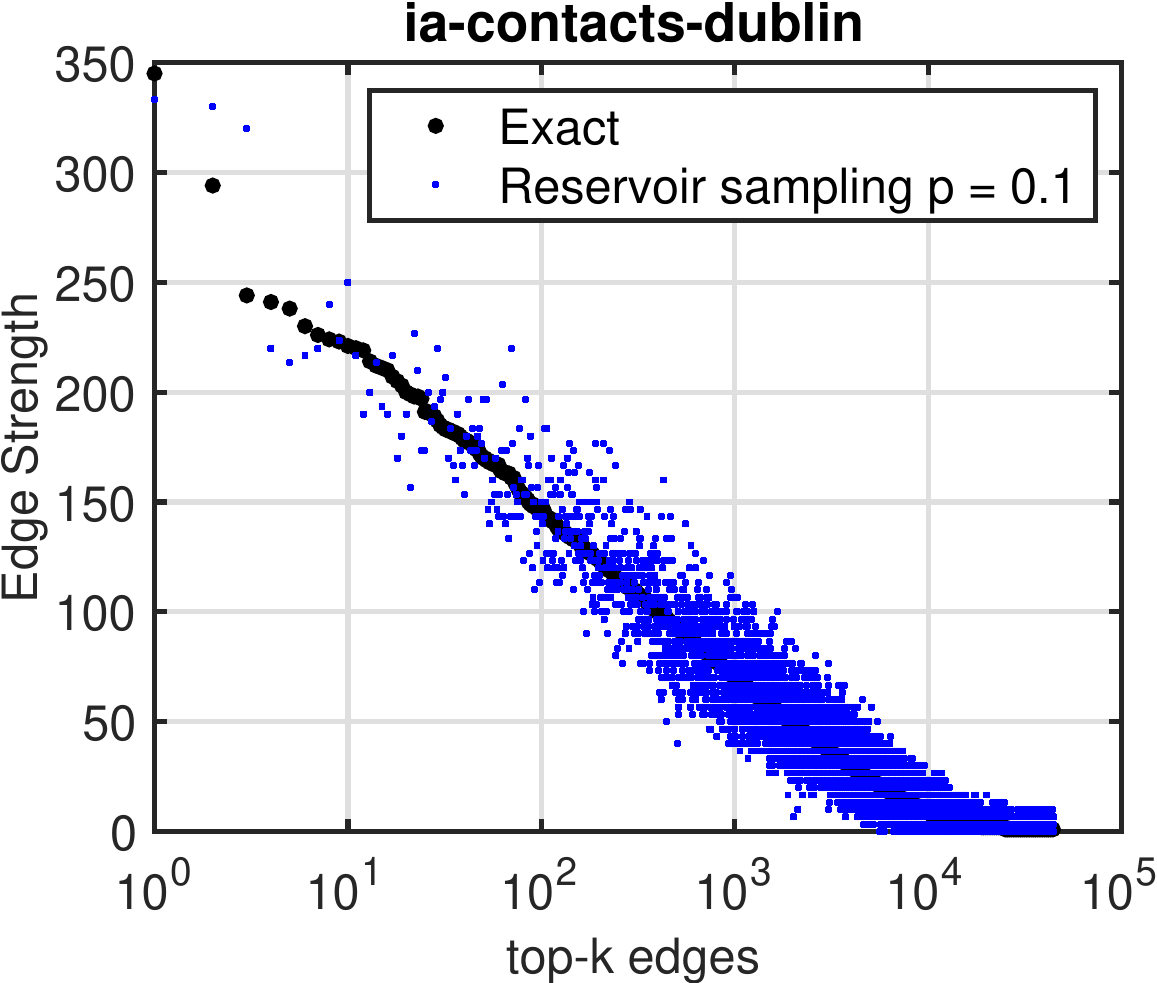}}
\hspace{3mm}
\subfigure
{\includegraphics[width=\figsz\linewidth]{new-figs/fig-edge-weight-mascot-nodecay-10-ia-contacts-dublin-eps-converted-to.pdf}}
\vspace{-2mm}
\caption{Baseline comparison with Triest, reservoir sampling, and MultiWMascot sampling. Temporal link strength (No-decay) estimated distribution vs exact distribution for top-k links .
Results are shown for sampling fraction $p=0.1$.
}
\label{fig:edge-weight-baseline}
\end{figure*}

\subsection{Comparison to Published Baselines}
\label{sec:baseline}

In Table~\ref{table:baseline-spectral-error}, we compare  Algorithm~\ref{alg-TNS} to the state-of-the-art methods for multi-graph streams (in which edges can appear more than once in the stream), Triest sampling~\cite{stefani2017triest}, reservoir sampling~\cite{vitter1985random}, and  MultiWMascot~\cite{lim2018memory} for the estimation of link strength (with no-decay). Note that both Triest and reservoir sampling methods sample edges separately, and multiple occurrences of an edge $(u,v)$ may appear in the final sample. 
On the other hand, our proposed Algorithm~\ref{alg-TNS} and MultiWMascot incrementally update the overall estimate of link strength of an edge $(u,v)$, and stores an edge only once with its estimator. This leads to more space-efficient samples. However, while MultiWMascot maintains the edge estimators, it is still uses a single uniform probability $p$ for sampling any of the edges.  

We observe that both Triest and reservoir sampling were unable to produce a reasonable estimate with $59\%-83\%$ accuracy, while MultiWMascot performed slightly better with an accuracy of $60\%-84\%$. Algorithm~\ref{alg-TNS} produced more accurate estimates with $82\%-97\%$ accuracy, with an average of $20\%$ gain in accuracy compared to the baselines. Figure~\ref{fig:edge-weight-baseline} shows the distribution of the top-k edges ($k = 10$ million) and compares the exact link strength with the estimated link strength for the four sampling algorithms, Online-TNS (Alg.~\ref{alg-TNS}), Triest, reservoir sampling, and MultiWMascot. Notably, Online-TNS not only accurately estimates the strength of the link but also captures the correct order of the links compared to the baselines.   

\subsection{Ablation Study}
\label{sec:ablation}

\newcommand\MyBox[2]{
  \fbox{\lower0.7cm
    \vbox to 2.0cm{\vfil
      \hbox to 
      2.9cm{\hfil\parbox{2.9cm}{#1\\#2}\hfil}
      \vfil}%
  }%
}
\begin{table}
\vspace{-5mm}
\renewcommand\arraystretch{1.5}
\setlength\tabcolsep{0pt}
\begin{tabular}{c >{\bfseries}r @{\hspace{0.4em}}c @{\hspace{0.4em}}c}
  \multirow{12}{*}{\parbox{0.5cm}{\bfseries\rotatebox{90}{Sampling Weights}}} & 
    & \multicolumn{2}{c}{\bfseries Link Strength Estimator} \\
  & & \bfseries No decay & \bfseries Temporal decay \\
  & \rotatebox{0}{Uniform} & \MyBox{$w(e) = \phi$}{\\$\hat{C}(e)$ using Line~\ref{func:edge_strength}} & \MyBox{$w(e) = \phi$}{\\$\hat{C}(e)$ using Line~\ref{func:edge_strength_decay}}  \\[2.7em]
  & \rotatebox{0}{Adaptive} & \MyBox{\textbf{if} $e \in \hat{K}$ \\ $w(e) = w(e) + 1$ \\ \textbf{else} $w(e) = \phi$}{\\$\hat{C}(e)$ using Line~\ref{func:edge_strength}} & \MyBox{\textbf{if} $e \in \hat{K}$ \\ $w(e) = w(e) + 1$ \\ \textbf{else} $w(e) = \phi$}{\\$\hat{C}(e) $ using Line~\ref{func:edge_strength_decay}}   \\
\end{tabular}
\centering
\vspace{1mm}
\caption{Online-TNS Framework Main Components (see Algorithm~\ref{alg-TNS}). 
}
\label{table:options}
\end{table}

Our proposed framework is flexible and generic with various components and design choices. To help understand the contributions of the major components, we performed a thorough set of ablation study experiments. Our framework (in Algorithm~\ref{alg-TNS}) consists of two major components: Adaptive sampling and estimation (uniform vs adaptive sampling weights), and Link-decay models (no-decay vs exponential decay). We summarize these components and design choices in Table~\ref{table:options}. Our first ablation study experiment investigates the impact of sampling weights on the estimation accuracy. We explore two variants of Algorithm~\ref{alg-TNS}, (a) \emph{Adaptive}: using Algorithm~\ref{alg-TNS} with adaptive/importance sampling weights, where the sampling weights/ranks adapt to allow edges to gain importance during stream processing. (b) \emph{Uniform}: using Algorithm~\ref{alg-TNS} with fixed uniform sampling weights, where the sampling weights/ranks are uniform and assigned at the first time of sampling, and fixed during the rest of the streaming process. For variant (a), we use the exact procedure in Algorithm~\ref{alg-TNS}. For variant (b), we omit Lines~\ref{line:weight_adapt} and~\ref{line:rank_update} from Algorithm~\ref{alg-TNS}. Note that for both variants, the established estimators in Section~\ref{sec:topology-adaptive-est} are unbiased. We compare performance of the two variants for the estimation of link strength and temporally weighted motif counts.

\renewcommand{\figsz}{0.22}
\begin{figure*}[h!]
\centering
\subfigure
{\includegraphics[width=\figsz\linewidth]{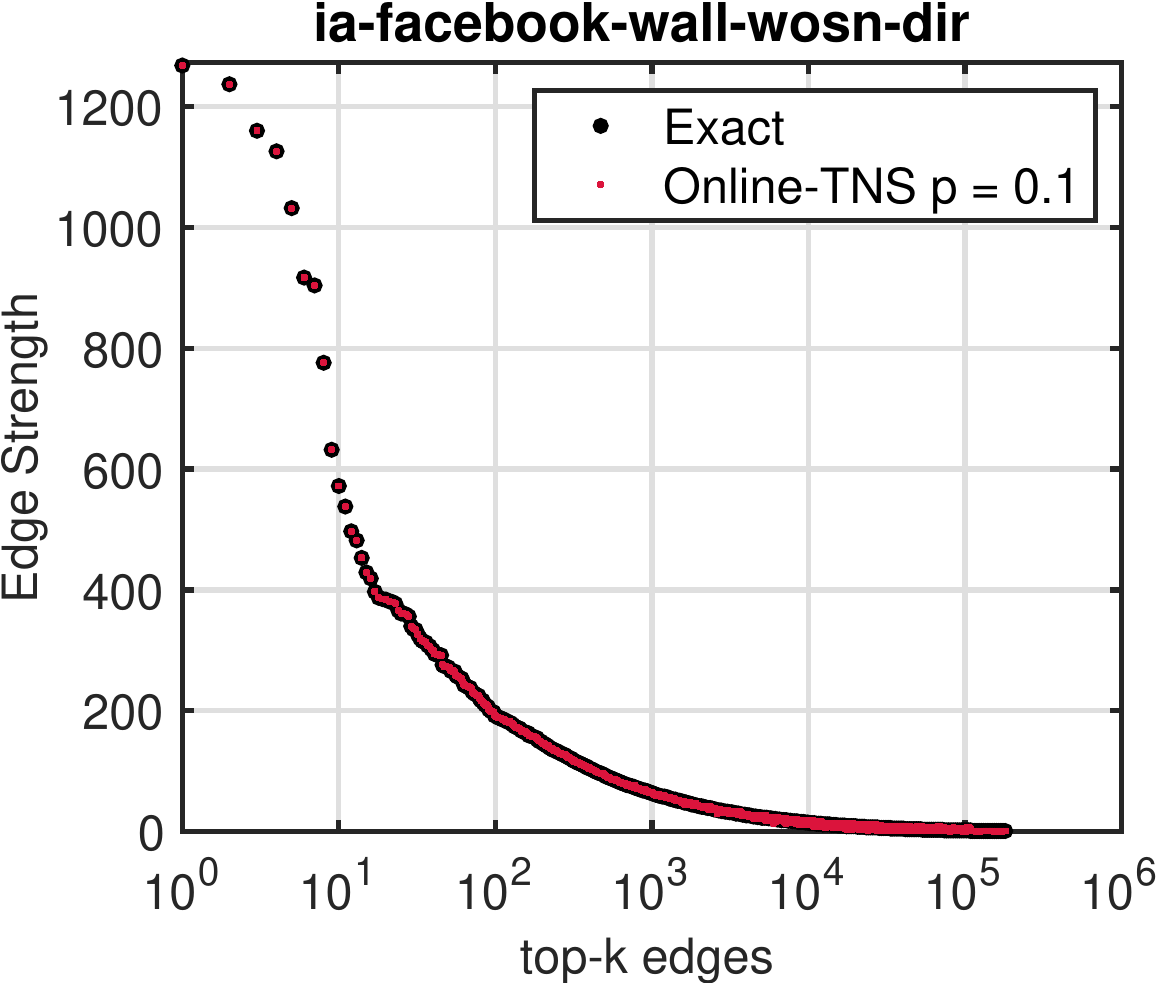}}
\hspace{3mm}
\subfigure
{\includegraphics[width=\figsz\linewidth]{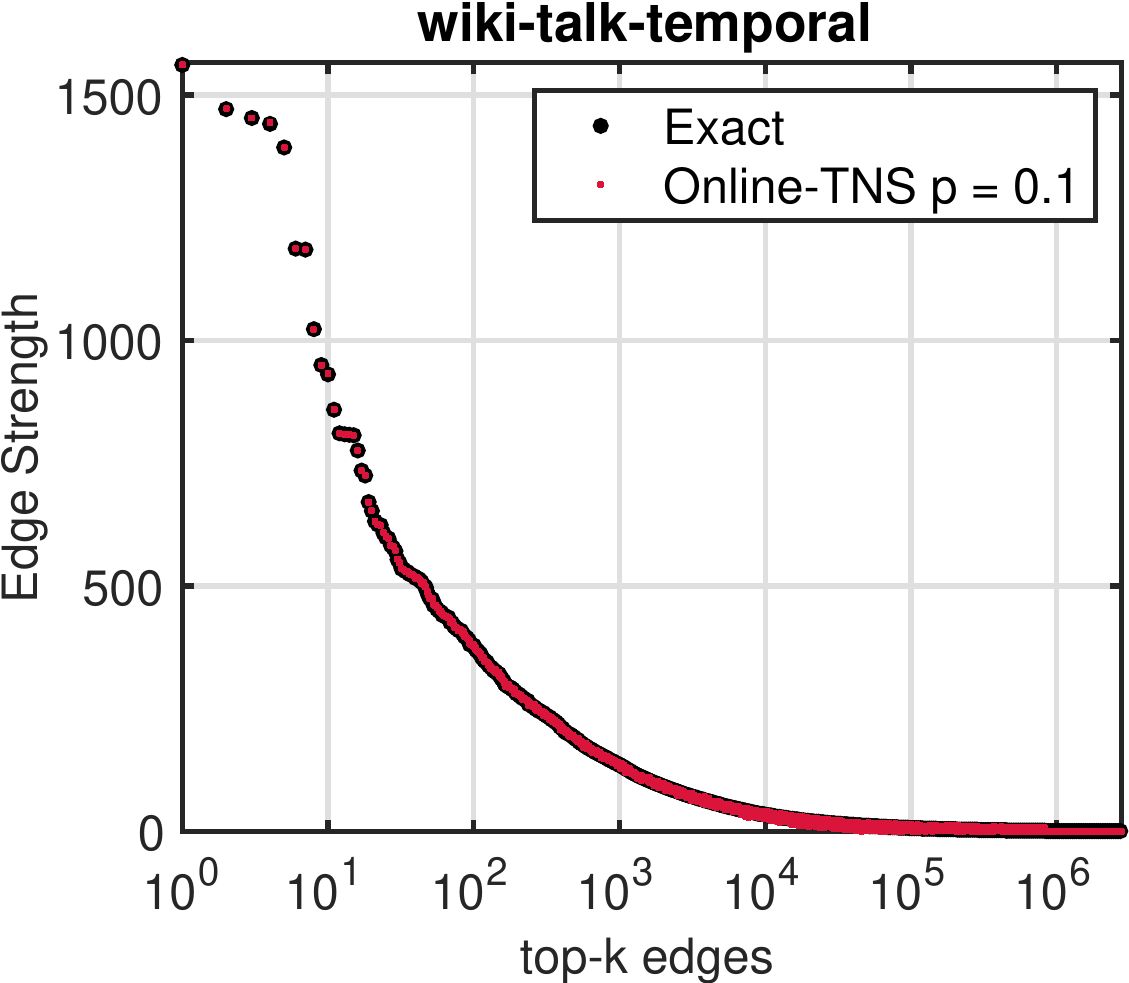}}
\hspace{3mm}
\subfigure
{\includegraphics[width=\figsz\linewidth]{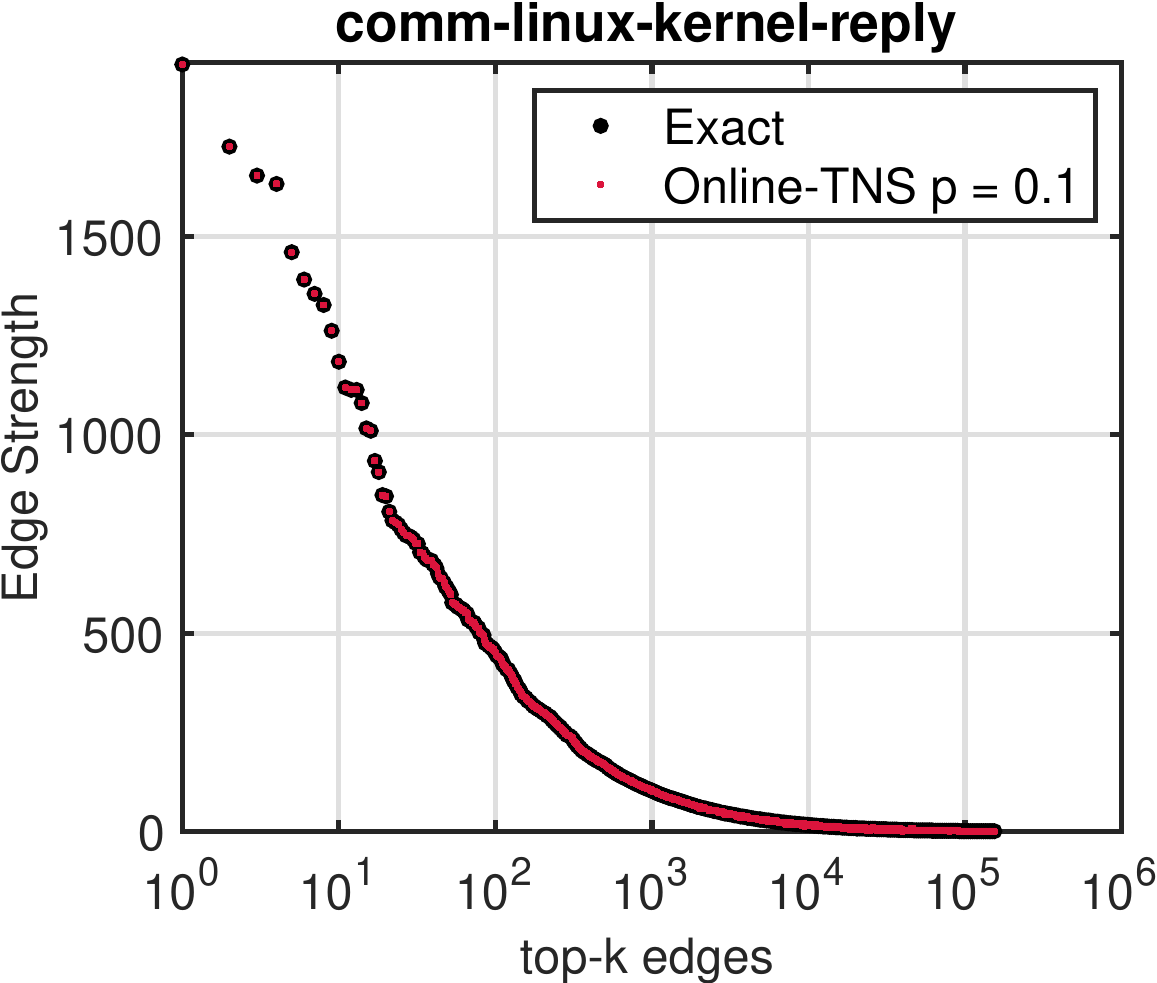}}
\hspace{3mm}
\subfigure
{\includegraphics[width=\figsz\linewidth]{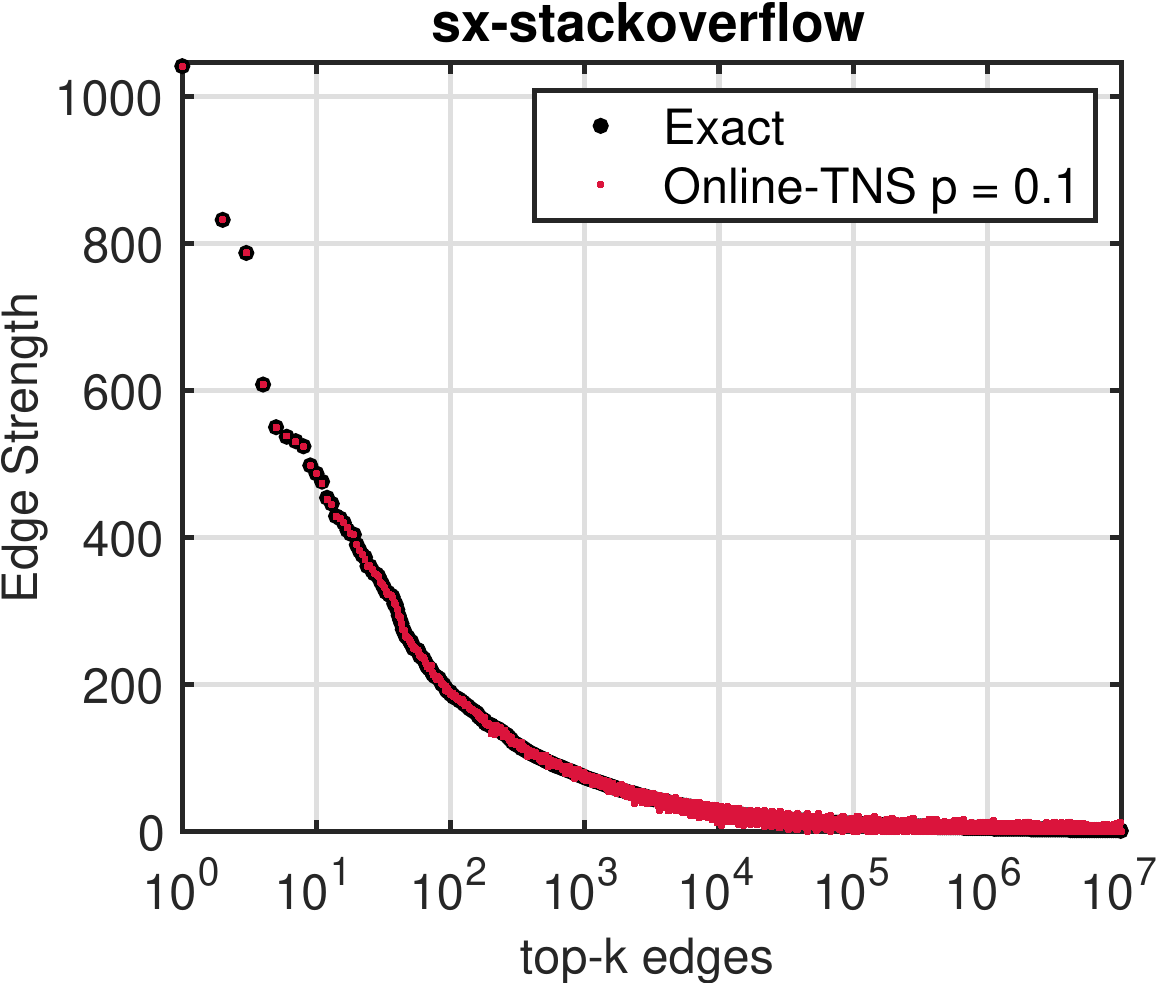}}
\subfigure
{\includegraphics[width=\figsz\linewidth]{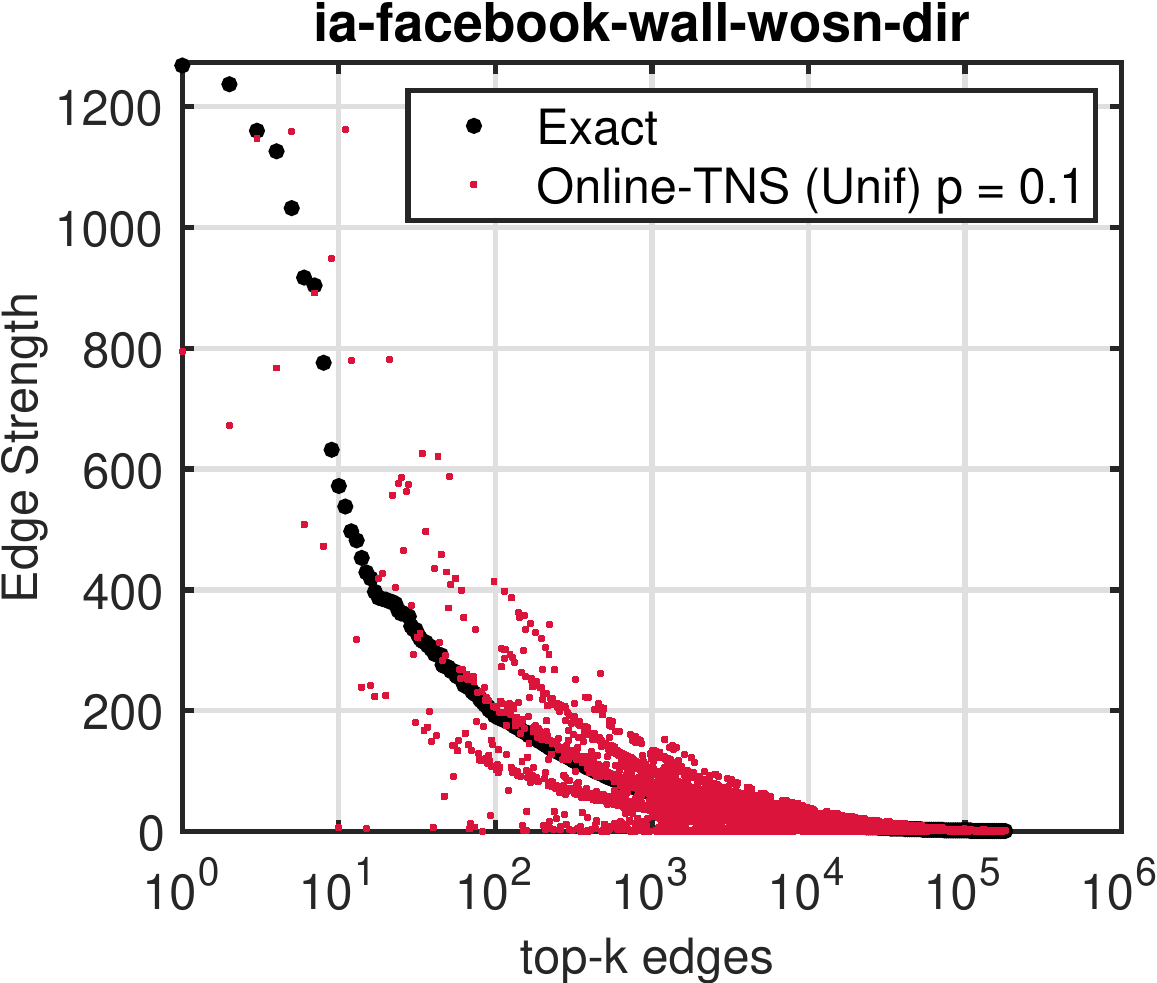}}
\hspace{3mm}
\subfigure
{\includegraphics[width=\figsz\linewidth]{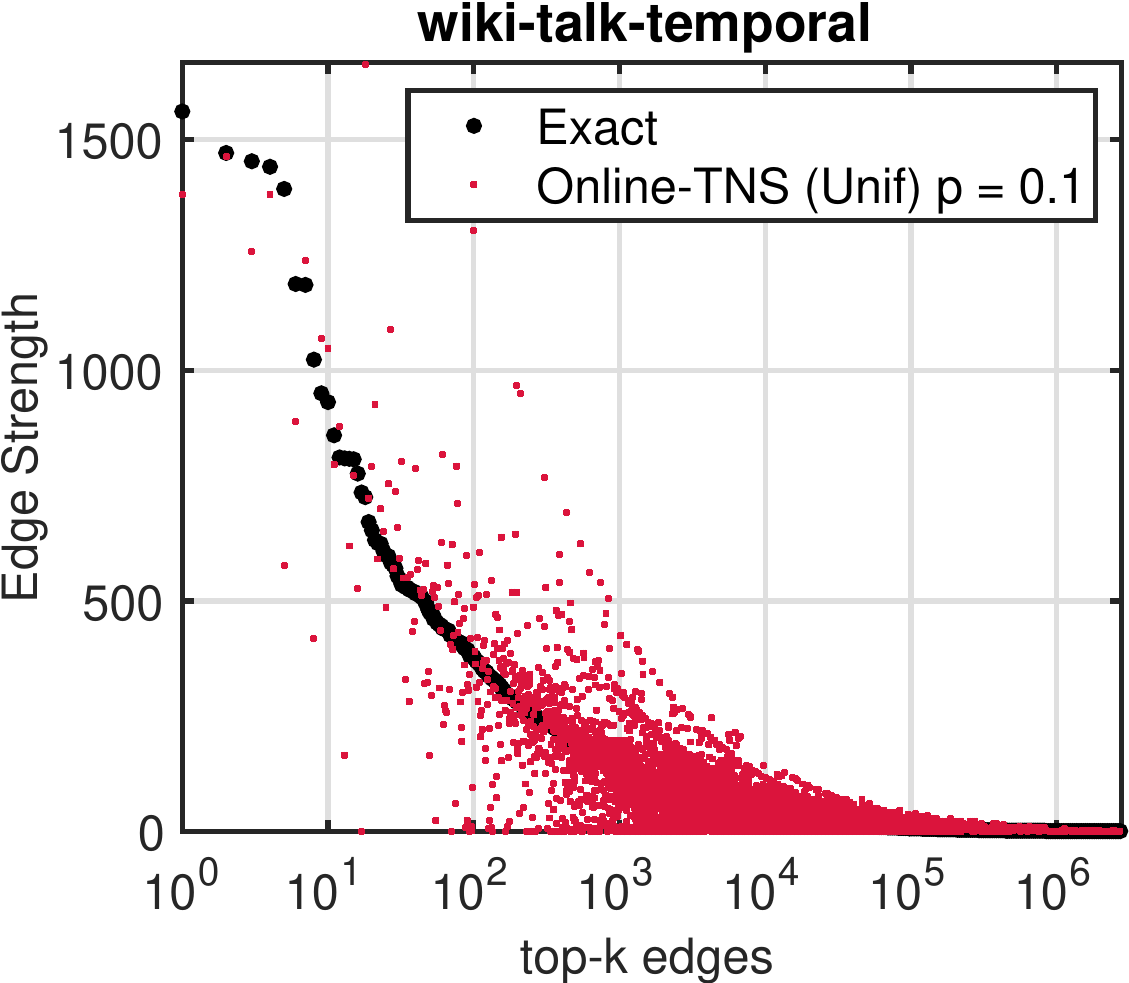}}
\hspace{3mm}
\subfigure
{\includegraphics[width=\figsz\linewidth]{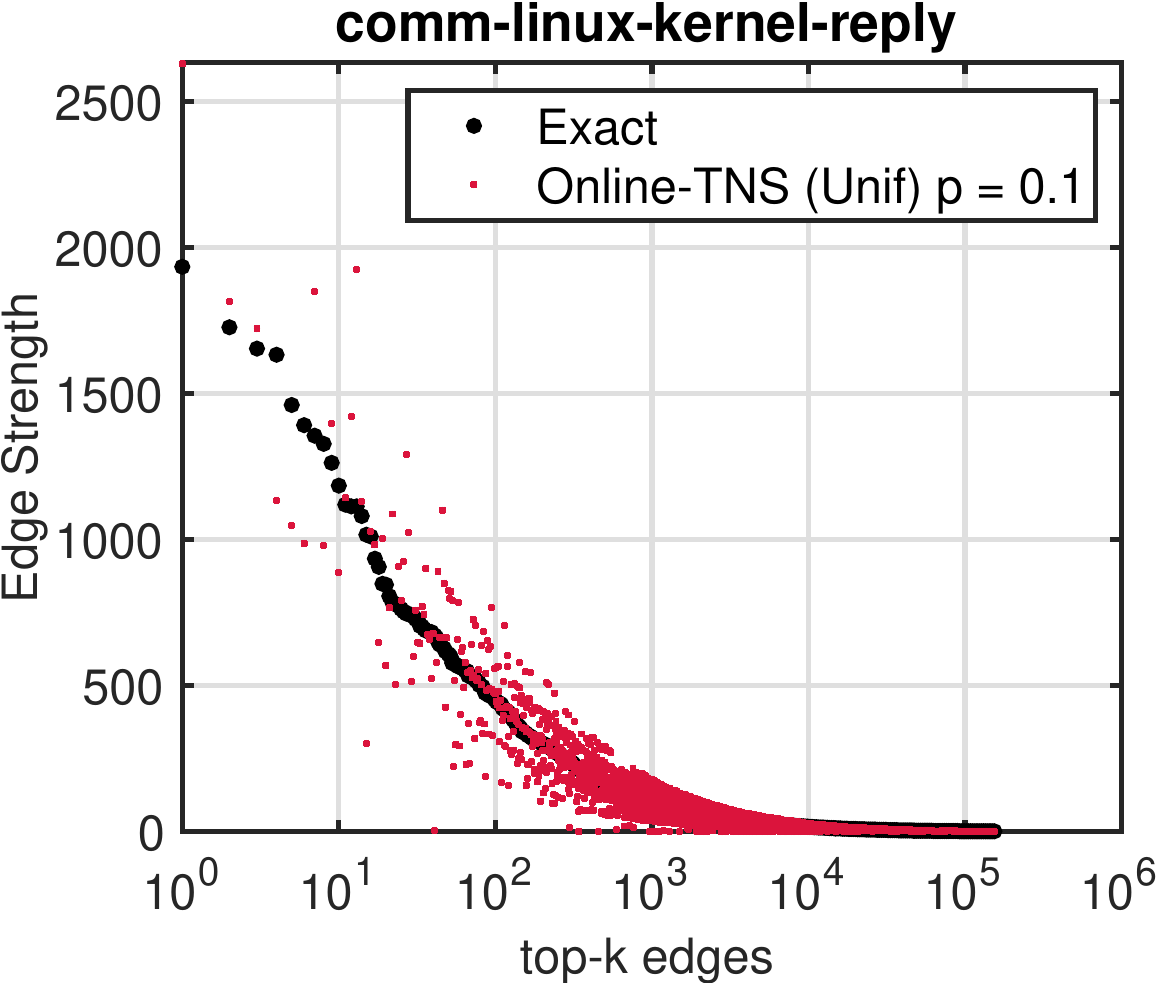}}
\hspace{3mm}
\subfigure
{\includegraphics[width=\figsz\linewidth]{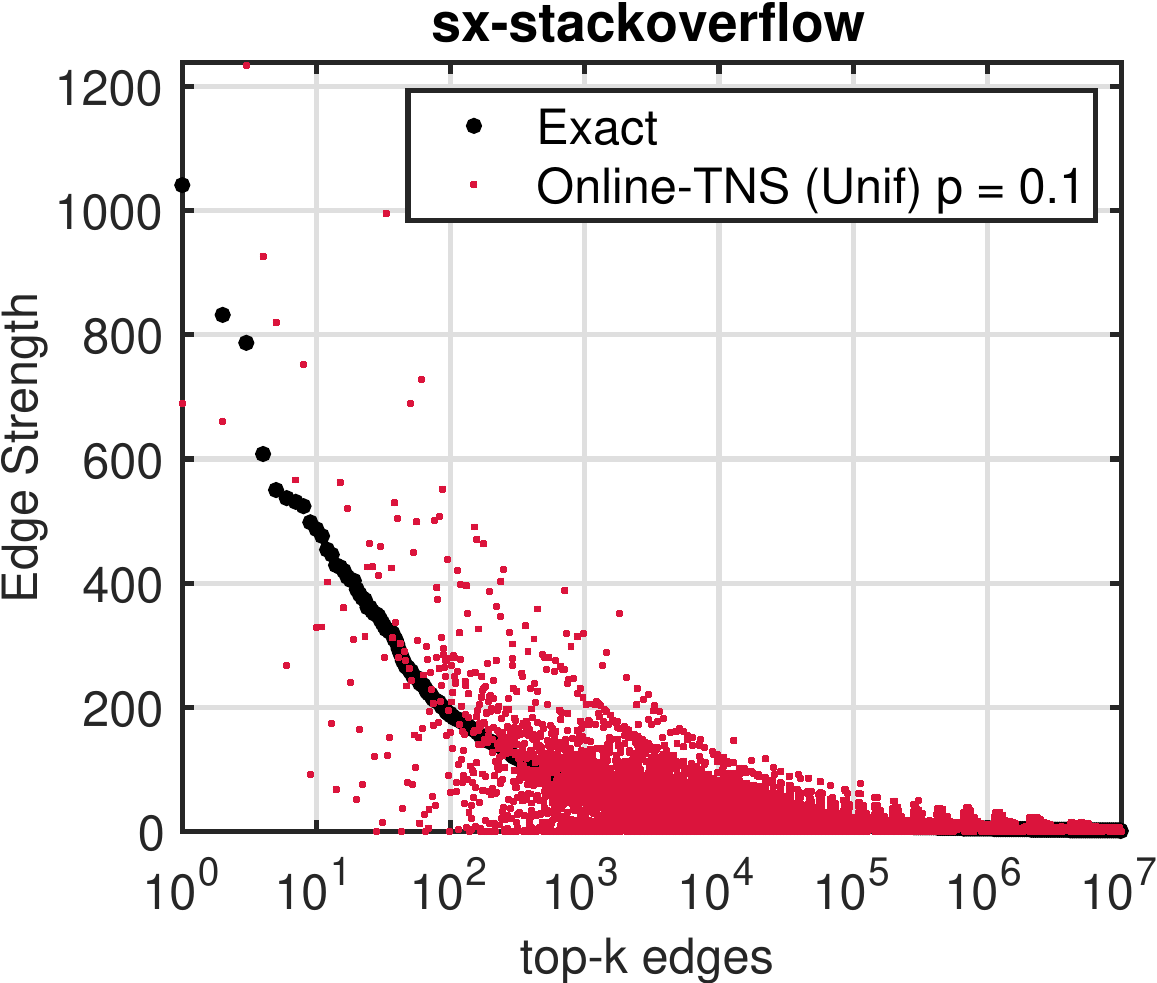}}

\caption{Temporal link strength (No-decay) estimated distribution vs exact distribution for top-k links .
Results are shown for sampling fraction $p=0.1$. (Top) Results for
Online-TNS Algorithm with adaptive sampling weights~\ref{alg-TNS}. 
(Bottom) Results for 
Online-TNS Algorithm with uniform sampling weights.
}
\label{fig:edge-weight}
\end{figure*}

\renewcommand{\figsz}{0.22}
\begin{figure*}[h!]
\centering
\subfigure
{\includegraphics[width=\figsz\linewidth]{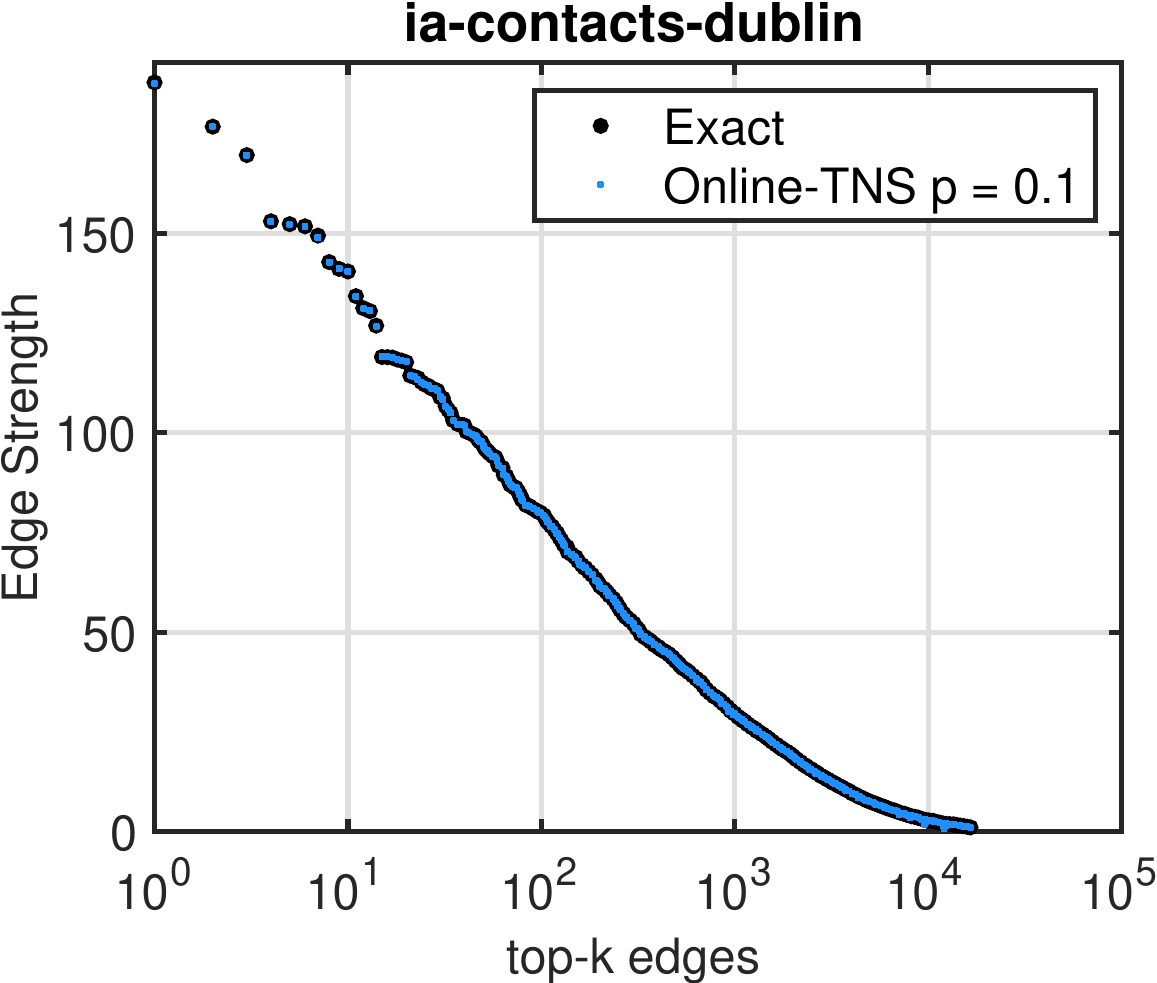}}
\hspace{3mm}
\subfigure
{\includegraphics[width=\figsz\linewidth]{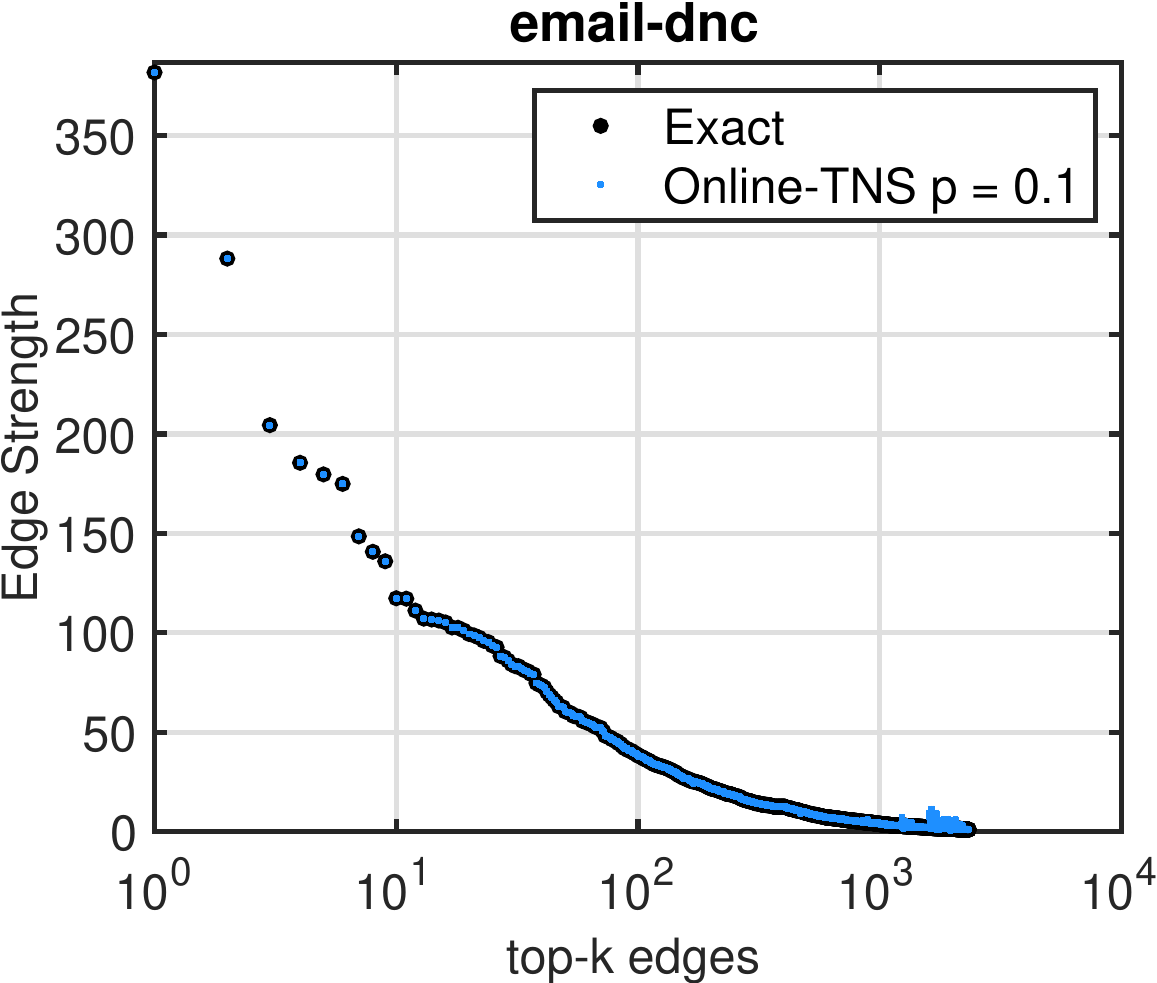}}
\hspace{3mm}
\subfigure
{\includegraphics[width=\figsz\linewidth]{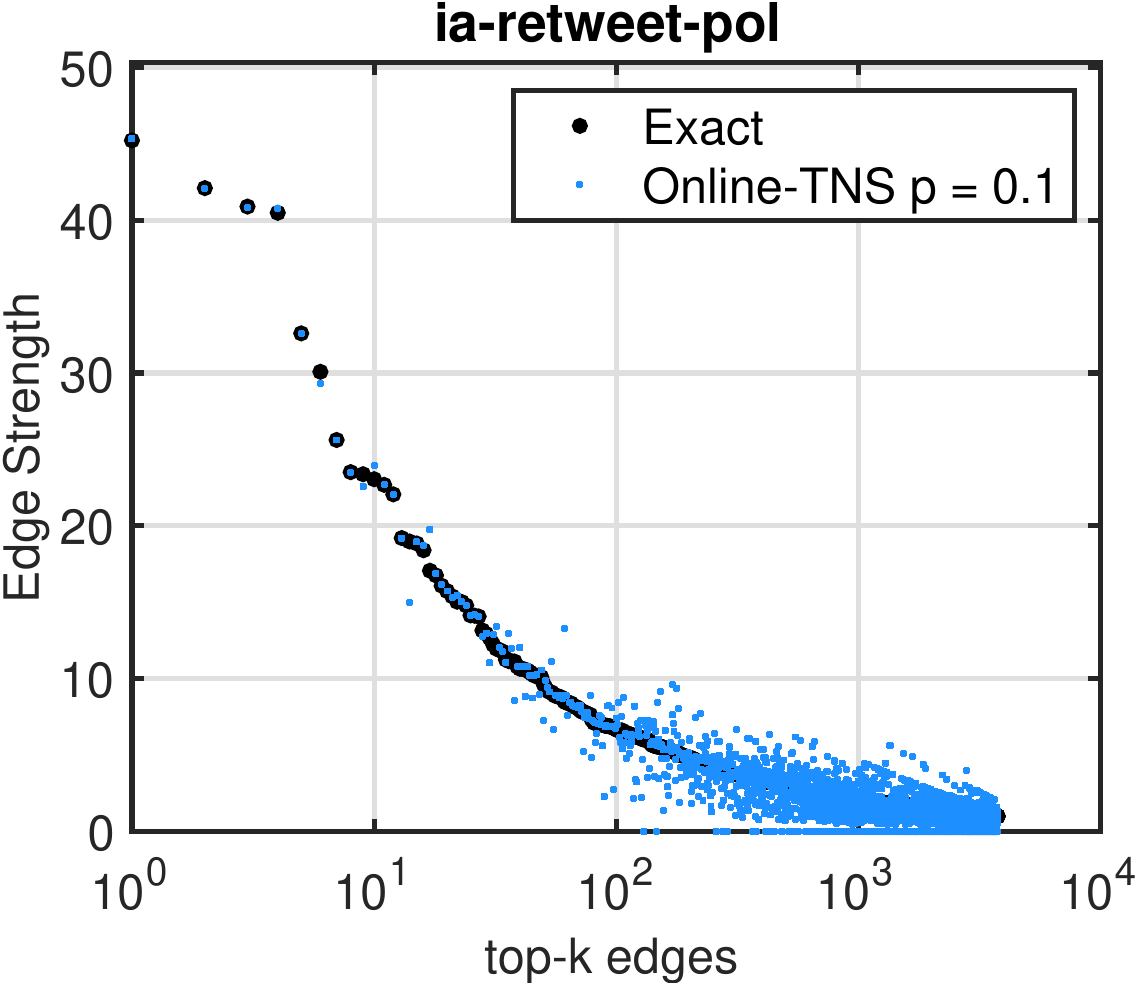}}
\hspace{3mm}
\subfigure
{\includegraphics[width=\figsz\linewidth]{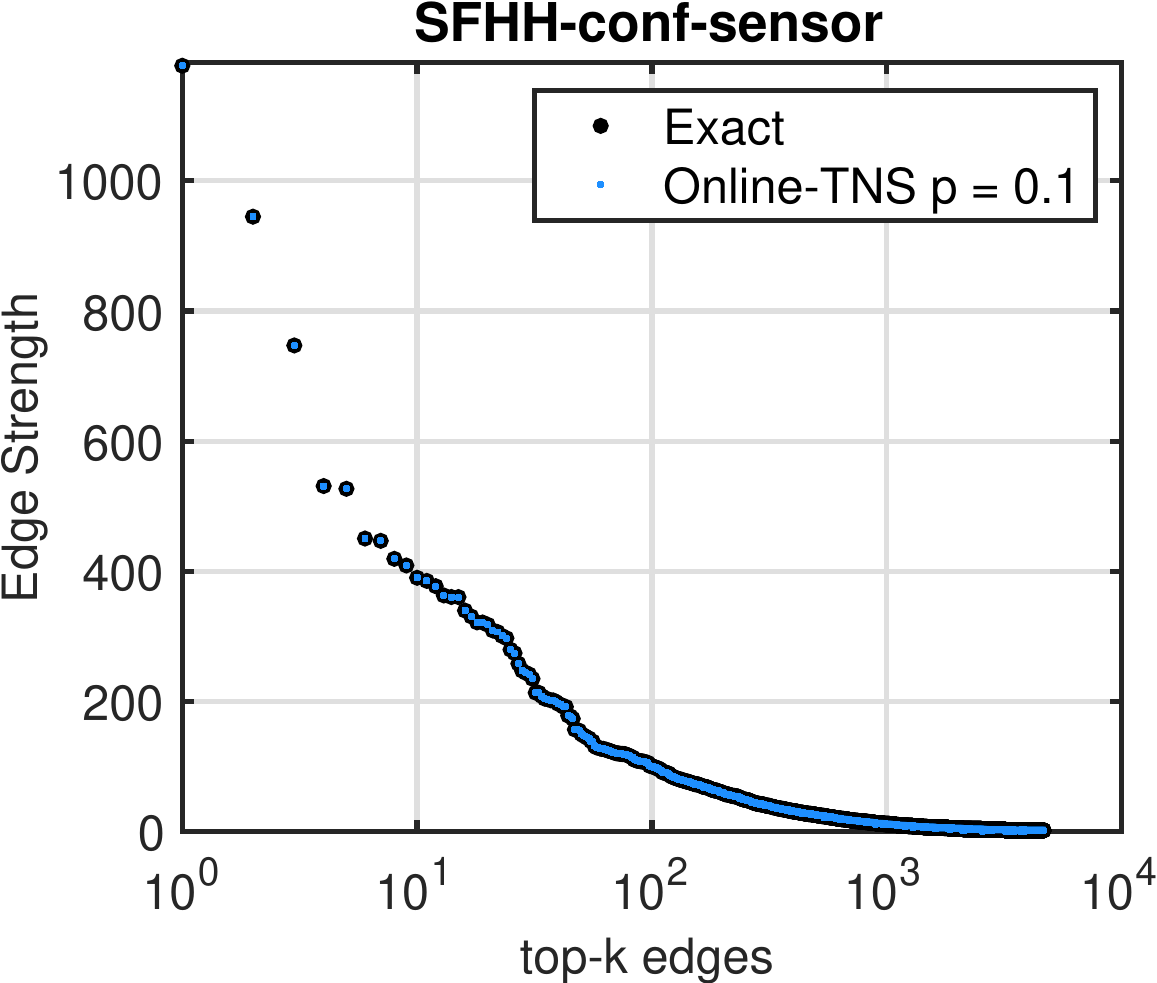}}

\subfigure
{\includegraphics[width=\figsz\linewidth]{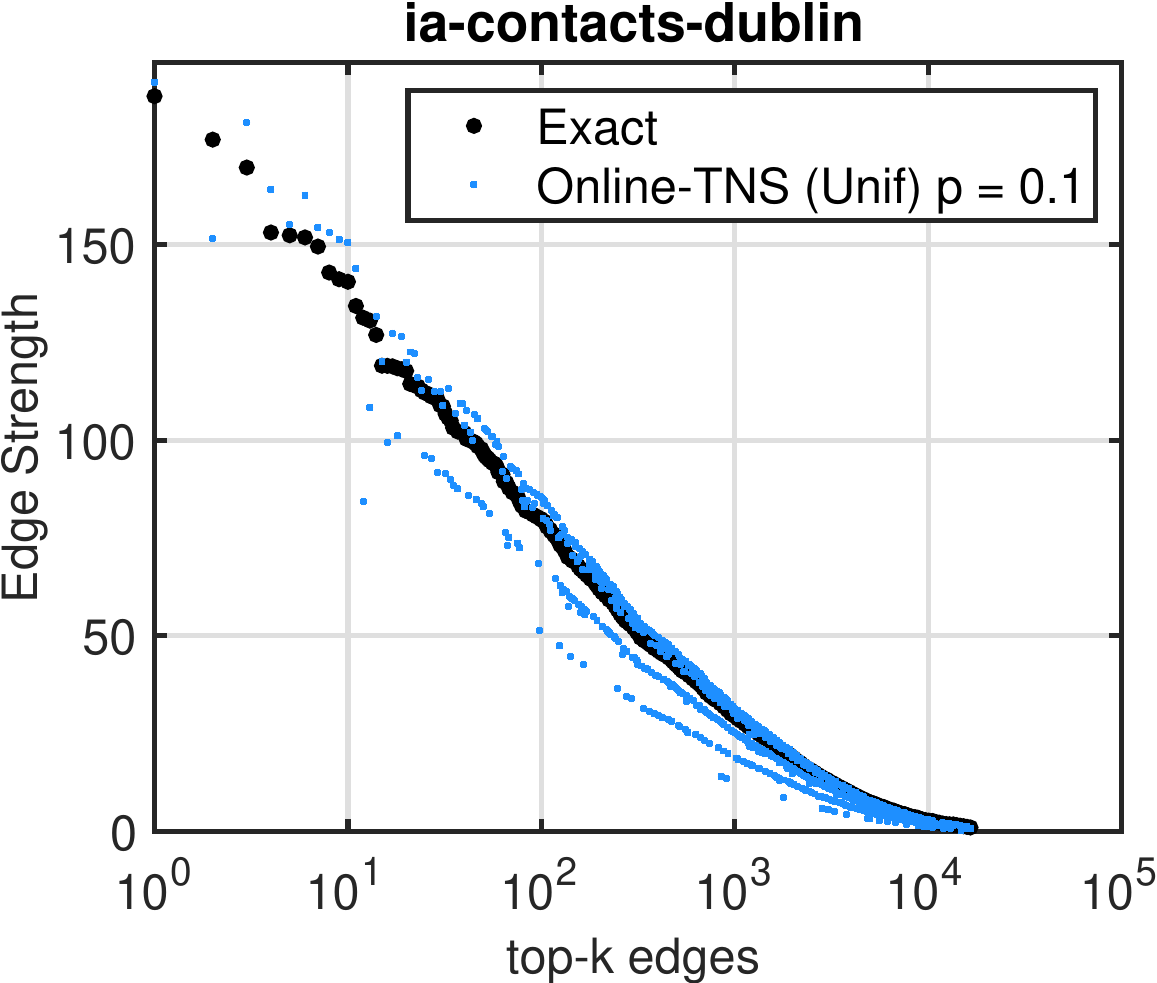}}
\hspace{3mm}
\subfigure
{\includegraphics[width=\figsz\linewidth]{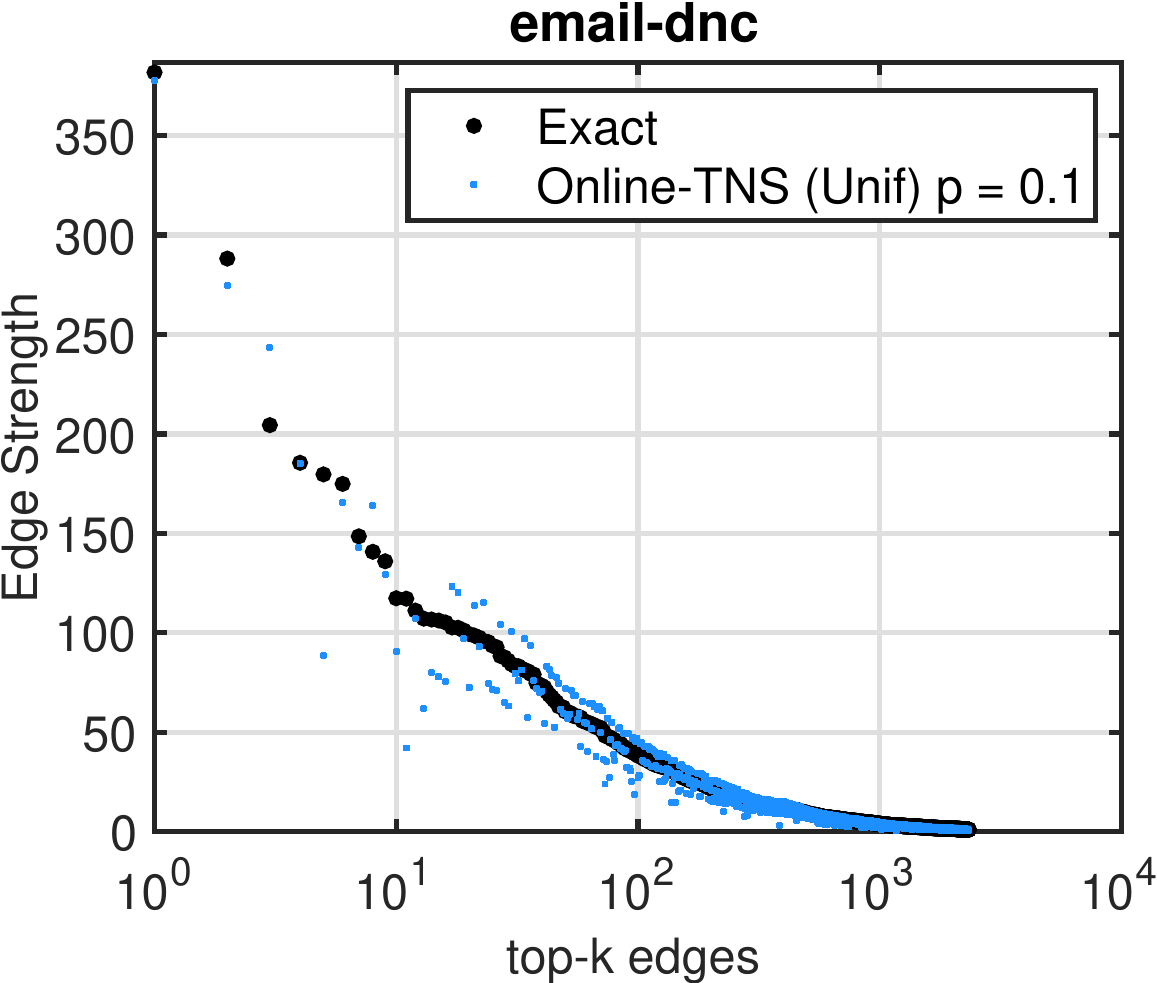}}
\hspace{3mm}
\subfigure
{\includegraphics[width=\figsz\linewidth]{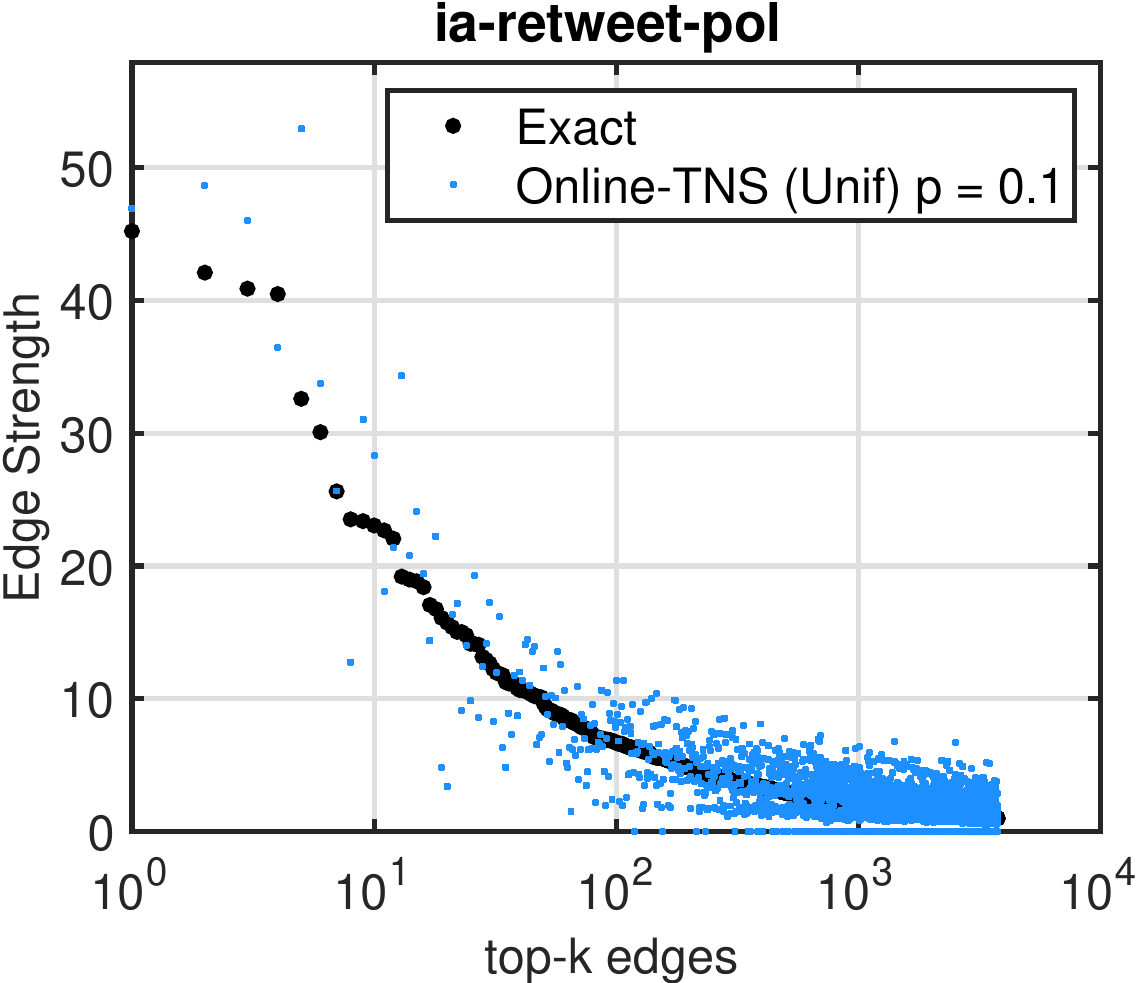}}
\hspace{3mm}
\subfigure
{\includegraphics[width=\figsz\linewidth]{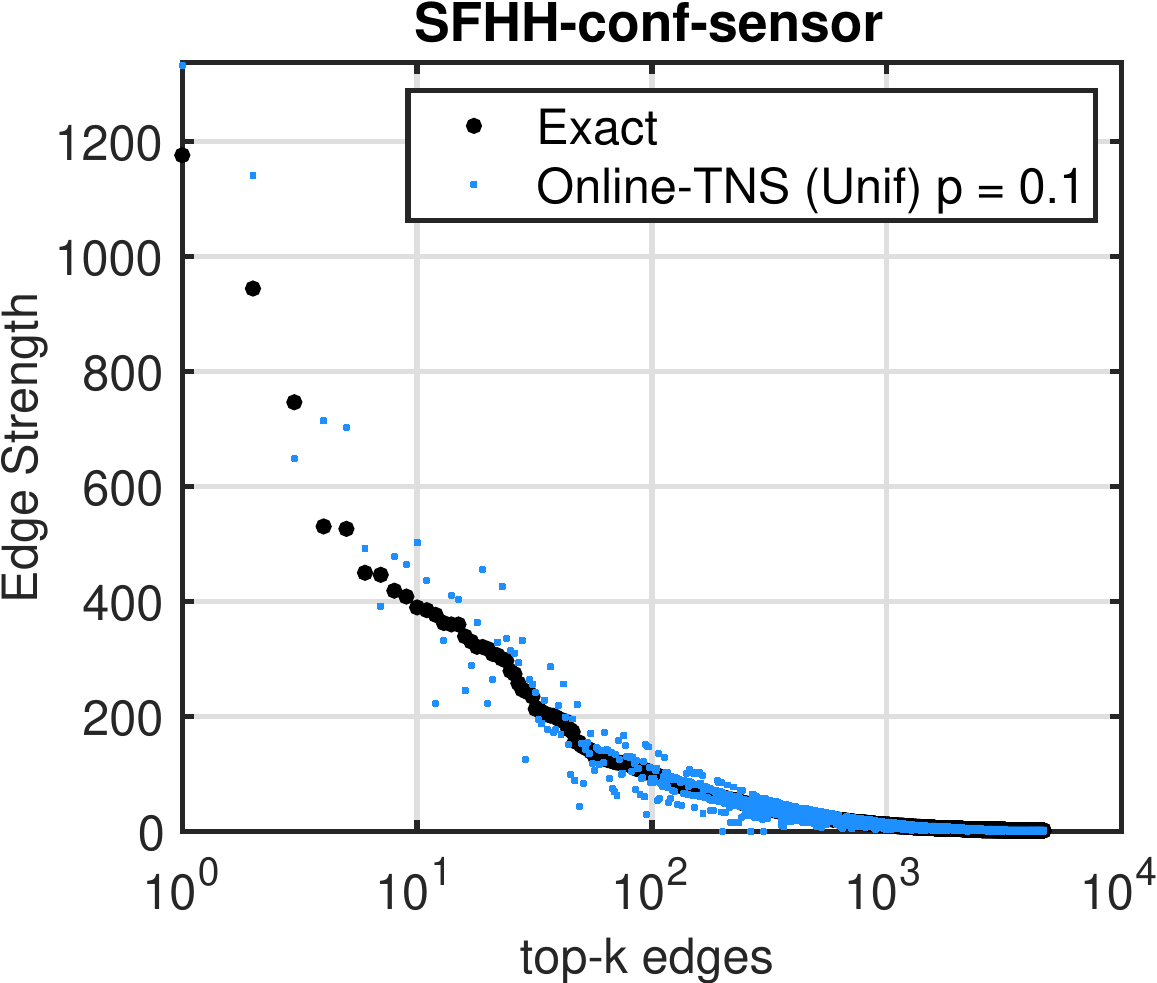}}
\subfigure

\caption{Temporal decay link strength estimated distribution vs exact distribution for top-k links .
Results are shown for sampling fraction $p=0.1$. (Top) Results for
Online-TNS Algorithm with adaptive sampling weights~\ref{alg-TNS}. 
(Bottom) Results for 
Online-TNS Algorithm with uniform sampling weights.
}
\label{fig:edge-weight-decay}
\end{figure*}

\subsubsection{Estimation of Temporal Link Strength} \label{sec:exp-temporal-link-strength}
Link strength is one of the most fundamental properties of temporal networks~\cite{xiang2010modeling}.
Therefore, estimating it in an online fashion is clearly important. 
Results using Alg.~\ref{alg-TNS} with adaptive sampling weights and unbiased estimators (see Section~\ref{sec:topology-adaptive-est}) for temporal link strength estimation are provided in Figures~\ref{fig:edge-weight} and~\ref{fig:edge-weight-decay} (top row) for the no-decay and link-decay models respectively. We show the distribution of the top-k edges ($k= 10$ million) and compare the exact link strength with the estimated link strength. Notably our approach not only accurately estimates the strength of the link but also captures the correct order of the links (top-links ordered by their strength from high to low). 
From Figures~\ref{fig:edge-weight} and~\ref{fig:edge-weight-decay} (top row), we observe the exact and estimated link strengths for the top-k edges to be nearly indistinguishable from one another. We also compare to Alg.~\ref{alg-TNS} with uniform sampling weights, \ie, Online-TNS (Unif), in which edges are assigned uniform sampling weights at the sampling time, and fixed for the rest of the streaming process, results are shown in Figures~\ref{fig:edge-weight} and~\ref{fig:edge-weight-decay} (bottom row). While the estimated link strengths from Online-TNS with adaptive weights are nearly identical to the exact link strengths, estimated distributions from uniform sampling weights are significantly worse. Unlike uniform sampling weights where the weights remain constant, using the adaptive sampling weights helps the algorithm to adapt to the changing topology of the streaming network, which leads to favoring the retention of edges with higher strength. 

\begin{table}[h!]
\centering
\caption{%
No-decay Results:  Relative spectral norm (\ie, $\norm{\mathbf{C} - \hat{\mathbf{C}}}_2/\norm{\mathbf{C}}_2$) for sampling fraction $p=0.1$, comparison between adaptive sampling weights (Alg.~\ref{alg-TNS}) and uniform sampling weights. 
}
\label{table:results-spectral-error}
\renewcommand{\arraystretch}{1.10} 
\setlength{\tabcolsep}{8.0pt} 
\begin{tabularx}{0.7\linewidth}{@{}
r cc
HHHHH H
@{}
}
\toprule

\textsc{Temporal Network} & \textbf{Adaptive} \textbf{(Alg.~\ref{alg-TNS})} & \textbf{Uniform} &  \\
\midrule
\textsf{ia-facebook-wall} & 0.0090 & 0.3976 &  \\
\textsf{sx-stackoverflow}  & 0.0992 & 0.4360 & \\ 
\textsf{wiki-talk} & 0.1020 & 0.4272 & \\
\textsf{comm-linux-reply} & 0.0041 & 0.1978 &  \\
\textsf{fb-forum} & 0.0390 & 0.2640 & \\
\textsf{ia-enron-email} & 0.0098 & 0.4080 & \\
\textsf{SFHH-conf-sensor} & 0.0090 & 0.2770 & \\
\textsf{ia-contacts-hyper} & 0.0034 & 0.0529 & \\
\bottomrule
\end{tabularx}
\end{table}

\begin{table}[h!]
\centering
\caption{%
Link-decay Results:  Relative spectral norm (\ie, $\norm{\mathbf{C} - \hat{\mathbf{C}}}_2/\norm{\mathbf{C}}_2$) for sampling fraction $p=0.1$, comparison between adaptive sampling weights (Alg.~\ref{alg-TNS}) and uniform sampling weights. 
}
\label{table:results-spectral-error-decay}
\renewcommand{\arraystretch}{1.10} 
\setlength{\tabcolsep}{8.0pt} 
\begin{tabularx}{0.7\linewidth}{@{}
r cc
HHHHH H
@{}
}
\toprule

\textsc{Temporal Network} & \textbf{Adaptive} \textbf{(Alg.~\ref{alg-TNS})} & \textbf{Uniform} &  \\
\midrule
\textsf{CollegeMsg} & 0.0797 & 0.0819 & \\
\textsf{ia-retweet-pol} & 0.1451 & 0.2723 & \\ 
\textsf{ia-contacts-dublin} & 0.0058 & 0.1624 & \\
\textsf{ia-facebook-wall} & 0.0335 & 0.0914 & \\
\textsf{ia-contacts-hyper} & 0.0009 & 0.0511 & \\
\textsf{SFHH-conf-sensor} & 0.0034 & 0.2283 & \\
\textsf{email-dnc} & 0.0143 & 0.1506 & \\
\bottomrule
\end{tabularx}
\vspace{-3mm}
\end{table}

In Tables~\ref{table:results-spectral-error} and~\ref{table:results-spectral-error-decay}, we show the relative spectral norm for online-TNS with adaptive and uniform sampling weights for no-decay and link-decay models respectively. The relative spectral norm is defined as $\norm{\mathbf{C} - \hat{\mathbf{C}}}_2/\norm{\mathbf{C}}_2$, where $\mathbf{C}$ is the exact time-dependent adjacency matrix of the input graph, whose entries represent the link strength, $\hat{\mathbf{C}}$ is the average estimated time-dependent adjacency matrix (estimated from the sample), and $\norm{\mathbf{C}}_2$ is the spectral norm of $\mathbf{C}$. The spectral norm $\norm{\mathbf{C} - \hat{\mathbf{C}}}_2$ is widely used for matrix approximations~\cite{achlioptas2013near}. $\norm{\mathbf{C} - \hat{\mathbf{C}}}_2$ measures the strongest linear trend of $\mathbf{C}$ not captured by the estimate $\hat{\mathbf{C}}$. The results show Online-TNS with adaptive sampling weights significantly outperforms the variant using uniform sampling weights, and captures the linear trend and structure of the data better, with an average $20\%$ improvement over the uniform sampling weights. We also measured the error using relative Frobenius norm (\ie, $\norm{\mathbf{C} - \hat{\mathbf{C}}}_F/\norm{\mathbf{C}}_F$) and observed similar conclusions. 

\subsubsection{Temporally Weighted Motif Estimation}
\label{sec:exp-temporal-motifs}
Recall that our formulation of temporal motif differs from previous work in that instead of counting motifs that occur within some time period $\delta$, our formulation focuses on counting temporally weighted motifs where the motifs are weighted such that motifs that occur more recent and contain active links are assigned larger weight than those occurring in the distant past. 
This formulation is clearly more useful and important, since it can capture the evolution of the network and relationships at a continuous-time scale. Also, this formulation would be useful for many practical applications involving prediction and forecasting since it appropriately accounts for temporal statistics (in this case, motifs) that occur more recently, which are by definition more predictive of some future event. 
In Table~\ref{table:results-motif-est-with-and-without-decay}, we show results for estimating the temporally weighted motif counts. For brevity, we only show results for triangle motifs (both decay and no-decay models), but the proposed framework and unbiased estimators in Algorithm~\ref{alg-TNS} and Section~\ref{sec:topology-adaptive-est} generalize to any network motifs of larger size. 
For these results, we set the decay factor $\delta$ to $30$ days. Notably, all of the temporally decayed motif count estimates have a relative error that is less than $0.03$ as shown in Table~\ref{table:results-motif-est-with-and-without-decay} (Adaptive). 
Nevertheless, this demonstrates that our efficient temporal sampling framework is able to leverage accurate estimators for even the smallest sample sizes. Table~\ref{table:results-motif-est-with-and-without-decay} also shows results of Algorithm~\ref{alg-TNS} with uniform sampling weights. Overall, we observe that Online-TNS with adaptive sampling weights generally outperforms Online-TNS with uniform sampling weights. We conjecture that Online-TNS with adaptive sampling is general and would be useful in various applications beyond the scope of this paper. In particular, for applications that require importance sampling with the ability to combine both topology (\eg, edge multiplicity, temporal strength, subgraphs) and auxiliary information (\eg, node/edge attributes and features). We will explore these applications in future work.

\begin{table*}[t!]
\centering
\caption{%
Results for temporally weighted motif count estimation. 
Online TNS with adaptive weights compared to online TNS with uniform weights.
Relative error for triangle counts are reported using $p=0.1$.
}
\label{table:results-motif-est-with-and-without-decay}
\renewcommand{\arraystretch}{1.20} 
\small 
\setlength{\tabcolsep}{6.3pt} 
\begin{tabularx}{0.90\linewidth}{@{}
r H 
XHXX
XHXX
HHHHH H
@{}
}
\toprule
& & \multicolumn{4}{c}{\sc Without Decay} & \multicolumn{4}{c}{\sc With Decay} \\
\cmidrule(l{1pt}r{5pt}){3-6}
\cmidrule(l{4pt}r{2pt}){7-10}

\textsc{Temporal Network} & $|E_T|$ & 
\textbf{Exact} & 
\textbf{est.} & 
\textbf{Adaptive} &
\textbf{Uniform} & 
\textbf{Exact} & 
\textbf{est.} & 
\textbf{Adaptive} &
\textbf{Uniform} 
\\
\midrule

\textbf{sx-stackoverflow} & 48M & 15B & 15B & 0 & 0.0133 & 168M & 168M & 0.00004 &  0.0006 \\
\textbf{ia-facebook-wall} & 856k & 435M & 435M & 0.0004 & 0.0605 & 9.9M & 9.9M & 0.0004 &  0.0247 \\
\textbf{wiki-talk} & 6.1M & 12B & 12B & 0.0003 & 0.0171 & 394M & 394M & 0.0001 & 0.0022 \\
\textbf{CollegeMsg} & 60k & 6.2M & 6.1M & 0.0148 & 0.0545 & 2.0M & 2.0M & 0.0003 & 0.0277 \\
\textbf{ia-retweet-pol} & 61k & 380k & 371k & 0.0236 & 0.0341 & 147k & 147k & 0.0001 & 0.0247 \\
\textbf{ia-prosper-loans} & 3.4M & 1.4M & 1.4M & 0.0056 & 0.0078 & 232k & 230k & 0.0067 & 0.0048 \\
\textbf{comm-linux-reply} & 984k & 148B & 148B & 0 & 0.0055 & 242M & 242M & 0.00003 & 0.000024\\
\textbf{email-dnc} & 37k & 483M & 483M & 0 & 0.0004 & 251M & 251M & 0.00006 & 0.0097 \\
\textbf{ia-enron-email} & 1.1M & 14B & 14B & 0 & 0.0726 & 329M & 329M & 0.0003 & 0.0034 \\
\textbf{ia-contacts-dublin} & 416k & 382M & 382M & 0.00005 & 0.00008 & 381M & 381M & 0 & 0.0001 \\
\textbf{fb-forum} & 34k & 3.3M & 3.3M & 0.0036 & 0.1761 & 763k & 758k & 0.0067 & 0.0137 \\
\textbf{ia-contacts-hyper09} & 21k & 93M & 93M & 0 & 0.00041 & 88M & 88M & 0 & 0.00086\\
\textbf{SFHH-conf-sensor} & 70k & 622M & 622M & 0 & 0.0665 & 604M & 604M & 0.0002 & 0.0434 \\
\textbf{sx-superuser} & 1.1M & 83M & 82M & 0.0072 & 0.0137 & 2.1M & 2.1M & 0.0013 & 0.0016 \\
\textbf{sx-askubuntu} & 727k & 71M & 70M & 0.0035 & 0.0077 & 2.7M & 2.7M & 0.0069 & 0.0135 \\
\textbf{sx-mathoverflow} & 390k & 269M & 269M & 0.0008 & 0.0688 & 2.8M & 2.8M & 0.0004 & 0.0003 \\
\bottomrule
\end{tabularx}
\end{table*}

\subsubsection{Sensitivity Analysis of the Decay Factor}
\label{sec:exp-delta-analysis}
We now study the impact of the choice of the decay factor $\delta$ on the quality of the estimates. When we choose the value for the mean lifetime $\delta$, it is more intuitive to think about the half-life $\eta_{1/2}$ of an edge. The half-life of an edge gives the amount of time for an edge to lose half of its weight/strength in the absence of new interactions. Given $\delta > 0$, the half-life of an edge is defined as $\eta_{1/2} = \delta \ln 2$. As such the choice of $\delta$ is crucial to filter-out and down-weight the old activity in continuously evolving networks. When the half-life is short (\ie, the mean lifetime $\delta$ is short), the interactions result in weak links among nodes, where the link strength dies off quickly unless the interactions occur more frequently and sustainably among the nodes. On the other hand, when the half-life is long (\ie, $\delta$ is long), links are able to build a momentum and strengthen from interactions that otherwise occurred too far in time. In Figure~\ref{fig:edge-weight-delta}, we show the estimated link strengths obtained using Online-TNS with adaptive weights for two different decay factors $\delta = 1$-day and $\delta = 7$-days. Clearly, the range and scale of the link strength is much higher when $\delta$ is long. In addition, we also observe that the estimated distributions of link strengths for the top-k edges are nearly indistinguishable from the exact distributions. This is due to the unbiasedness property of the proposed estimators, as the unbiasedness property holds regardless the choice of the decay parameter $\delta$. In Table~\ref{table:results-delta}, we provide the temporally weighted motif count estimation obtained using Online TNS (Alg~\ref{alg-TNS}) with adaptive weights and different decay $\delta$ parameters. Clearly, the relative error is small for $\delta$ values, which is a result of the unbiasedness property of the proposed estimators.

\begin{table*}[h!]
\centering
\caption{%
Sensitivity analysis of the decay factor $\delta$. Results for temporally weighted motif count estimation. 
Relative error reported using $p=0.1$ for triangle counts using Online TNS (Alg~\ref{alg-TNS}) with adaptive weights and different decay $\delta$ parameters.
}
\label{table:results-delta}
\renewcommand{\arraystretch}{1.05} 
\setlength{\tabcolsep}{10.0pt} 
\begin{tabularx}{0.7\linewidth}{
r rrr
}
\toprule
& \multicolumn{3}{c}{\sc Decay Factor $\delta$} \\
\cmidrule(l{0pt}r{5pt}){2-4}

\textsc{Temporal Network} & \textbf{1-day} & \textbf{1-week} & \textbf{1-month}  \\

\midrule
\textsf{CollegeMsg} & 0.002 & 0.003 & 0.0003  \\
\textsf{ia-retweet-pol} & 0.074 & 0.023 & 0.0001  \\
\textsf{ia-contacts-dublin} & 0.00003 & 0 & 0  \\
\textsf{ia-facebook-wall} & 0.0091 & 0.0011 & 0.0004  \\
\textsf{email-dnc} & 0.00004 & 0.00005 & 0.00006  \\
\bottomrule
\end{tabularx}
\end{table*}
\renewcommand{\figsz}{0.28}
\begin{figure*}[h!]
\centering
\subfigure
{\includegraphics[width=\figsz\linewidth]{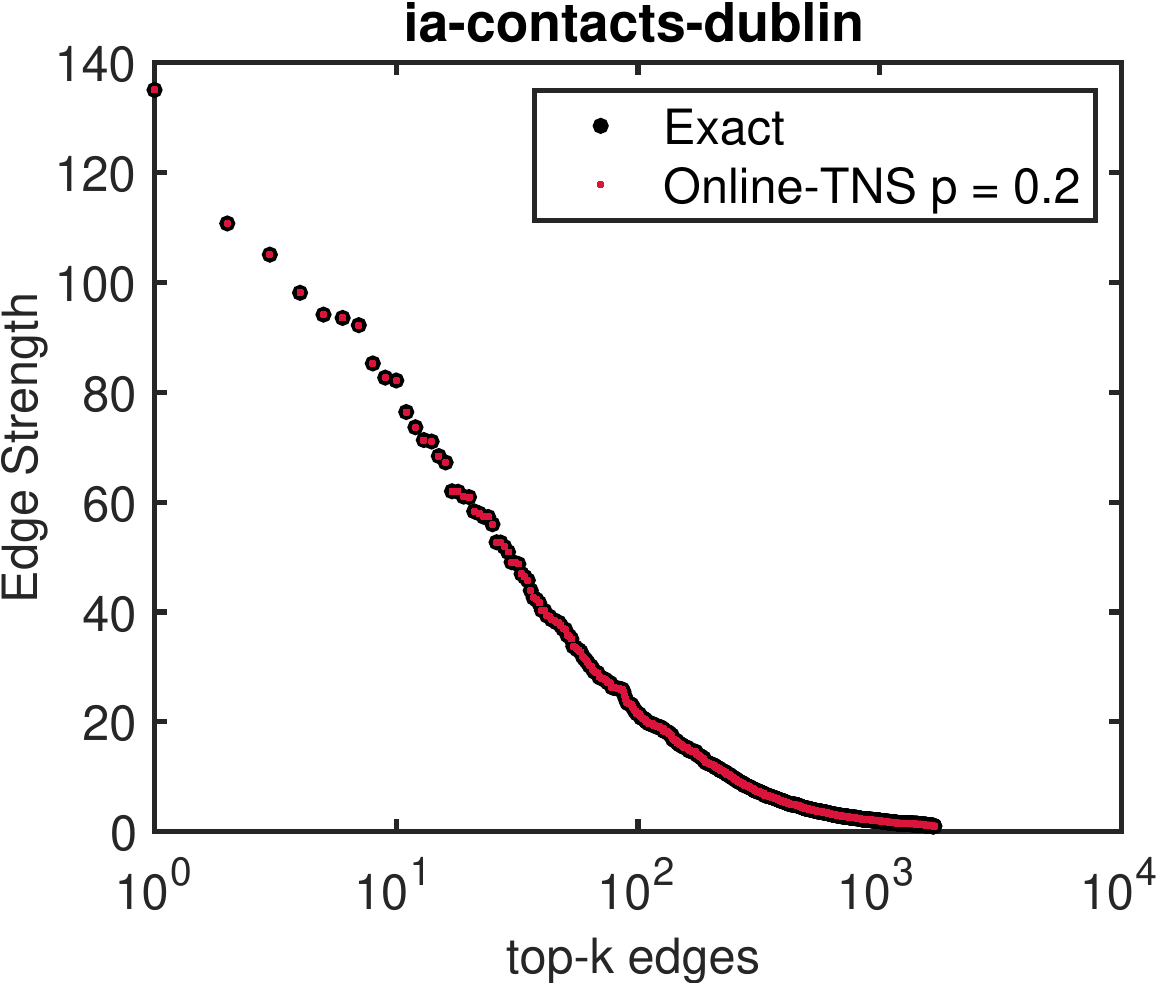}}
\hspace{3mm}
\subfigure
{\includegraphics[width=\figsz\linewidth]{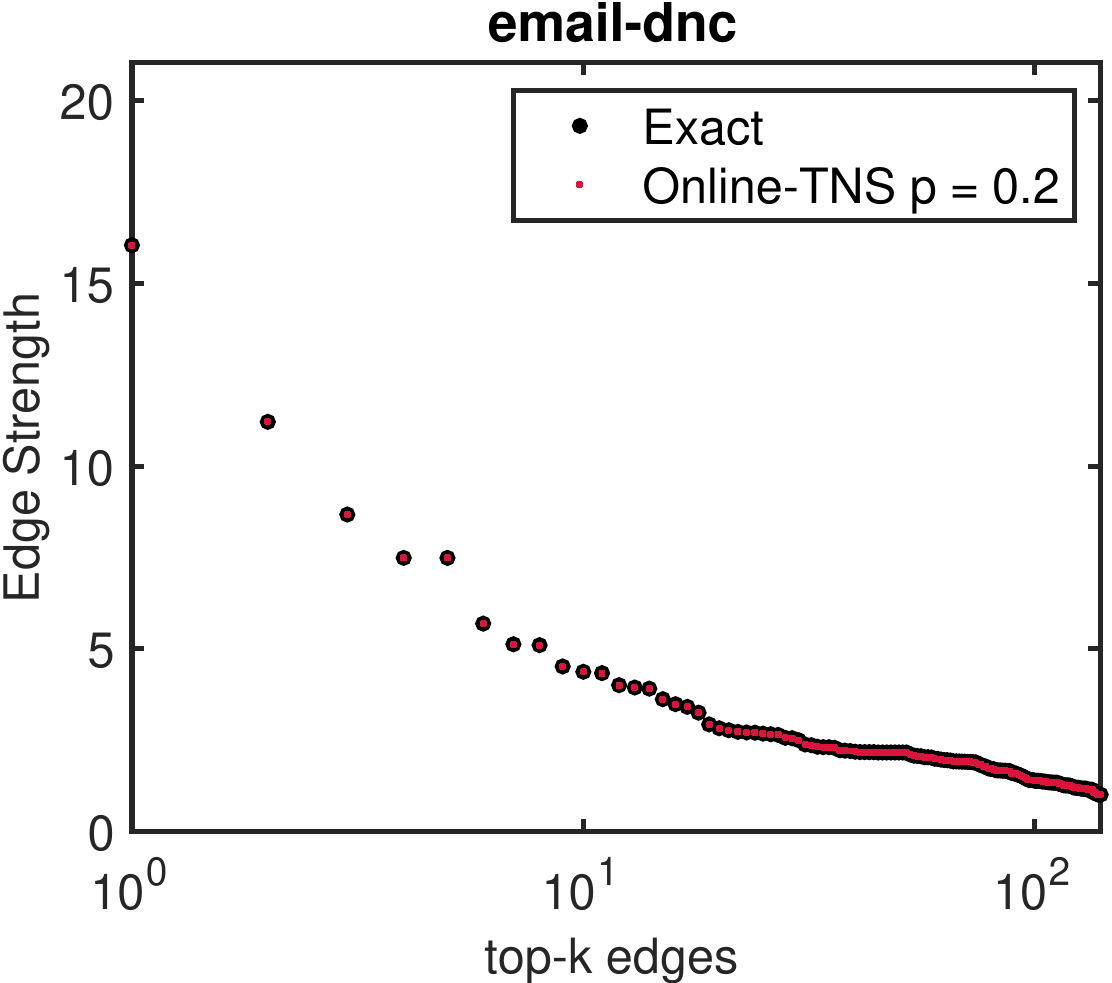}}
\hspace{3mm}
\subfigure
{\includegraphics[width=\figsz\linewidth]{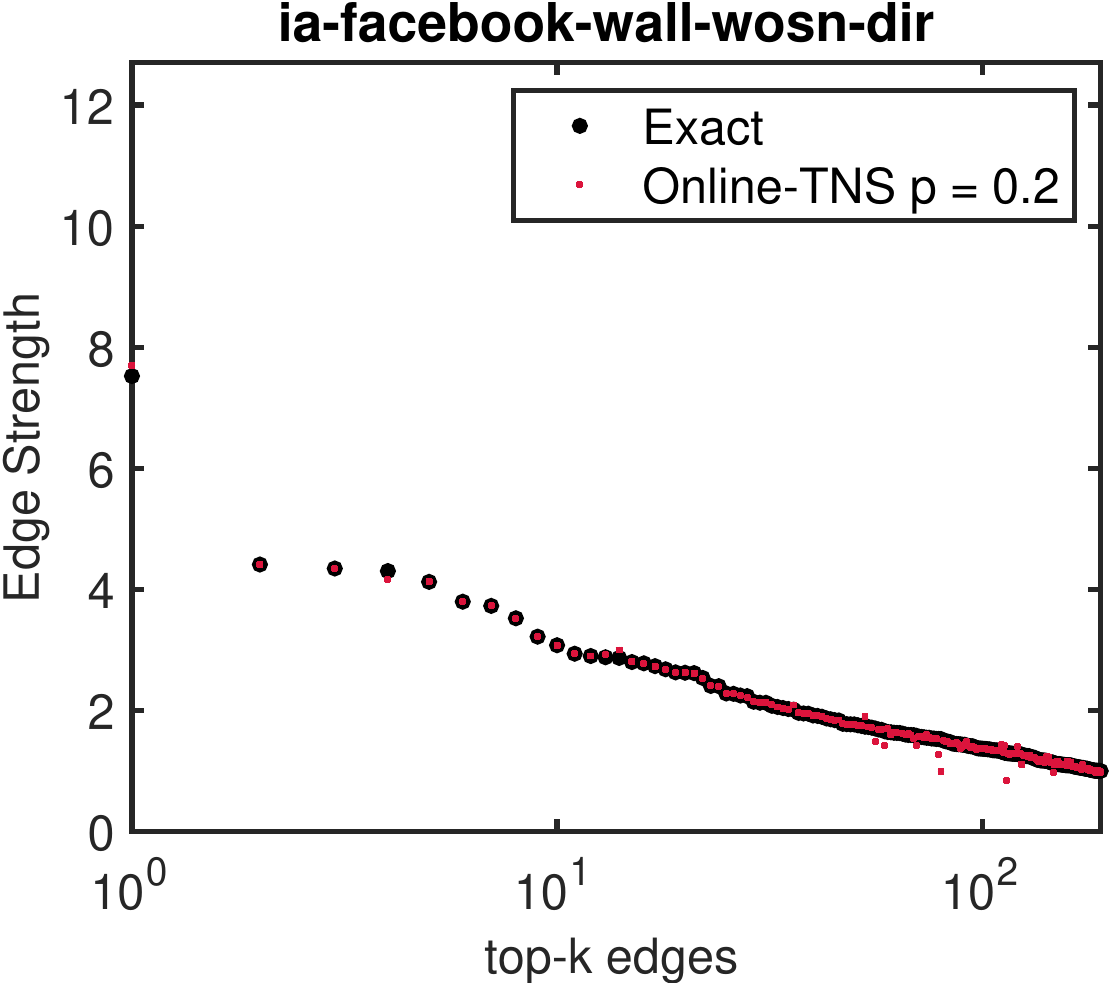}}
\vspace{-3mm}
\subfigure
{\includegraphics[width=\figsz\linewidth]{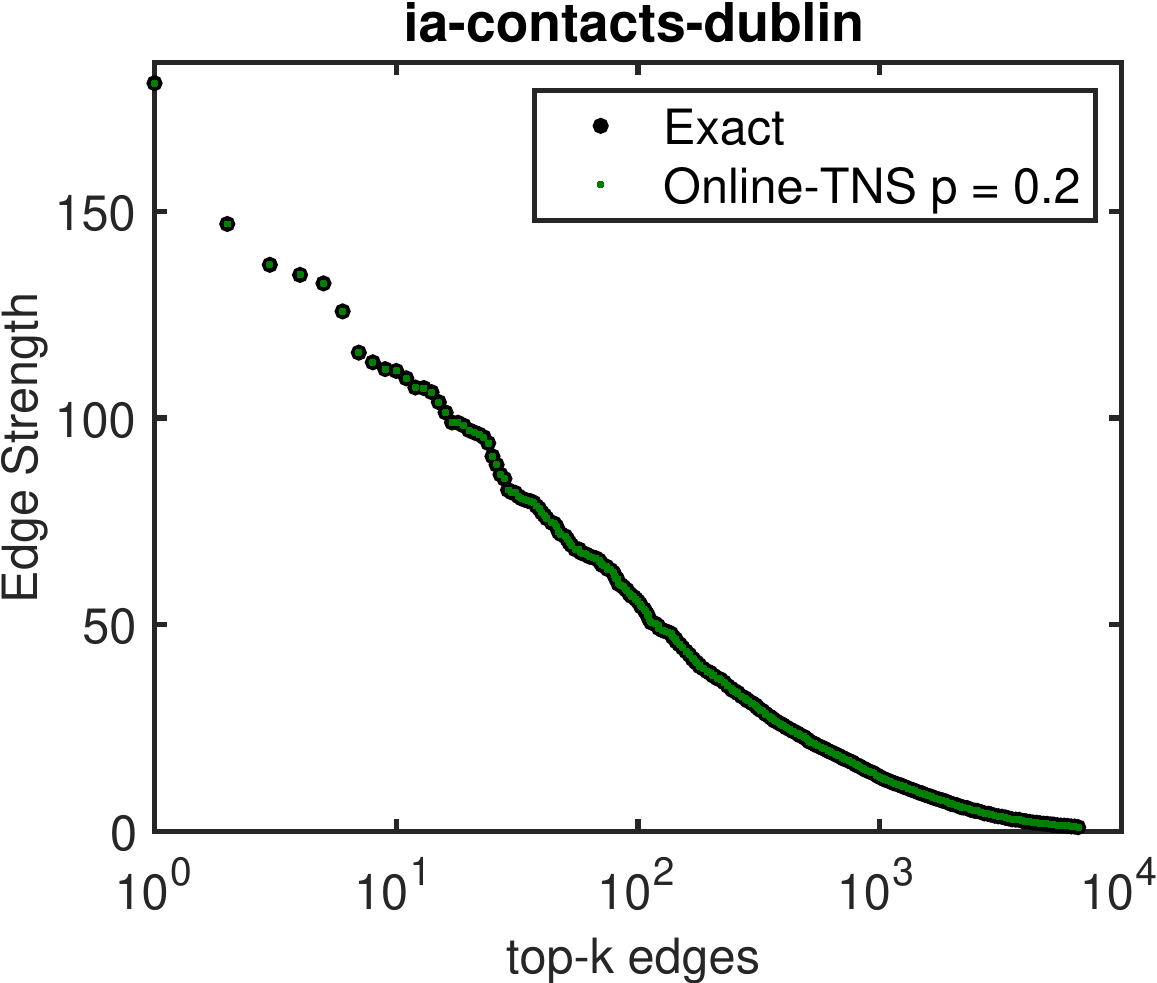}}
\hspace{3mm}
\subfigure
{\includegraphics[width=\figsz\linewidth]{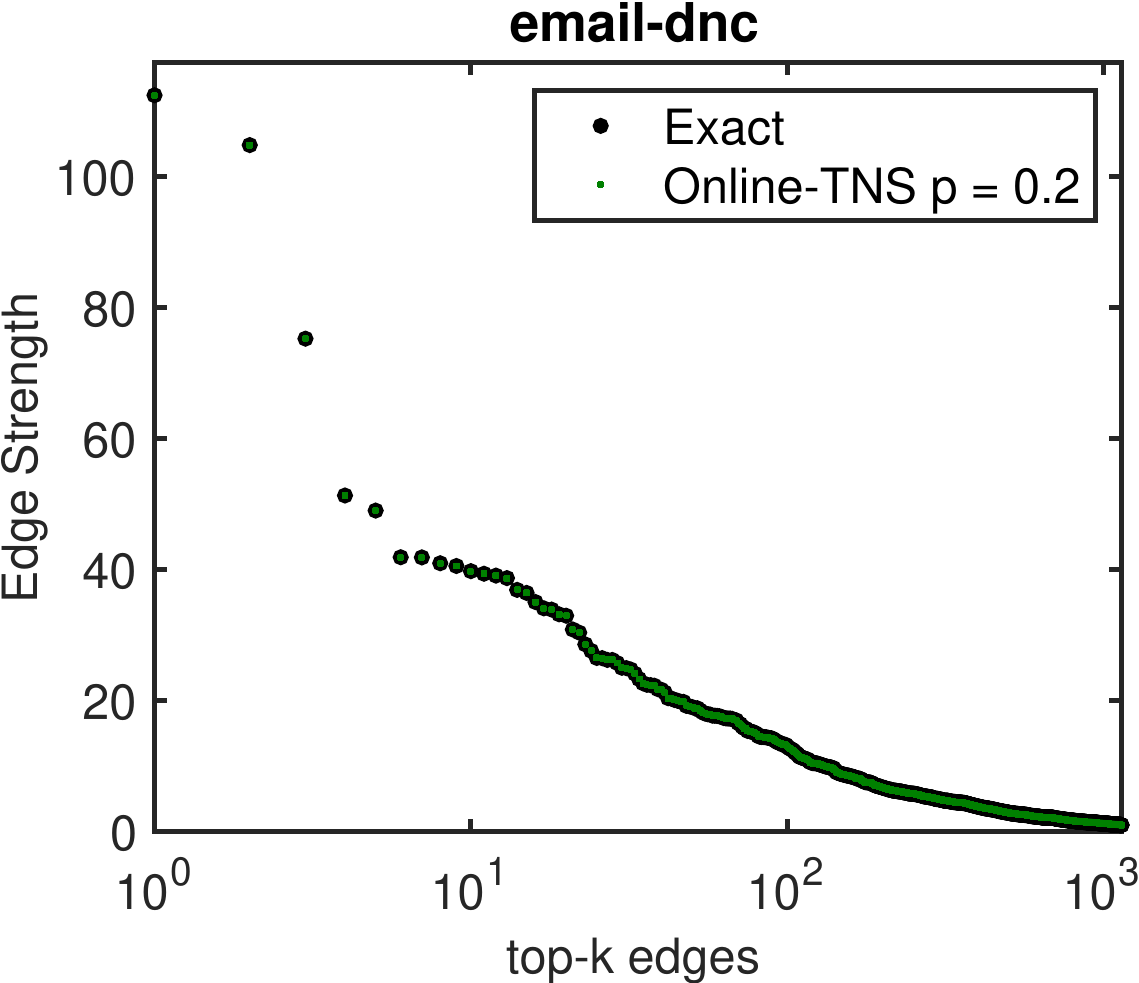}}
\hspace{3mm}
\subfigure
{\includegraphics[width=\figsz\linewidth]{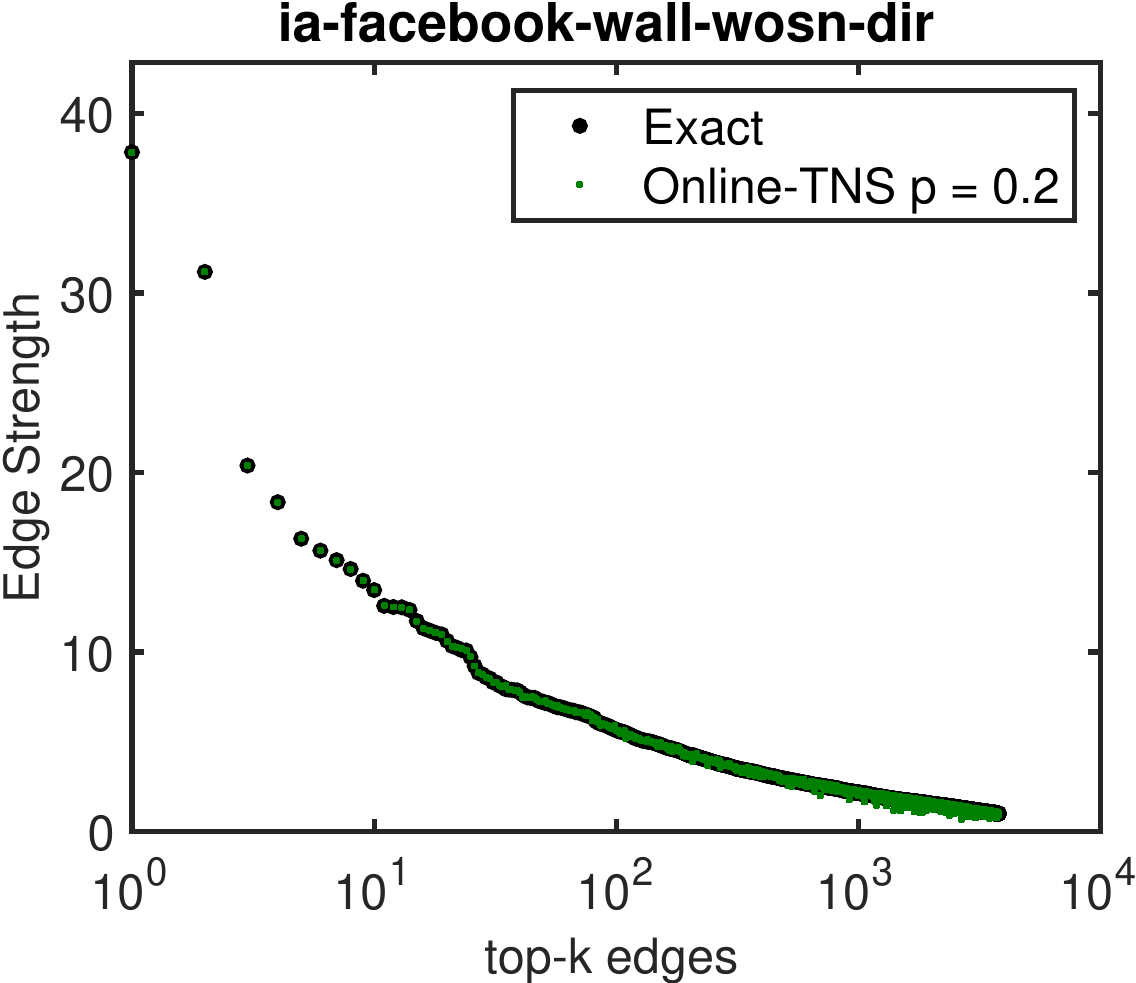}}

\caption{Temporal link strength estimated distribution vs exact distribution for top-k links.
Results are shown for sampling fraction $p=0.2$. (Top: Red) Results for
Online-TNS Algorithm~\ref{alg-TNS} with adaptive sampling weights and decay factor $\delta = 1$ day. 
(Bottom: Green) Results for
Online-TNS Algorithm~\ref{alg-TNS} with adaptive sampling weights~\ref{alg-TNS} and decay factor $\delta = 7$ days. 
}
\label{fig:edge-weight-delta}
\end{figure*}

\begin{table*}[h!]
\centering
\caption{%
Results for estimating temporal burstiness.
For each temporal network, we show the estimated burstiness using different sampling probabilities (first row) compared to the exact.
The relative error $\nicefrac{|\hat{B}-B|}{B}$ of the estimates is also shown.
}

\renewcommand{\arraystretch}{1.10} 
\renewcommand{\arraystretch}{1.05} 
\setlength{\tabcolsep}{7.0pt} 
\begin{tabularx}{0.95\linewidth}{
r rrrrr r 
HHHHH H
@{}
}
\toprule
& \multicolumn{5}{c}{\sc Sampling Fraction} & \\
\cmidrule(l{0pt}r{5pt}){2-6}

\textsc{Temporal Network} & \textbf{0.1} & \textbf{0.2} & \textbf{0.3} & \textbf{0.4} & \textbf{0.5} & \textbf{Exact} \\

\midrule
\textsf{wiki-talk} & 0.6196 & 0.6208 & 0.6208 & 0.6207 & 0.6206 & 0.6206 \\
\emph{(error)} & 0.0017   & 0.0003  &  0.0003  &  $<$0.0001   & $<$0.0001\\

\midrule
\textsf{ia-facebook-wall-wosn} & 0.4482 & 0.4534 & 0.4535 & 0.4535 & 0.4535 & 0.4535 \\ 
\emph{(error)} & 0.0116  &  0.0002  &  $<\!\!10^{-5}$  &  $<\!\!10^{-5}$  &  $<\!\!10^{-5}$ \\

\midrule
\textsf{bitcoin} & 0.7738 & 0.7642 & 0.7606 & 0.7586 & 0.7579 & 0.7576 \\
\emph{(error)} & 0.0214  &  0.0087  &  0.0040  &  0.0013  &  0.0004 \\

\midrule
\textsf{sx-stackoverflow} & 0.6517 & 0.6712 & 0.6808 & 0.6863 & 0.6891 &  0.6898 \\
\emph{(error)} & 0.0552   & 0.0269  &  0.0130  &  0.0050  &  0.0010 \\

\bottomrule
\end{tabularx}
\label{table:results-burstiness}
\end{table*}

\subsection{Estimation of Temporal Statistics} \label{sec:exp-temporal-network-statistics}
While the proposed framework can be used to obtain unbiased estimates of arbitrary temporal network statistics, we focus in this section on two important temporal properties and their distributions
including burstiness~\cite{barabasi2005origin} and temporal link persistence~\cite{clauset2012persistence}. 
For a survey of other important temporal network statistics that are applicable for estimation using the framework, see~\cite{holme2012temporal}.

\subsubsection{Burstiness}
Burstiness $B$ is widely used to characterize the link activity in temporal networks~\cite{holme2012temporal}. Burstiness is computed using the mean $\mu$ and standard deviation $\sigma$ of the distribution of same-edge inter-contact times collected from all links, \ie, $B = (\sigma - \mu)/(\sigma + \mu)$. The inter-contact time is the elapsed time between two subsequent same-edge interactions (\ie, time between two text messages from the same pair of friends). Burstiness measures the deviation of relationship activity from a Poisson process.     
In Table~\ref{table:results-burstiness}, we use the proposed framework to estimate burstiness (\ie, computed using the sampled network).
We show the estimated burstiness for sampling fraction $p \in \{0.1,0.2,0.3,0.4,0.5\}$.
In addition, we also provide the relative error of the estimates across the different sampling fractions.
From Table~\ref{table:results-burstiness}, we observe the relative errors are small and the estimates are shown to converge as the sampling fraction $p$ increases. In Figure~\ref{fig:intercontact-time}, we show the exact and estimated distribution of inter-contact times for sampling fractions $p=0.1$ (top row) and $p=0.2$ (bottom row). We observe that the estimated distribution from the sample accurately captures the exact distribution.

\subsubsection{Temporal Link Persistence}
The persistence of an edge measures the lifetime of relationships, and is computed as the elapsed time between the first interaction and the last interaction of the same edge~\cite{holme2012temporal}. Let $L$ denote the average link persistence (or lifetime) computed over all edges in the full (sampled) network defined as $L = \frac{1}{|K|} \sum_{(i,j) \in K} \tau_{ij}^{(\text{last})} - \tau_{ij}^{(\text{first})}$. Relative error of estimated link persistence is shown in Table~\ref{table:results-persistence}. In Figure~\ref{fig:persistence}, we show the exact and estimated distribution (\ie, computed using the sampled network) of link persistence scores for sampling fractions $p=0.1$ (top row) and $p=0.2$ (bottom row). We observe that the estimated distributions from the sampled network across all graphs accurately captures the exact distribution (for both burstiness and persistence). 
We also note that the proposed algorithm Alg~\ref{alg-TNS} can also handle labeled graphs (with vertex/edge categorical variables), where the estimators are computed separately for each possible label combination in a stratified fashion.
\begin{table*}[h!]
\centering
\caption{%
Estimation results for temporal persistence.
For each temporal network, we report relative error $\nicefrac{|\hat{L}-L|}{L}$ of the estimates using different sampling fractions.
}
\label{table:results-persistence}
\renewcommand{\arraystretch}{1.10} 
\setlength{\tabcolsep}{12.0pt} 
\begin{tabularx}{0.96\linewidth}{
r rrrrr r 
HHHHH H
@{}
}
\toprule
& \multicolumn{5}{c}{\sc Sampling Fraction} & \\
\cmidrule(l{0pt}r{1pt}){2-6}

\textsc{Temporal Network} & \textbf{0.1} & \textbf{0.2} & \textbf{0.3} & \textbf{0.4} & \textbf{0.5} \\
\midrule
\textbf{wiki-talk} & 0.1380 & 0.0412 & 0.0112 & 0.0009 & $<\!\!10^{-6}$ \\
\textbf{ia-facebook-wall-wosn} & 0.1232 & 0.0023 & $<\!\!10^{-7}$ & $<\!\!10^{-7}$ & $<\!\!10^{-7}$ \\
\textbf{sx-stackoverflow} & 0.1718 & 0.0794 & 0.0357 & 0.0119 & 0.0016 \\
\textbf{bitcoin} & 0.1056 & 0.0207 & 0.0023 & 0.0020 & 0.0024 \\
\bottomrule
\end{tabularx}
\end{table*}

\renewcommand{\figsz}{0.22}
\begin{figure*}[h!]
\centering
\subfigure
{\includegraphics[width=\figsz\linewidth]{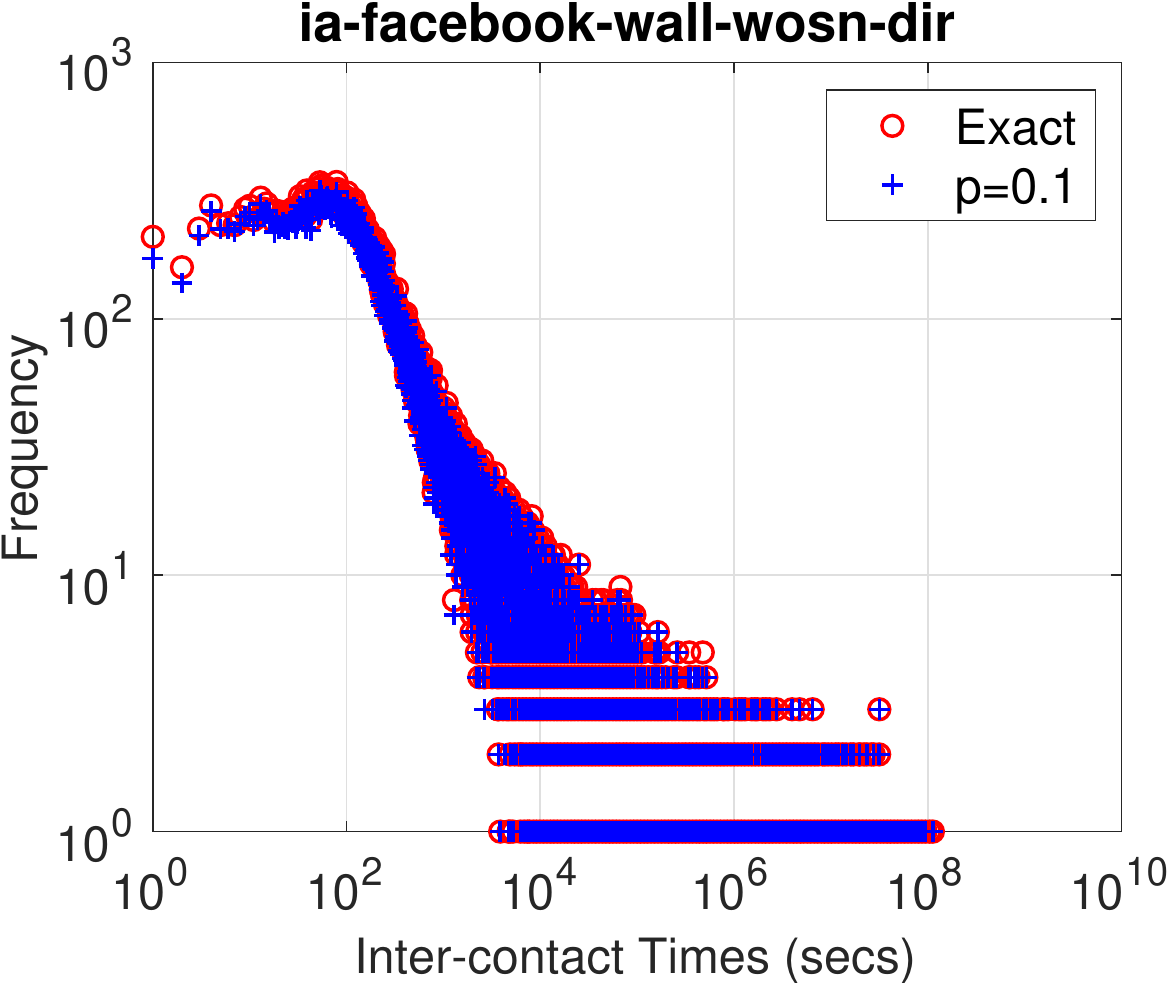}}
\subfigure
{\includegraphics[width=\figsz\linewidth]{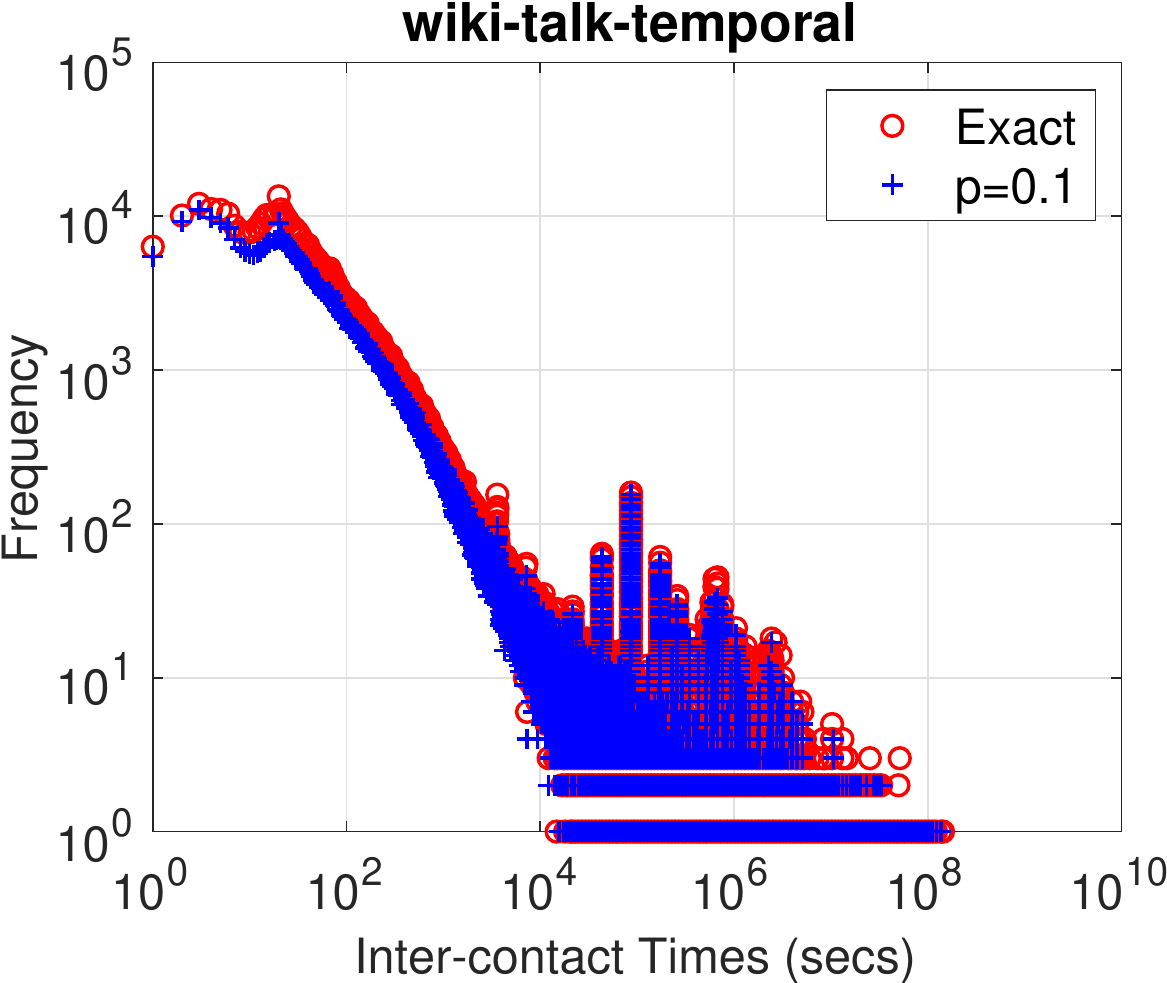}}
\subfigure
{\includegraphics[width=\figsz\linewidth]{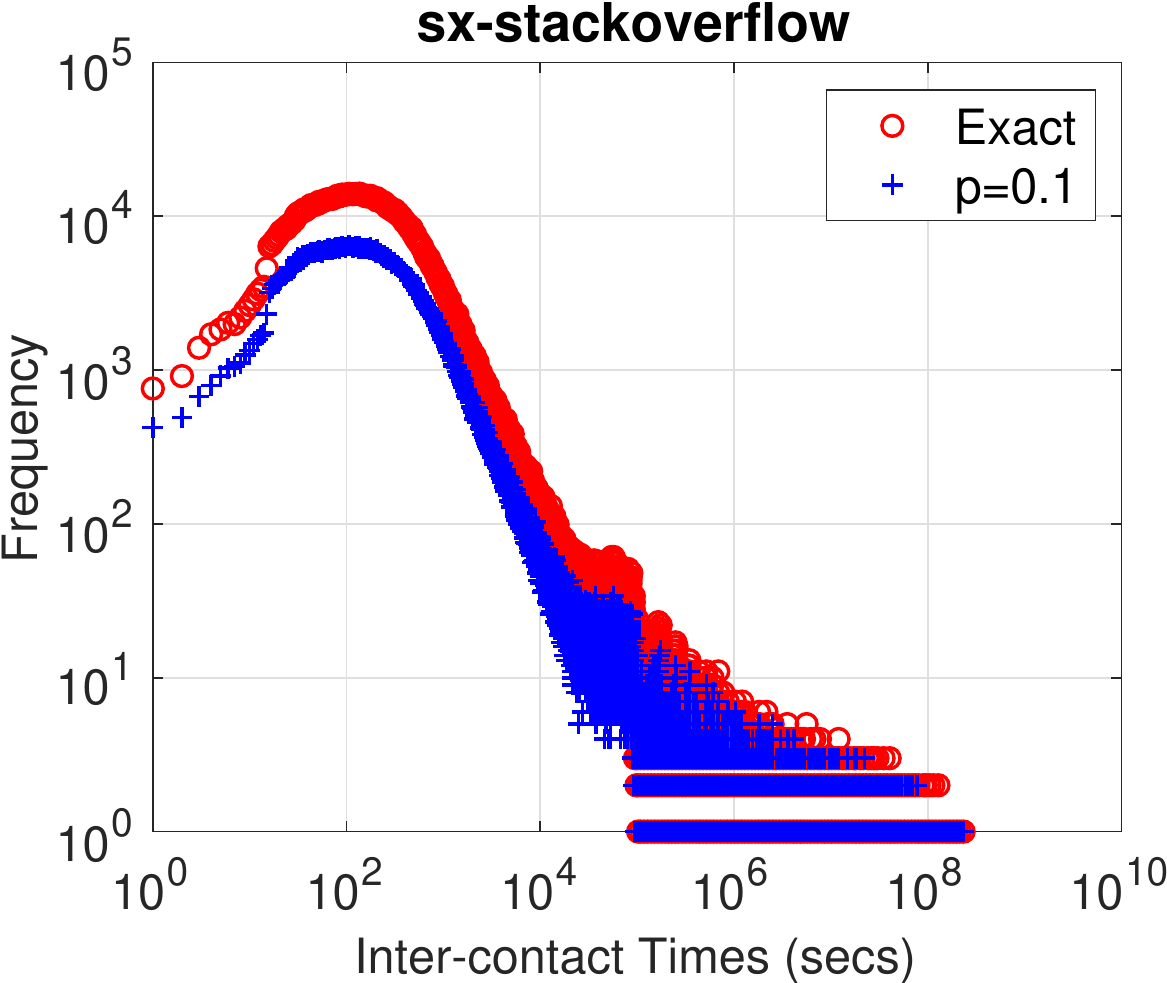}}
\subfigure
{\includegraphics[width=\figsz\linewidth]{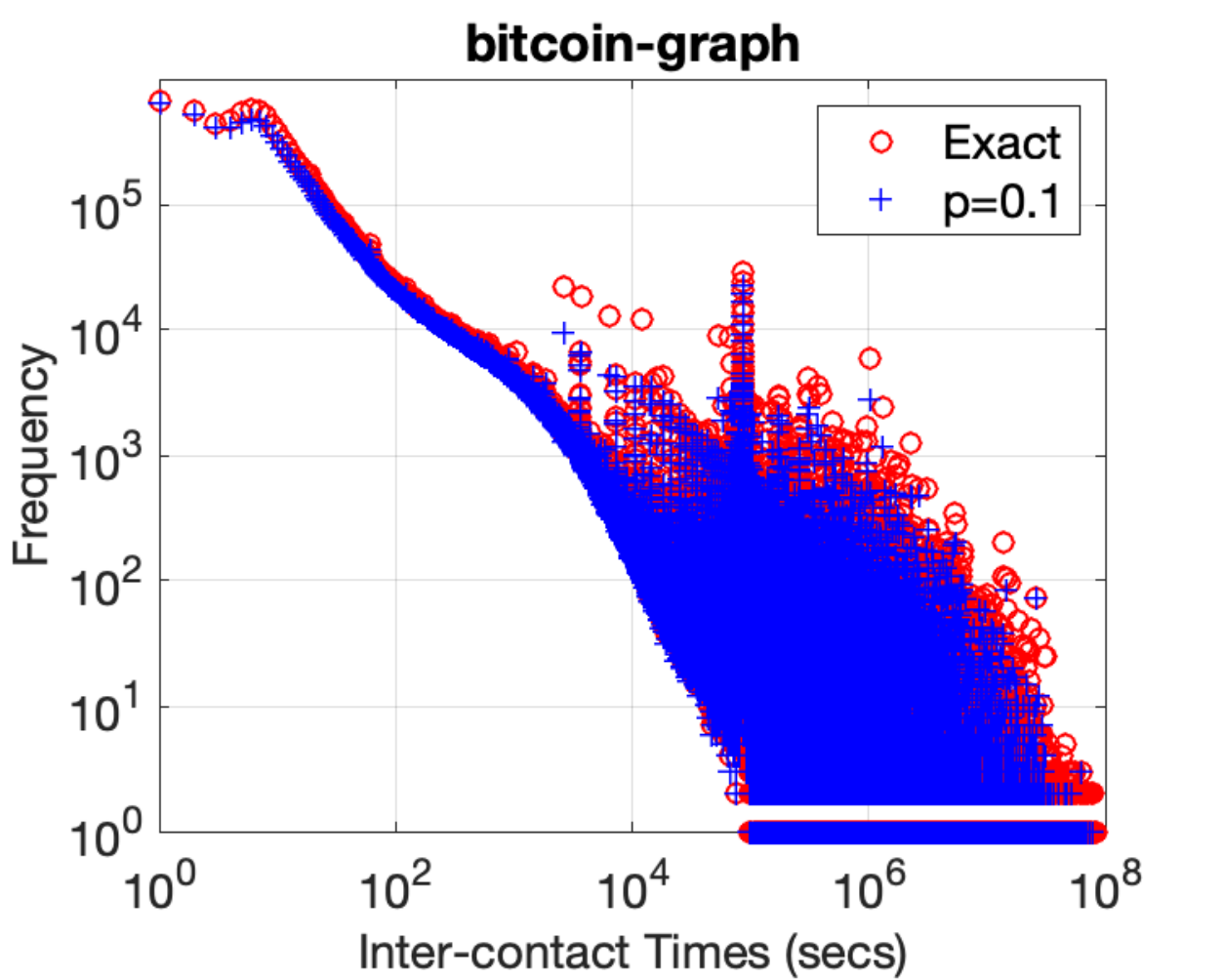}}

\subfigure
{\includegraphics[width=\figsz\linewidth]{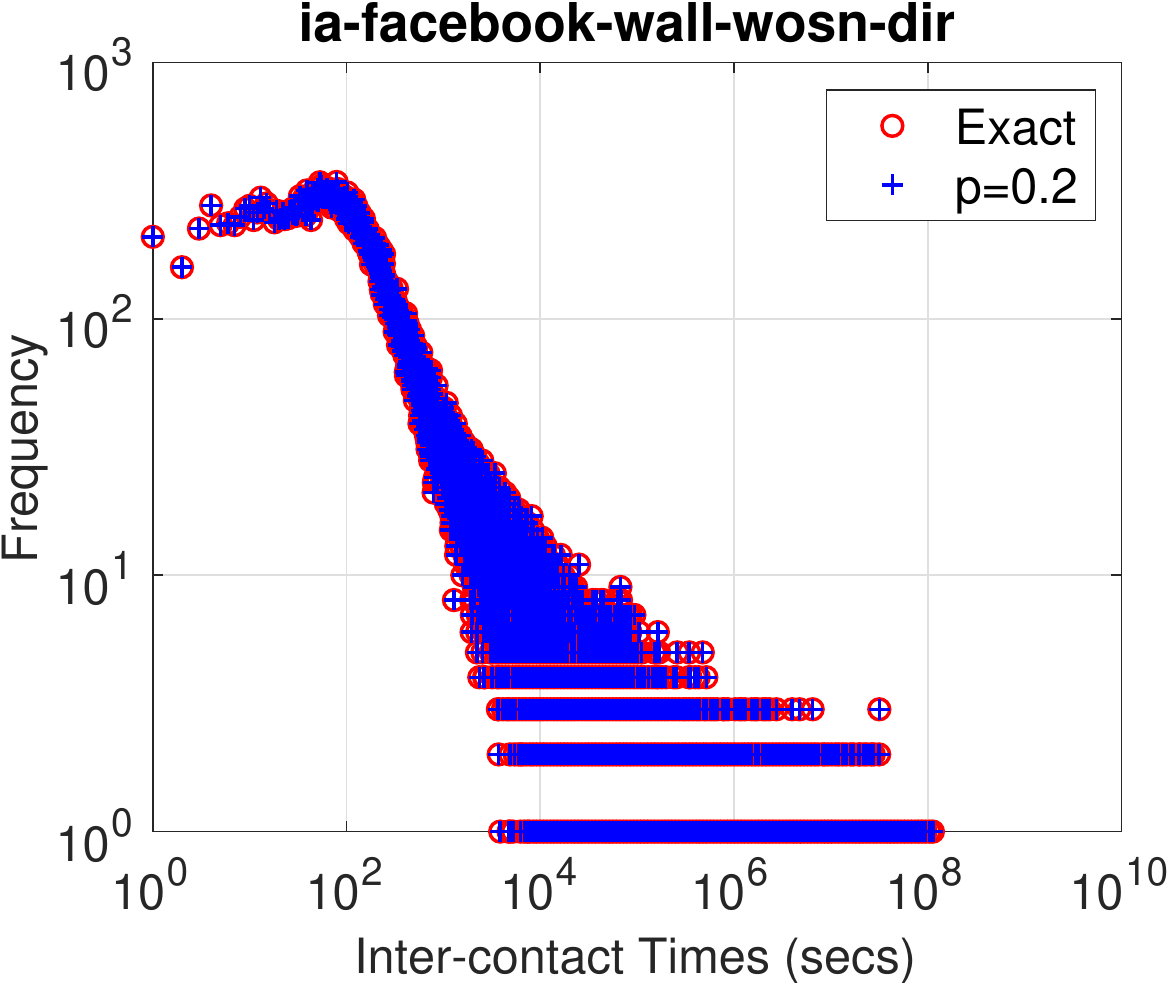}}
\subfigure
{\includegraphics[width=\figsz\linewidth]{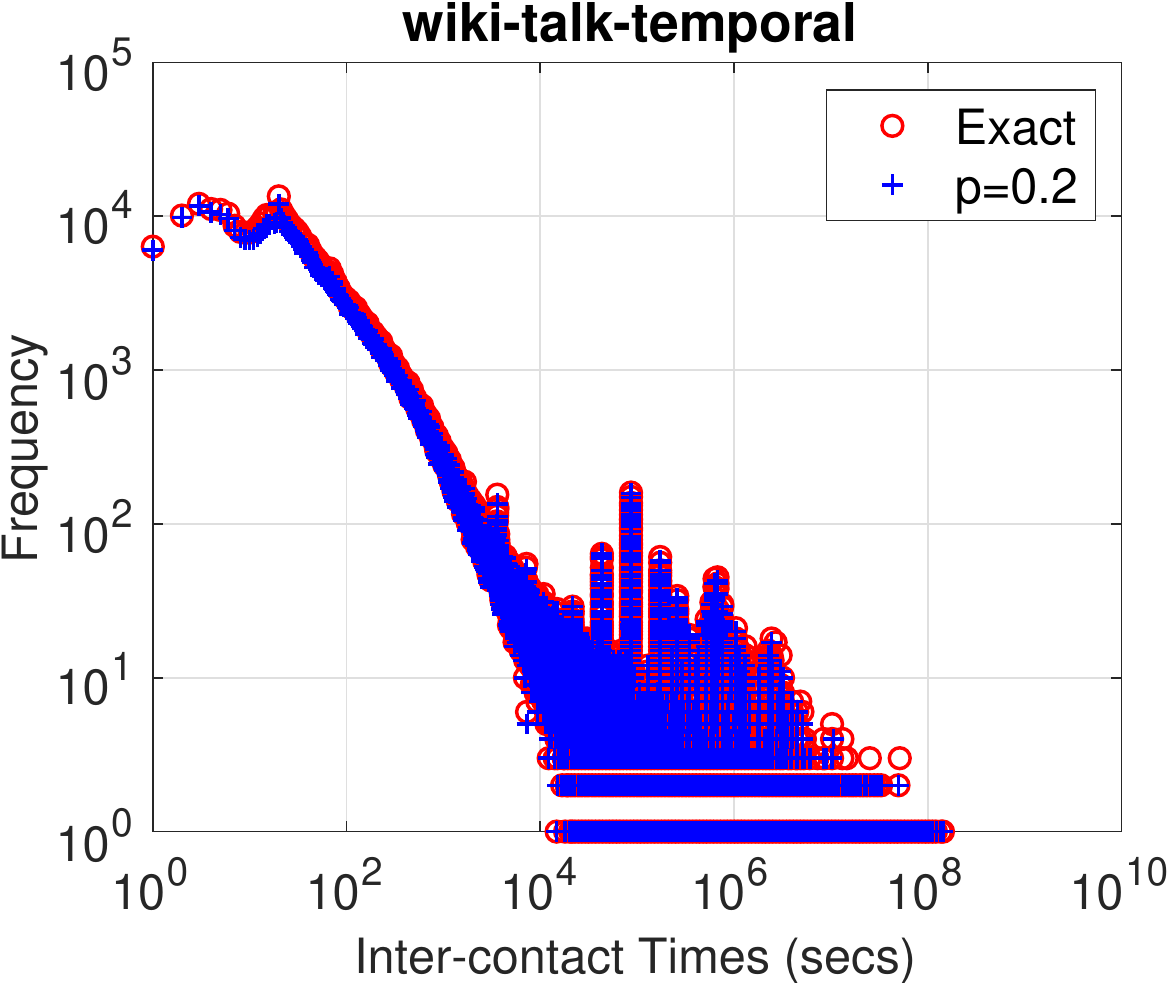}}
\subfigure
{\includegraphics[width=\figsz\linewidth]{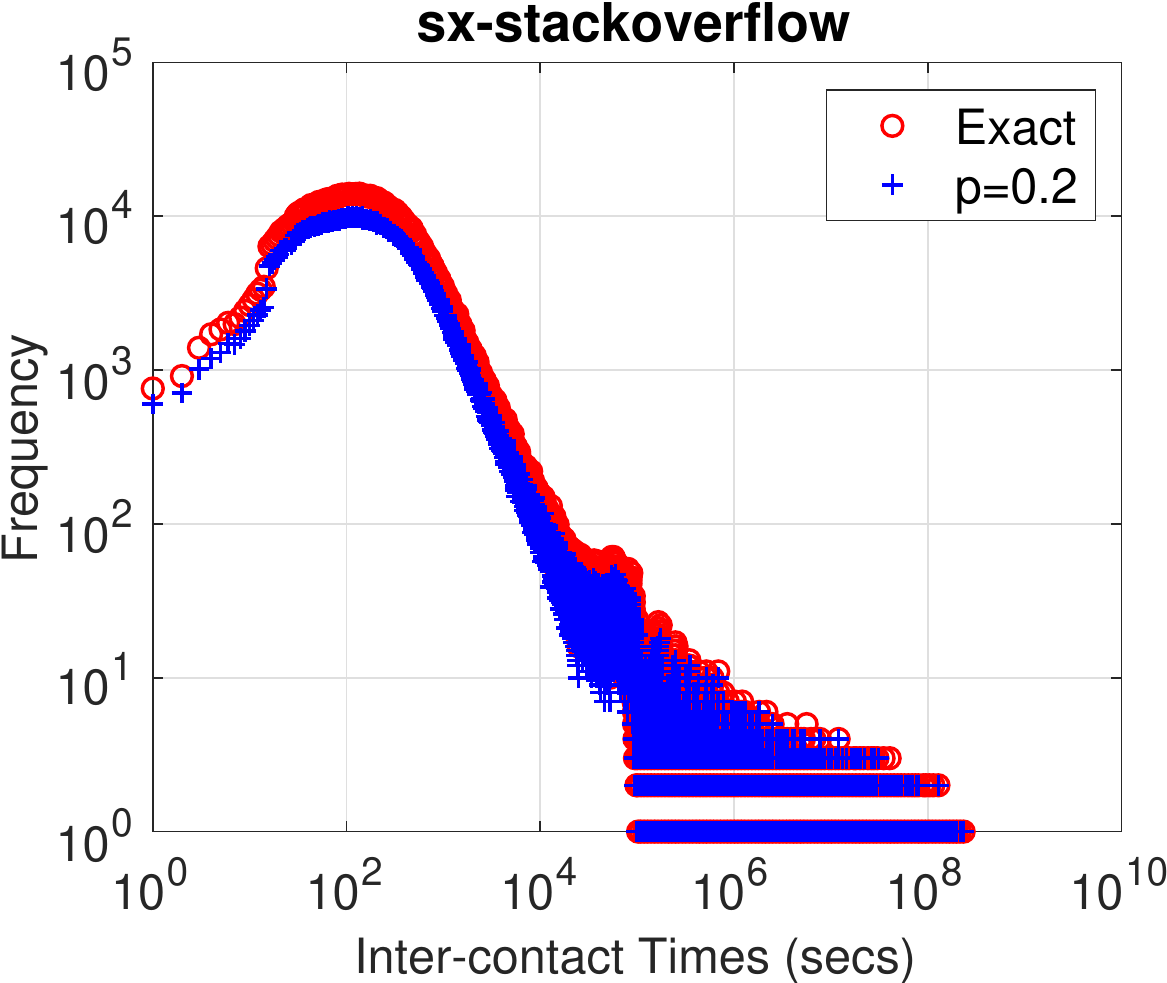}}
\subfigure
{\includegraphics[width=\figsz\linewidth]{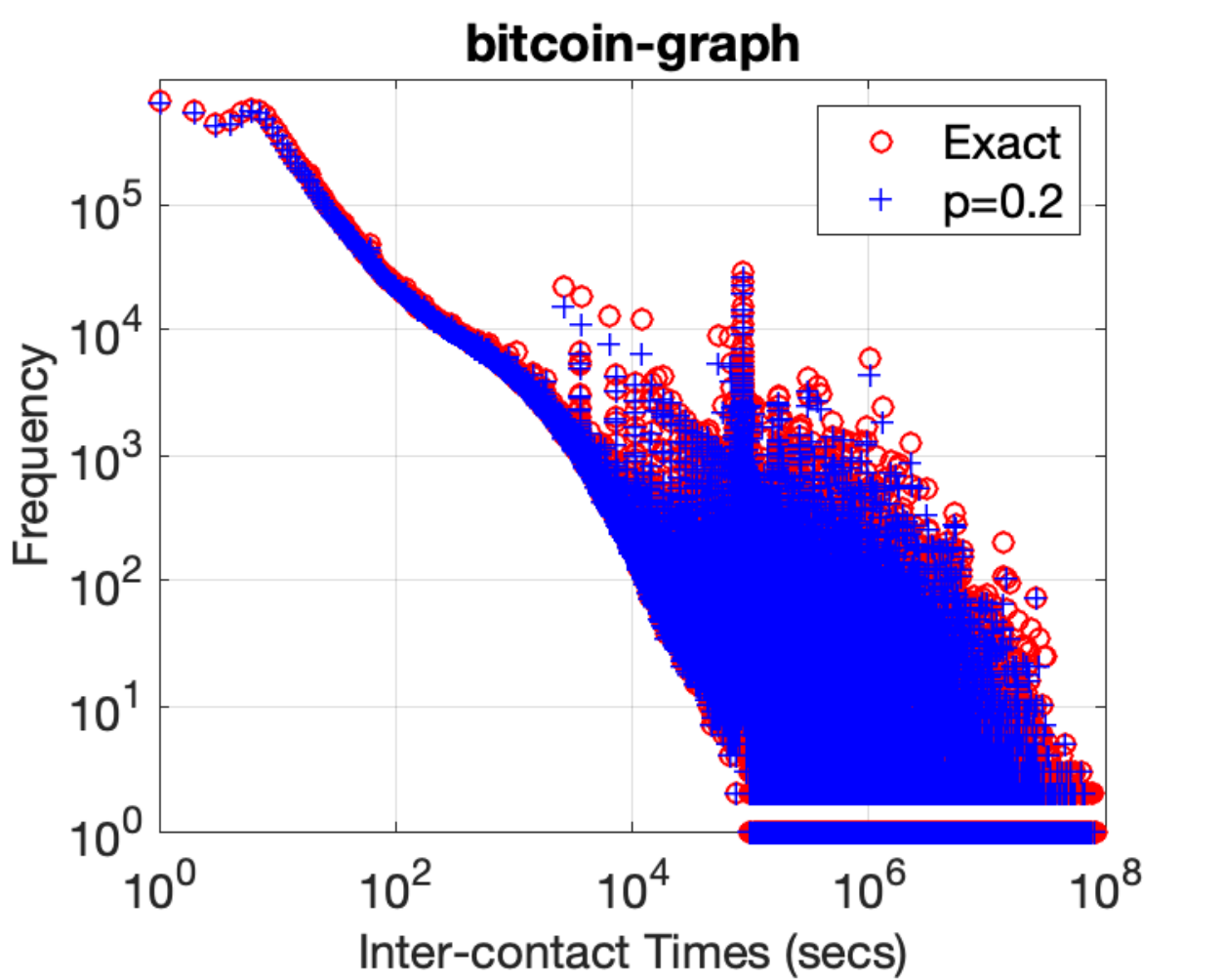}}

\caption{Estimation results for the distribution of inter-contact times compared to the exact distribution.
Results are shown for $p=0.1$ (top) and $p=0.2$ (bottom).
}
\label{fig:intercontact-time}
\end{figure*}
\renewcommand{\figsz}{0.22}
\begin{figure*}[h!]
\centering
\subfigure
{\includegraphics[width=\figsz\linewidth]{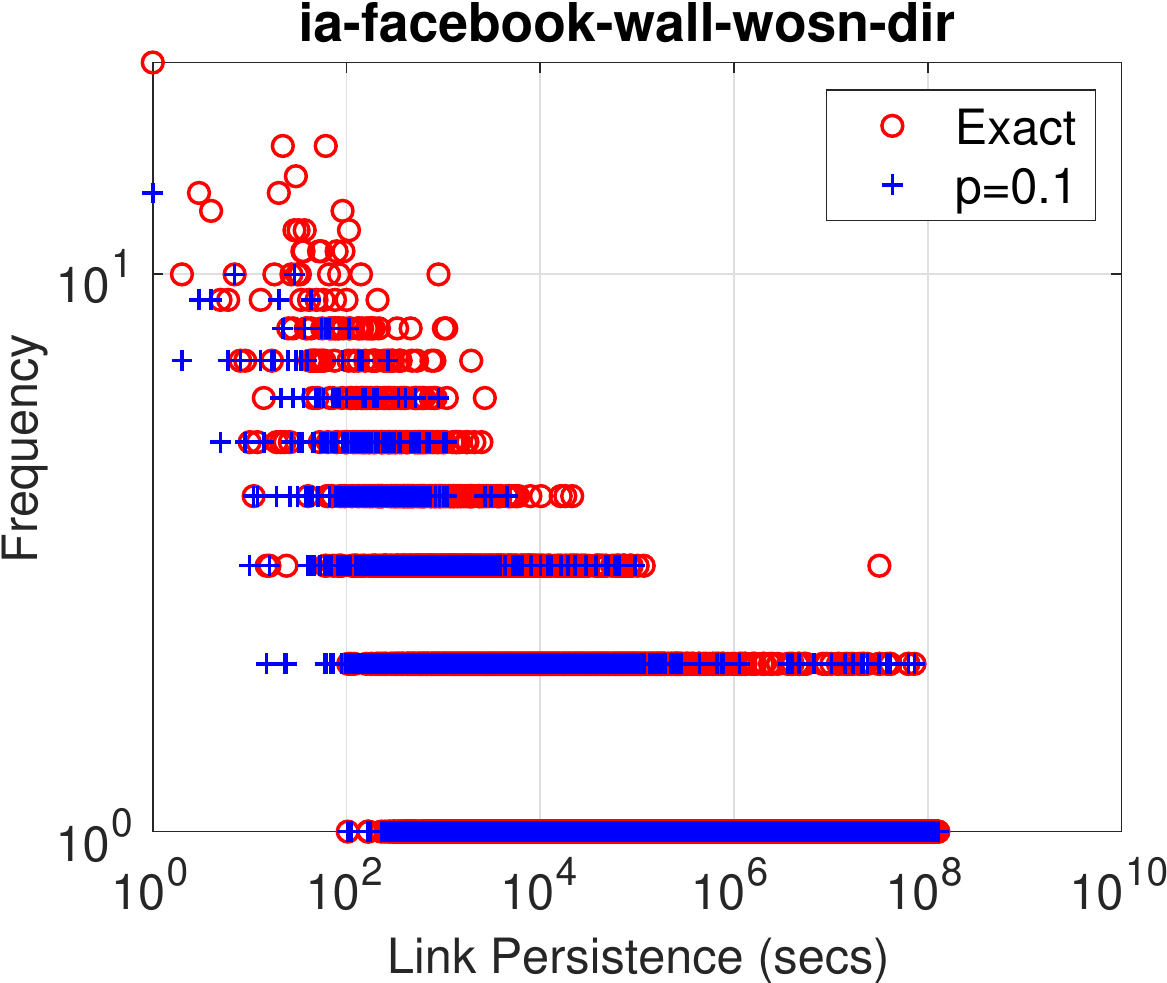}}
\subfigure
{\includegraphics[width=\figsz\linewidth]{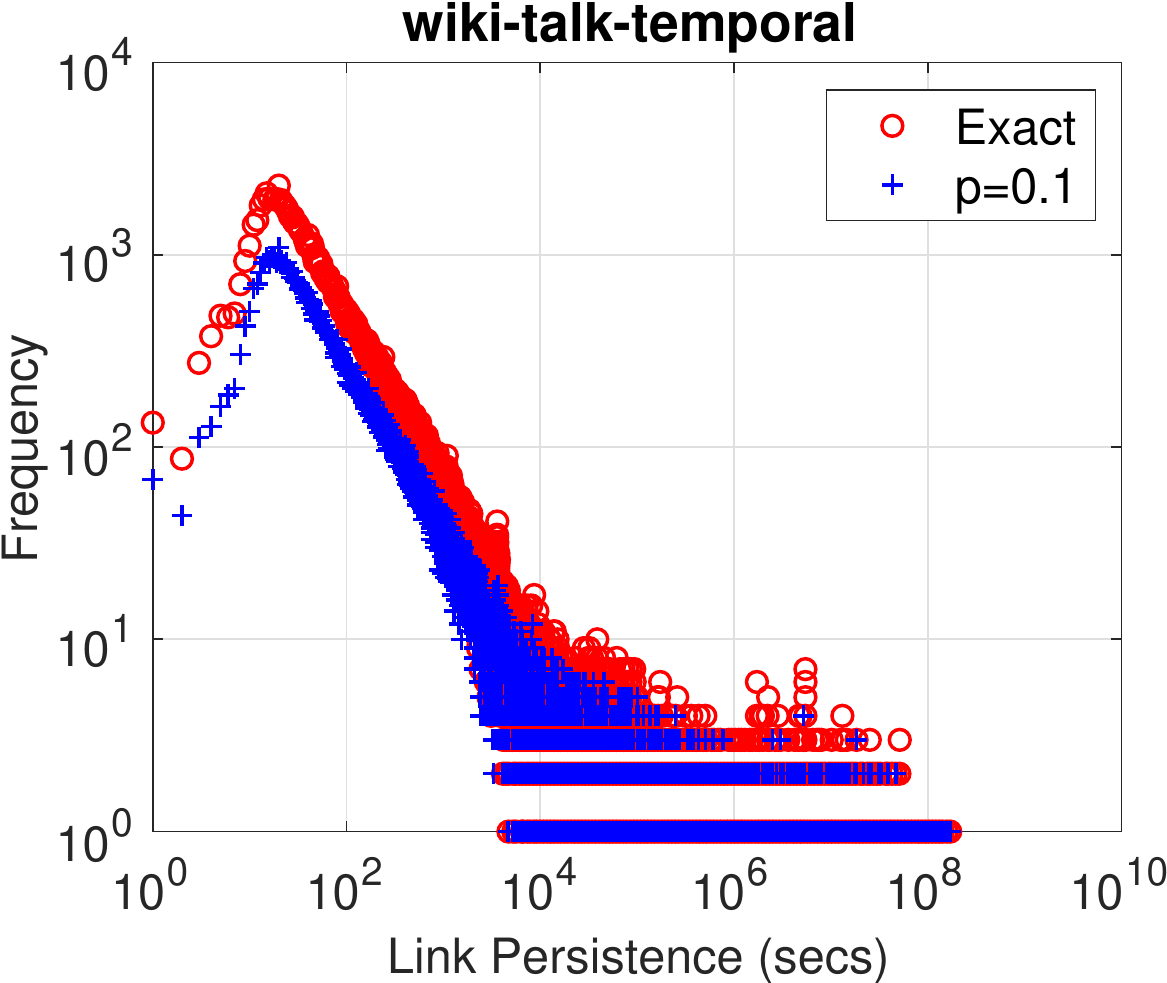}}
\subfigure
{\includegraphics[width=\figsz\linewidth]{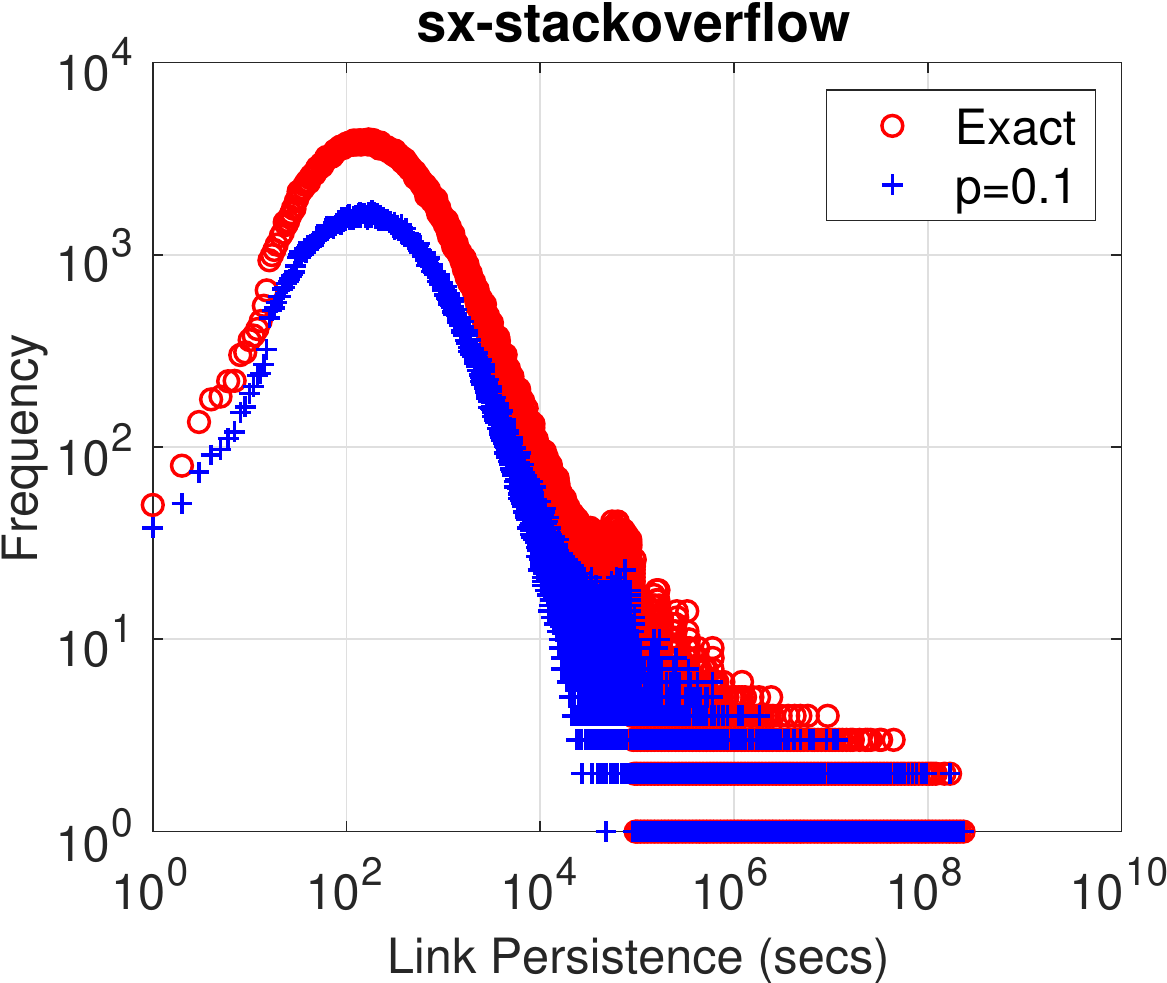}}
\subfigure
{\includegraphics[width=\figsz\linewidth]{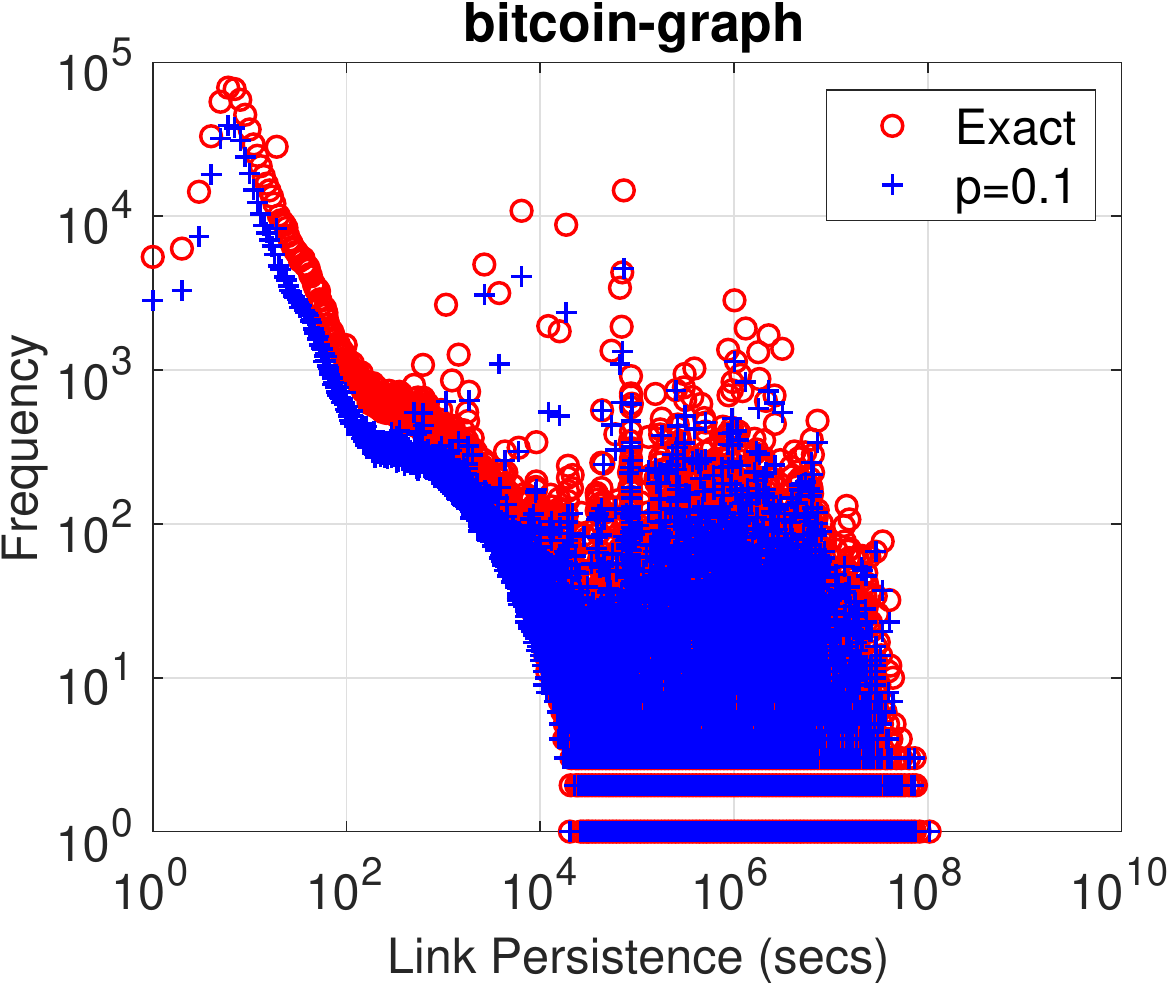}}

\subfigure
{\includegraphics[width=\figsz\linewidth]{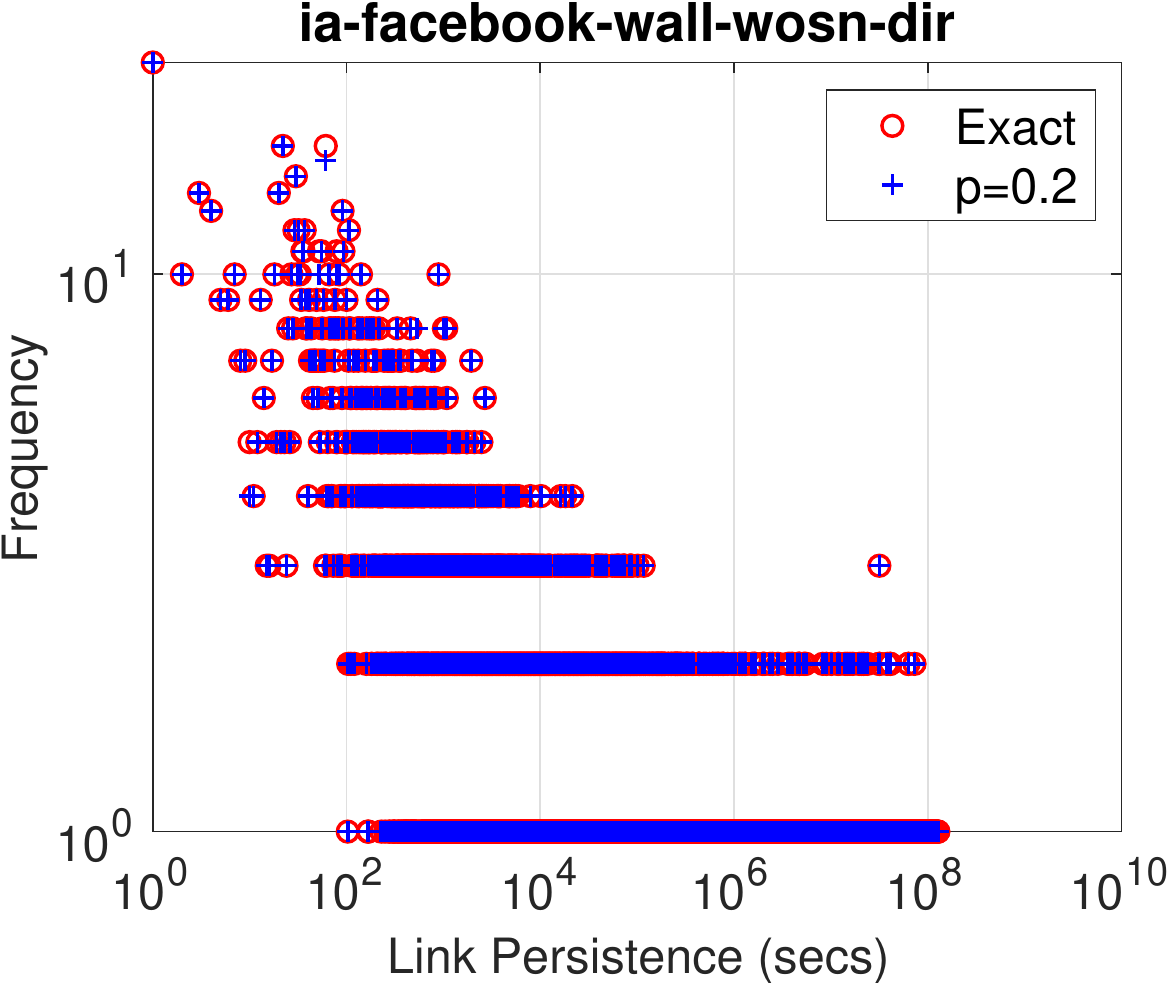}}
\subfigure
{\includegraphics[width=\figsz\linewidth]{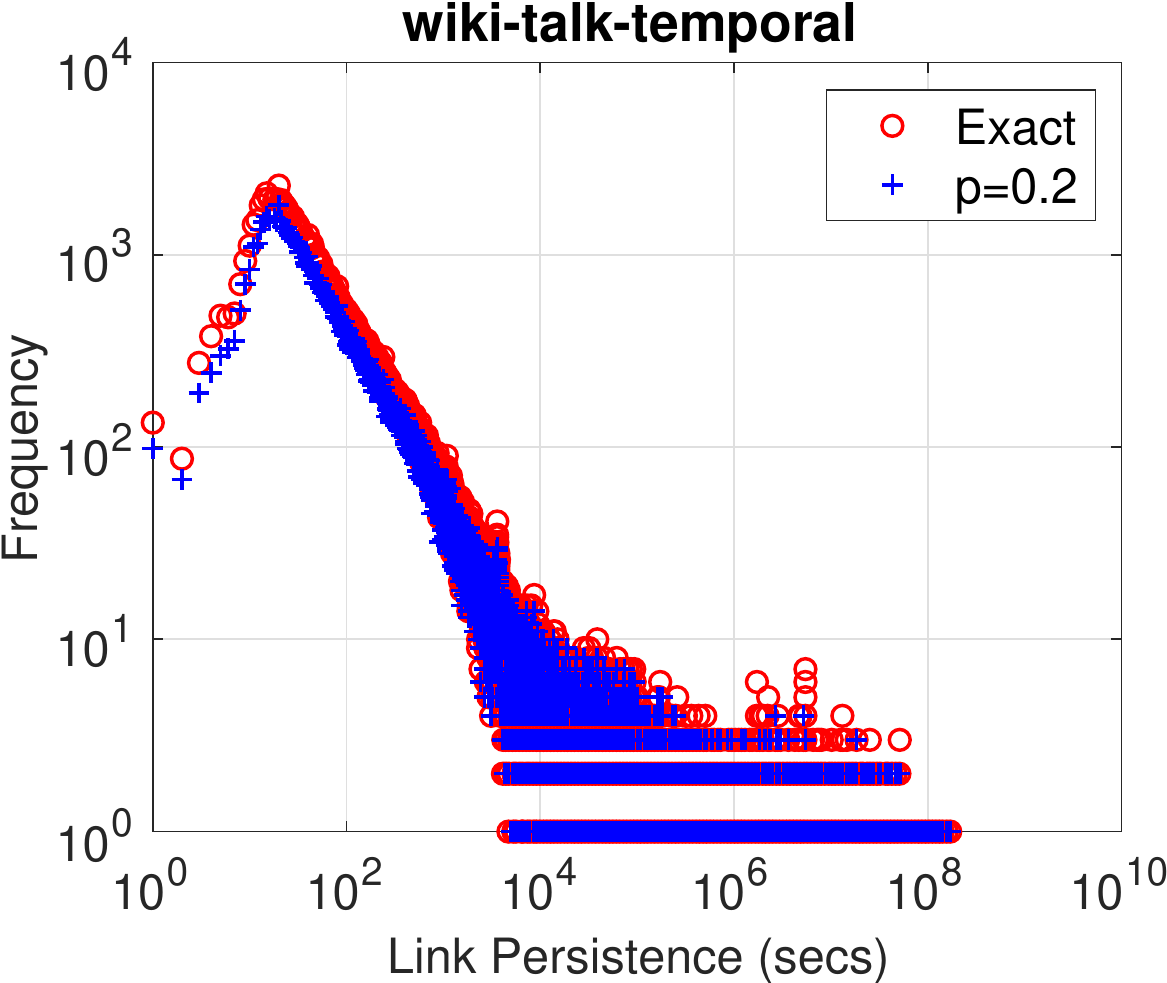}}
\subfigure
{\includegraphics[width=\figsz\linewidth]{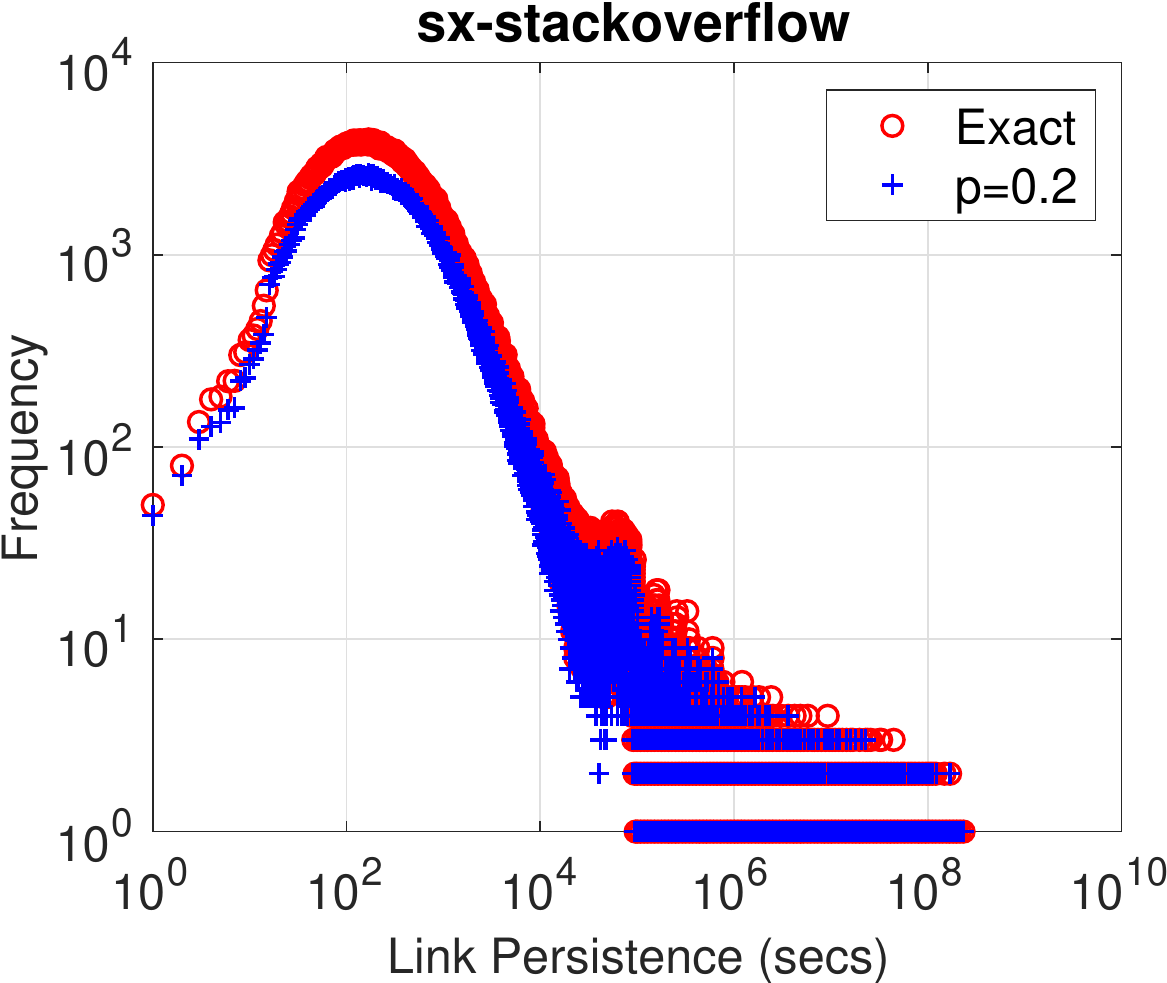}}
\subfigure
{\includegraphics[width=\figsz\linewidth]{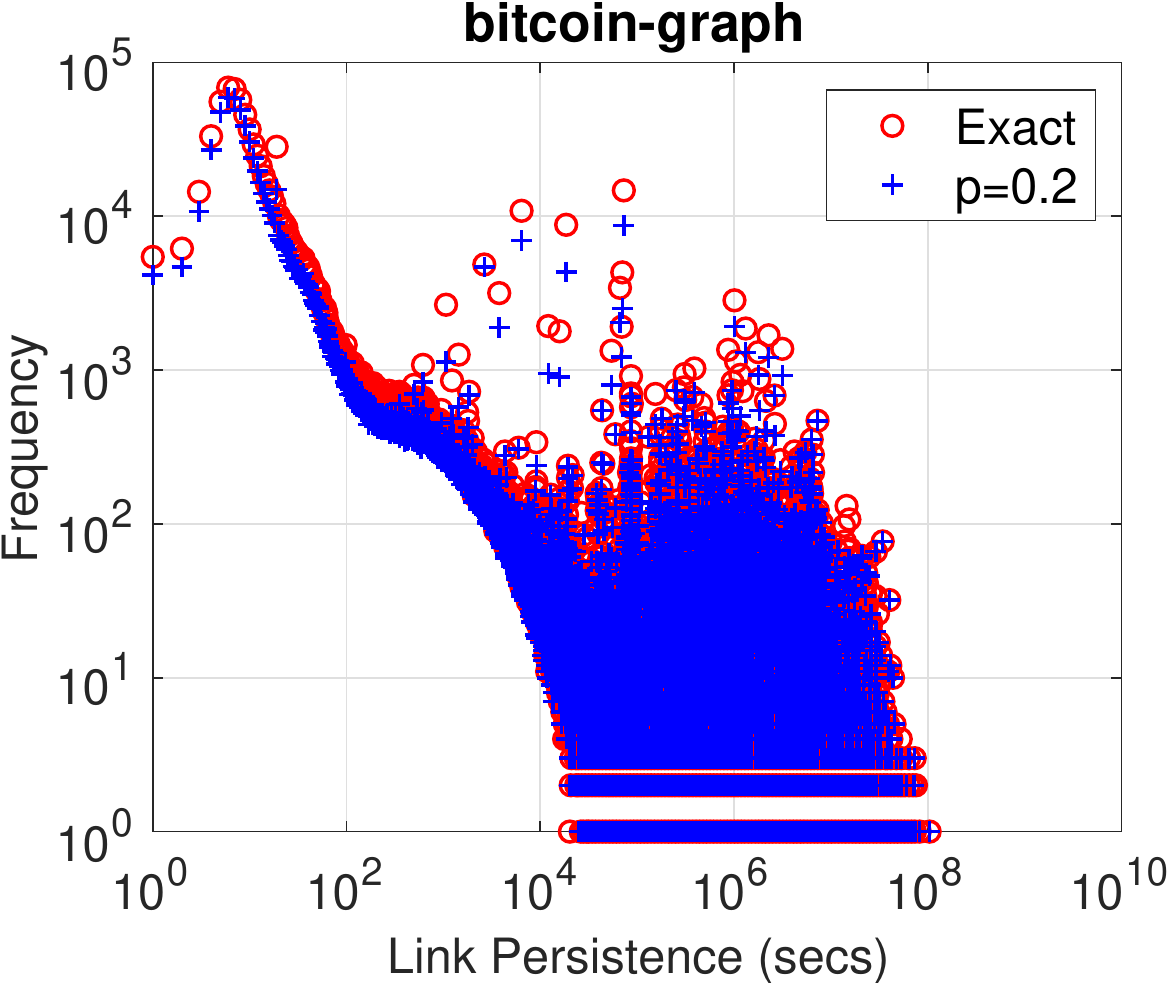}}

\caption{
Estimation results for the distribution of link persistence scores compared to the exact distribution.
Results are shown for $p=0.1$ (top) and $p=0.2$ (bottom).
}
\label{fig:persistence}
\end{figure*}

\section{Related work} 
\label{sec:related-work}
Sampling algorithms are fundamental in studying and understanding networks~\cite{vitter1985random,newman2018networks,kolaczyk2014statistical,ahmed2014network}, where the goal is to collect a representative sample that capture the characteristics of the full network. 
Network sampling has been widely studied in the context of small static networks that can fit entirely in memory~\cite{kolaczyk2014statistical}. For instance, there is uniform node sampling~\cite{stumpf2005subnets}, random walk sampling\cite{leskovec2006sampling}, edge sampling~\cite{ahmed2014network}, among others~\cite{ahmed2016estimation,liu2019sampling}. 
More recently, there has been a growing interest in sampling techniques for streaming network data in which temporal networks evolve continuously in time~\cite{ahmed2020neurips,cormode2014second,zhao2016link,jha2013space,ahmed2014graph,stefani2017triest,ahmed2017sampling,ahmedijcai18sampling,jha2015counting,pavan2013counting,lim2015mascot,simpson2015catching,ahmed2010time,lim2018memory}. 
For seminal surveys on the topic, see~\cite{ahmed2014network,mcgregor2014graph}. 

However, most existing methods for sampling streaming network data have focused on the primary objective of selecting a sample to estimate static network properties, \eg, point statistics such as global triangle count or clustering coefficient~\cite{ahmed2017sampling}. 
As such, it is unclear how representative these samples are for temporal network statistics such as the link strength~\cite{xiang2010modeling}, link persistence~\cite{clauset2012persistence}, burstiness~\cite{barabasi2005origin}, temporal motifs~\cite{kovanen2011temporal}, among others~\cite{holme2012temporal}. 
Despite the fundamental importance of this question, it has not been addressed in the context of streaming and online methods. 

Some of the recent work on stream sampling focused on multi-graph streams~\cite{stefani2017triest,vitter1985random,lim2018memory}, where edges may appear multiple times in the stream. However, most of these methods (e.g., Triest, reservoir sampling) mainly sample edges separately, thus, multiple occurrences of an edge $(u,v)$ may appear in the final sample, which allocates more space. In addition, the recent work in~\cite{lim2018memory} uses a uniform sampling probability to sample the edges, and stores an edge only once with its estimator. On the other hand, our proposed Algorithm~\ref{alg-TNS} adaptively samples edges with probabilities proportional to their link strength, and incrementally updates the overall estimate of link strength of an edge $(u,v)$, and stores an edge only once with its estimator. This leads to more space-efficient and accurate samples.  

There has been one recent work for sampling temporal motifs~\cite{liu2019sampling}. 
However, their work focused on a different problem based on counting motifs in temporal networks that form within some time $\bigtriangleup t$ ~\cite{kovanen2011temporal}. 
More specifically, their approach uses judicious partitioning of interactions in time bins, which can obfuscate or dilute temporal and structural information. 
In this paper, we formulate instead the notion of a \emph{temporally weighted motif} based on the temporal network link decay model.
We argue that this formulation is more meaningful and useful for practical applications especially related to prediction and forecasting where links and motifs that occur more recently are more important than those occurring in the distant past. 
In addition to the difference in problem, that work does \emph{not} focus on streaming nor the online setting since the entire graph is loaded into memory.

The temporally decaying model of temporal networks is useful for many important predictive modeling and forecasting tasks including classification~\cite{rossi2012time,sharan2008temporal}, link prediction~\cite{dunlavy2011temporal,timescore2012,choudhary2013link,muniz2018combining,chen2018exploiting,chi2019link,FLORENTINO2020106268,wu2020link,zhao2019incremental}, influence modeling~\cite{goyal2010learning}, regression~\cite{dpr-IM}, and anomaly detection~\cite{aggarwal2011outlier,rossi2013modeling}.
Despite the practical importance of the temporal link decaying model, our work is the first to propose network sampling and unbiased estimation algorithms for this setting.
Therefore, the proposed temporal decay sampling and unbiased estimation methods bring new opportunities for many real-world applications that involve prediction and forecasting from temporal networks representing a sequence of timestamped edges. This includes recommendation~\cite{dunlavy2011temporal,Bharadhwaj2018}, influence modeling~\cite{goyal2010learning}, visitor stitching~\cite{HONE}, etc.

Moreover, there has been a lot of research on deriving new and important temporal network statistics and properties that appropriately characterize the temporal network~\cite{holme2012temporal}. Other recent work has focused on extending node ranking and importance measures to dynamic networks such as Katz~\cite{grindrod2011communicability} and eigenvector centrality~\cite{taylor2017eigenvector}. These centrality measures use a sequence of static snapshot graphs to compute an importance or node centrality score of nodes. Since the proposed temporal sampling framework is general and can be used to estimate a time-dependent representation of the temporal network, it can be used to obtain unbiased estimates of these recent dynamic node centrality measures.

The proposed temporal network sampling framework can also be leveraged for estimation of node embeddings~\cite{from-comm-to-structural-role-embeddings} including both community-based (proximity) and role-based structural node embeddings~\cite{ahmed2018learning,roles2015tkde}. More recently, there has been a surge in activity for developing node embedding and graph representation learning methods for temporal networks. There have been embedding methods proposed for both continuous-time dynamic networks consisting of a stream of timestamped edges~\cite{nguyen2018continuous,lee2019temporal,Liu2019www} as well as discrete-time dynamic networks where the actual edge stream is approximated with a sequence of static snapshot graphs~\cite{rossi2013modeling,mahdavi2019dynamic, taheri2019predictive,hogun-ijcai19,Taheri-www2019}. All of these works may benefit from the proposed framework as it estimates a time-dependent representation of the temporal network that can be used as input to any of these methods for learning time-dependent node embeddings.

In the context of accumulating sample based counts of repeated objects, Sticky Sampling~\cite{10.5555/1287369.1287400} and counting Samples~\cite{10.5555/314500.315083} have been proposed, together with Sample and Hold~\cite{10.1145/859716.859719} in the context of network measurement, along with adaptive versions for fixed size reservoirs ~\cite{10.1145/2783258.2783279,10.1145/1064212.1064223}. Our approach differs from these methods in many ways and provide the following advantages. First, the cost of updating the sample is much cheaper compared to these methods. Second, the discard step is computationally cheaper being $O(1)$ to pick the minimum priority element. Third, our approach provides unbiased estimators not only for single links, but also for link-product counts for temporal motifs through Theorem~\ref{thm:nonu:C}(iii).

\section{Conclusion} 
\label{sec:conc}
This work proposed a novel general framework for online sampling and unbiased estimation of temporal networks.
The framework gives rise to online single-pass streaming sampling algorithms for estimating arbitrary temporal network statistics.
We also proposed a temporal decay sampling algorithm for estimating statistics based on the temporal decay model that assumes the strength of links evolve as a function of time, and the temporal statistics and temporal motif patterns are temporally weighted accordingly.
To the best of our knowledge, this work is the first to propose sampling and unbiased estimation algorithms for this setting, which is fundamentally important for practical applications involving prediction and forecasting from temporal networks.
The proposed framework and temporal network sampling algorithms that arise from it, enable fast, accurate, and memory-efficient statistical estimation of temporal network patterns and properties. 
Finally, the experiments demonstrated the effectiveness of the proposed approach for unbiased estimation of temporal network statistics. Other graph properties such as page rank, degree distribution, and centrality would be suitable for future extensions of the proposed framework.

\bibliographystyle{unsrt}
\bibliography{paper}

\begin{thebibliography}{10}

\bibitem{newman2018networks}
Mark Newman.
\newblock {\em Networks}.
\newblock Oxford university press, 2018.

\bibitem{newman2006structure}
Mark~Ed Newman, Albert-L{\'a}szl{\'o}~Ed Barab{\'a}si, and Duncan~J Watts.
\newblock {\em The structure and dynamics of networks.}
\newblock Princeton university press, 2006.

\bibitem{holme2015modern}
Petter Holme.
\newblock Modern temporal network theory: a colloquium.
\newblock {\em The European Physical Journal B}, 88(9):234, 2015.

\bibitem{nguyen2018continuous}
Giang~Hoang Nguyen, John~Boaz Lee, Ryan~A Rossi, Nesreen~K Ahmed, Eunyee Koh,
  and Sungchul Kim.
\newblock Continuous-time dynamic network embeddings.
\newblock In {\em WWW}, pages 969--976, 2018.

\bibitem{rossi2012time}
Ryan Rossi and Jennifer Neville.
\newblock Time-evolving relational classification and ensemble methods.
\newblock In {\em PAKDD}, pages 1--13. Springer, 2012.

\bibitem{sharan2008temporal}
Umang Sharan and Jennifer Neville.
\newblock Temporal-relational classifiers for prediction in evolving domains.
\newblock In {\em ICDM}, pages 540--549, 2008.

\bibitem{holme2012temporal}
Petter Holme and Jari Saram{\"a}ki.
\newblock Temporal networks.
\newblock {\em Physics reports}, 519(3):97--125, 2012.

\bibitem{li2017fundamental}
Aming Li, Sean~P Cornelius, Y-Y Liu, Long Wang, and A-L Barab{\'a}si.
\newblock The fundamental advantages of temporal networks.
\newblock {\em Science}, 358(6366):1042--1046, 2017.

\bibitem{eckmann2004entropy}
Jean-Pierre Eckmann, Elisha Moses, and Danilo Sergi.
\newblock Entropy of dialogues creates coherent structures in e-mail traffic.
\newblock {\em PNAS}, 101(40):14333--14337, 2004.

\bibitem{rocha2017dynamics}
Luis~EC Rocha.
\newblock Dynamics of air transport networks: A review from a complex systems
  perspective.
\newblock {\em C. J. of Aero.}, 30(2):469--478, 2017.

\bibitem{peixoto2018change}
Tiago~P Peixoto and Laetitia Gauvin.
\newblock Change points, memory and epidemic spreading in temporal networks.
\newblock {\em Scientific reports}, 8(1):15511, 2018.

\bibitem{masuda2013predicting}
Naoki Masuda and Petter Holme.
\newblock Predicting and controlling infectious disease epidemics using
  temporal networks.
\newblock {\em F1000prime reports}, 5, 2013.

\bibitem{chen2013information}
Wei Chen, Laks~VS Lakshmanan, and Carlos Castillo.
\newblock Information and influence propagation in social networks.
\newblock {\em Syn. Lec. on Data Man.}, 5(4), 2013.

\bibitem{goyal2010learning}
Amit Goyal, Francesco Bonchi, and Laks~VS Lakshmanan.
\newblock Learning influence probabilities in social networks.
\newblock In {\em WSDM}, pages 241--250. ACM, 2010.

\bibitem{ahmed2014network}
Nesreen~K Ahmed, Jennifer Neville, and Ramana Kompella.
\newblock Network sampling: From static to streaming graphs.
\newblock {\em TKDD}, 8(2):7, 2014.

\bibitem{ahmed2014graph}
Nesreen~K Ahmed, Nick Duffield, Jennifer Neville, and Ramana Kompella.
\newblock Graph sample and hold: A framework for big-graph analytics.
\newblock In {\em KDD}, pages 1446--1455, 2014.

\bibitem{soundarajan2016generating}
Sucheta Soundarajan, Acar Tamersoy, Elias~B Khalil, Tina Eliassi-Rad,
  Duen~Horng Chau, Brian Gallagher, and Kevin Roundy.
\newblock Generating graph snapshots from streaming edge data.
\newblock In {\em WWW}, pages 109--110, 2016.

\bibitem{valdano2018epidemic}
Eugenio Valdano, Michele~Re Fiorentin, Chiara Poletto, and Vittoria Colizza.
\newblock Epidemic threshold in continuous-time evolving networks.
\newblock {\em Physical review letters}, 120(6):068302, 2018.

\bibitem{holme2019impact}
Petter Holme and Luis~EC Rocha.
\newblock Impact of misinformation in temporal network epidemiology.
\newblock {\em Network Science}, 7(1):52--69, 2019.

\bibitem{sulo2010meaningful}
Rajmonda Sulo, Tanya Berger-Wolf, and Robert Grossman.
\newblock Meaningful selection of temporal resolution for dynamic networks.
\newblock In {\em MLG KDD}, pages 127--136, 2010.

\bibitem{caceres2011temporal}
Rajmonda~Sulo Caceres, Tanya Berger-Wolf, and Robert Grossman.
\newblock Temporal scale of processes in dynamic networks.
\newblock In {\em ICDM Workshops}, pages 925--932. IEEE, 2011.

\bibitem{fenn2012dynamical}
Daniel~J Fenn, Mason~A Porter, Peter~J Mucha, Mark McDonald, Stacy Williams,
  Neil~F Johnson, and Nick~S Jones.
\newblock Dynamical clustering of exchange rates.
\newblock {\em Quantitative Finance}, 12(10):1493--1520, 2012.

\bibitem{flores2018eigenvector}
Julio Flores and Miguel Romance.
\newblock On eigenvector-like centralities for temporal networks: Discrete vs.
  continuous time scales.
\newblock {\em Journal of Computational and Applied Mathematics},
  330:1041--1051, 2018.

\bibitem{aggarwal2006biased}
Charu~C Aggarwal.
\newblock On biased reservoir sampling in the presence of stream evolution.
\newblock In {\em Proceedings of the 32nd international conference on Very
  large data bases}, pages 607--618, 2006.

\bibitem{zino2016continuous}
Lorenzo Zino, Alessandro Rizzo, and Maurizio Porfiri.
\newblock Continuous-time discrete-distribution theory for activity-driven
  networks.
\newblock {\em Physical review letters}, 117(22):228302, 2016.

\bibitem{yang2018influential}
Xin Yang and Ju~Fan.
\newblock Influential user subscription on time-decaying social streams.
\newblock {\em arXiv:1802.05305}, 2018.

\bibitem{zino2017analytical}
Lorenzo Zino, Alessandro Rizzo, and Maurizio Porfiri.
\newblock An analytical framework for the study of epidemic models on activity
  driven networks.
\newblock {\em Journal of Complex Networks}, 5(6):924--952, 2017.

\bibitem{lohr2019sampling}
Sharon~L Lohr.
\newblock {\em Sampling: Design and Analysis}.
\newblock Chapman and Hall/CRC, 2019.

\bibitem{kolaczyk2014statistical}
Eric~D Kolaczyk and G{\'a}bor Cs{\'a}rdi.
\newblock {\em Statistical analysis of network data with R}, volume~65.
\newblock Springer, 2014.

\bibitem{stumpf2005subnets}
Michael~PH Stumpf, Carsten Wiuf, and Robert~M May.
\newblock Subnets of scale-free networks are not scale-free: sampling
  properties of networks.
\newblock {\em PNAS}, 102(12):4221--4224, 2005.

\bibitem{leskovec2006sampling}
Jure Leskovec and Christos Faloutsos.
\newblock Sampling from large graphs.
\newblock In {\em KDD}, pages 631--636, 2006.

\bibitem{cormode2014second}
Graham Cormode and Hossein Jowhari.
\newblock A second look at counting triangles in graph streams.
\newblock {\em Theoretical Computer Science}, 552:44--51, 2014.

\bibitem{jha2013space}
Madhav Jha, Comandur Seshadhri, and Ali Pinar.
\newblock A space efficient streaming algorithm for triangle counting using the
  birthday paradox.
\newblock In {\em KDD}, pages 589--597, 2013.

\bibitem{stefani2017triest}
Lorenzo~De Stefani, Alessandro Epasto, Matteo Riondato, and Eli Upfal.
\newblock Triest: Counting local and global triangles in fully dynamic streams
  with fixed memory size.
\newblock {\em TKDD}, 11(4):43, 2017.

\bibitem{ahmed2017sampling}
Nesreen~K Ahmed, Nick Duffield, Theodore~L Willke, and Ryan~A Rossi.
\newblock On sampling from massive graph streams.
\newblock {\em Proceedings of the VLDB Endowment}, 10(11):1430--1441, 2017.

\bibitem{ahmedijcai18sampling}
Nesreen~K. Ahmed, Nick Duffield, and Liangzhen Xia.
\newblock Sampling for approximate bipartite network projection.
\newblock In {\em IJCAI}, pages 3286--3292, 2018.

\bibitem{simpson2015catching}
Olivia Simpson, C~Seshadhri, and Andrew McGregor.
\newblock Catching the head, tail, and everything in between: A streaming
  algorithm for the degree distribution.
\newblock In {\em ICDM}, pages 979--984, 2015.

\bibitem{jha2015counting}
Madhav Jha, Ali Pinar, and C~Seshadhri.
\newblock Counting triangles in real-world graph streams: Dealing with repeated
  edges and time windows.
\newblock In {\em 49th Asilomar Conference on Signals, Systems and Computers},
  pages 1507--1514. IEEE, 2015.

\bibitem{pavan2013counting}
A~Pavan, Kanat Tangwongsan, Srikanta Tirthapura, and Kun-Lung Wu.
\newblock Counting and sampling triangles from a graph stream.
\newblock {\em VLDB}, 6(14), 2013.

\bibitem{lim2015mascot}
Yongsub Lim and U~Kang.
\newblock Mascot: Memory-efficient and accurate sampling for counting local
  triangles in graph streams.
\newblock In {\em KDD}, pages 685--694. ACM, 2015.

\bibitem{ray2019efficient}
Abhik Ray, Lawrence~B Holder, and Albert Bifet.
\newblock Efficient frequent subgraph mining on large streaming graphs.
\newblock {\em Intelligent Data Analysis}, 23(1):103--132, 2019.

\bibitem{choudhury2013streamworks}
Sutanay Choudhury, Lawrence Holder, George Chin, Abhik Ray, Sherman Beus, and
  John Feo.
\newblock Streamworks: a system for dynamic graph search.
\newblock In {\em Proceedings of the 2013 ACM SIGMOD International Conference
  on Management of Data}, pages 1101--1104, 2013.

\bibitem{mcgregor2014graph}
Andrew McGregor.
\newblock Graph stream algorithms: a survey.
\newblock {\em ACM SIGMOD Record}, 43(1):9--20, 2014.

\bibitem{aggarwal2018extracting}
Charu~C Aggarwal.
\newblock Extracting real-time insights from graphs and social streams.
\newblock In {\em The 41st International ACM SIGIR Conference on Research \&
  Development in Information Retrieval}, pages 1339--1339, 2018.

\bibitem{xiang2010modeling}
Rongjing Xiang, Jennifer Neville, and Monica Rogati.
\newblock Modeling relationship strength in online social networks.
\newblock In {\em WWW}, pages 981--990, 2010.

\bibitem{clauset2012persistence}
Aaron Clauset and Nathan Eagle.
\newblock Persistence and periodicity in a dynamic proximity network.
\newblock {\em arXiv:1211.7343}, 2012.

\bibitem{barabasi2005origin}
Albert-Laszlo Barabasi.
\newblock The origin of bursts and heavy tails in human dynamics.
\newblock {\em Nature}, 435(7039):207, 2005.

\bibitem{kovanen2011temporal}
Lauri Kovanen, M{\'a}rton Karsai, Kimmo Kaski, J{\'a}nos Kert{\'e}sz, and Jari
  Saram{\"a}ki.
\newblock Temporal motifs in time-dependent networks.
\newblock {\em Journal of Statistical Mechanics: Theory and Experiment},
  2011(11):P11005, 2011.

\bibitem{AggarwalSDM2020}
Charu~C. Aggarwal, Yao Li, and Philip~S. Yu.
\newblock {\em On Supervised Change Detection in Graph Streams}, pages
  289--297.

\bibitem{larock2020understanding}
Timothy LaRock, Timothy Sakharov, Sahely Bhadra, and Tina Eliassi-Rad.
\newblock Understanding the limitations of network online learning.
\newblock {\em arXiv preprint arXiv:2001.07607}, 2020.

\bibitem{fond2018designing}
Timothy~La Fond, Jennifer Neville, and Brian Gallagher.
\newblock Designing size consistent statistics for accurate anomaly detection
  in dynamic networks.
\newblock {\em ACM Transactions on Knowledge Discovery from Data (TKDD)},
  12(4):1--49, 2018.

\bibitem{la2017ensemble}
Timothy La~Fond, Geoffrey Sanders, Christine Klymko, et~al.
\newblock An ensemble framework for detecting community changes in dynamic
  networks.
\newblock In {\em 2017 IEEE High Performance Extreme Computing Conference
  (HPEC)}, pages 1--6. IEEE, 2017.

\bibitem{10.1109/ICDE.2009.65}
Graham Cormode, Vladislav Shkapenyuk, Divesh Srivastava, and Bojian Xu.
\newblock Forward decay: A practical time decay model for streaming systems.
\newblock In {\em Proceedings of the 2009 IEEE International Conference on Data
  Engineering}, ICDE ’09, page 138–149, USA, 2009. IEEE Computer Society.

\bibitem{milo2002network}
Ron Milo, Shai Shen-Orr, Shalev Itzkovitz, Nadav Kashtan, Dmitri Chklovskii,
  and Uri Alon.
\newblock Network motifs: simple building blocks of complex networks.
\newblock {\em Science}, 298(5594):824--827, 2002.

\bibitem{ahmed2015efficient}
Nesreen~K Ahmed, Jennifer Neville, Ryan~A Rossi, and Nick Duffield.
\newblock Efficient graphlet counting for large networks.
\newblock In {\em ICDM}, pages 1--10, 2015.

\bibitem{benson2016higher}
Austin~R Benson, David~F Gleich, and Jure Leskovec.
\newblock Higher-order organization of complex networks.
\newblock {\em Science}, 353(6295):163--166, 2016.

\bibitem{ahmed2017graphlet}
Nesreen~K Ahmed, Jennifer Neville, Ryan~A Rossi, Nick~G Duffield, and
  Theodore~L Willke.
\newblock Graphlet decomposition: Framework, algorithms, and applications.
\newblock {\em KAIS}, 50(3):689--722, 2017.

\bibitem{ahmed2018learning}
Nesreen~K Ahmed, Ryan Rossi, John~Boaz Lee, Theodore~L Willke, Rong Zhou,
  Xiangnan Kong, and Hoda Eldardiry.
\newblock Learning role-based graph embeddings.
\newblock {\em arXiv:1802.02896}, 2018.

\bibitem{HONE}
Ryan~A. Rossi, Nesreen~K. Ahmed, Eunyee Koh, Sungchul Kim, Anup Rao, and Yasin
  Abbasi-Yadkori.
\newblock A structural graph representation learning framework.
\newblock In {\em WSDM}, 2020.

\bibitem{Taheri-www2019}
Aynaz Taheri, Kevin Gimpel, and Tanya Berger-Wolf.
\newblock Learning to represent the evolution of dynamic graphs with recurrent
  models.
\newblock In {\em WWW}, pages 301--307, 2019.

\bibitem{liu2019sampling}
Paul Liu, Austin~R Benson, and Moses Charikar.
\newblock Sampling methods for counting temporal motifs.
\newblock In {\em WSDM}, pages 294--302, 2019.

\bibitem{vitter1985random}
Jeffrey~S Vitter.
\newblock Random sampling with a reservoir.
\newblock {\em ACM Transactions on Mathematical Software (TOMS)}, 11(1):37--57,
  1985.

\bibitem{duffield2007priority}
Nick Duffield, Carsten Lund, and Mikkel Thorup.
\newblock Priority sampling for estimation of arbitrary subset sums.
\newblock {\em Journal of the ACM (JACM)}, 54(6):32, 2007.

\bibitem{miritello2011dynamical}
Giovanna Miritello, Esteban Moro, and Rub{\'e}n Lara.
\newblock Dynamical strength of social ties in information spreading.
\newblock {\em Physical Review E}, 83(4):045102, 2011.

\bibitem{nr}
Ryan~A. Rossi and Nesreen~K. Ahmed.
\newblock The network data repository with interactive graph analytics and
  visualization.
\newblock In {\em AAAI}, 2015.

\bibitem{lim2018memory}
Yongsub Lim, Minsoo Jung, and U~Kang.
\newblock Memory-efficient and accurate sampling for counting local triangles
  in graph streams: from simple to multigraphs.
\newblock {\em ACM Transactions on Knowledge Discovery from Data (TKDD)},
  12(1):1--28, 2018.

\bibitem{achlioptas2013near}
Dimitris Achlioptas, Zohar~S Karnin, and Edo Liberty.
\newblock Near-optimal entrywise sampling for data matrices.
\newblock In {\em NeurIPS}, pages 1565--1573, 2013.

\bibitem{ahmed2016estimation}
Nesreen~K Ahmed, Theodore~L Willke, and Ryan~A Rossi.
\newblock Estimation of local subgraph counts.
\newblock In {\em IEEE Big Data}, pages 586--595. IEEE, 2016.

\bibitem{ahmed2020neurips}
Nesreen~K Ahmed and Nick Duffield.
\newblock Adaptive shrinkage estimation for streaming graphs.
\newblock In {\em Advances in Neural Information Processing Systems}, 2020.

\bibitem{zhao2016link}
Peixiang Zhao, Charu Aggarwal, and Gewen He.
\newblock Link prediction in graph streams.
\newblock In {\em 2016 IEEE 32nd International Conference on Data Engineering
  (ICDE)}, pages 553--564. IEEE, 2016.

\bibitem{ahmed2010time}
Nesreen~K Ahmed, Fredrick Berchmans, Jennifer Neville, and Ramana Kompella.
\newblock Time-based sampling of social network activity graphs.
\newblock In {\em MLG}, pages 1--9, 2010.

\bibitem{dunlavy2011temporal}
Daniel~M Dunlavy, Tamara~G Kolda, and Evrim Acar.
\newblock Temporal link prediction using matrix and tensor factorizations.
\newblock {\em TKDD}, 5(2):10, 2011.

\bibitem{timescore2012}
Lankeshwara Munasinghe and Ryutaro Ichise.
\newblock Time score: A new feature for link prediction in social networks.
\newblock {\em IEICE Transactions on Information and Systems},
  E95.D(3):821--828, 2012.

\bibitem{choudhary2013link}
Pankaj Choudhary, Nishchol Mishra, Sanjeev Sharma, and Ravindra Patel.
\newblock Link score: A novel method for time aware link prediction in social
  network.
\newblock {\em ICDMW}, 2013.

\bibitem{muniz2018combining}
Carlos~Pedro Muniz, Ronaldo Goldschmidt, and Ricardo Choren.
\newblock Combining contextual, temporal and topological information for
  unsupervised link prediction in social networks.
\newblock {\em Knowledge-Based Systems}, 156:129--137, 2018.

\bibitem{chen2018exploiting}
Huiyuan Chen and Jing Li.
\newblock Exploiting structural and temporal evolution in dynamic link
  prediction.
\newblock In {\em Proceedings of the 27th ACM International Conference on
  Information and Knowledge Management}, pages 427--436, 2018.

\bibitem{chi2019link}
Kuo Chi, Guisheng Yin, Yuxin Dong, and Hongbin Dong.
\newblock Link prediction in dynamic networks based on the attraction force
  between nodes.
\newblock {\em Knowledge-Based Systems}, 181:104792, 2019.

\bibitem{FLORENTINO2020106268}
{\'E}rick~S Florentino, Argus~AB Cavalcante, and Ronaldo~R Goldschmidt.
\newblock An edge creation history retrieval based method to predict links in
  social networks.
\newblock {\em Knowledge-Based Systems}, 205:106268, 2020.

\bibitem{wu2020link}
Xiaomin Wu, Jianshe Wu, Yafeng Li, and Qian Zhang.
\newblock Link prediction of time-evolving network based on node ranking.
\newblock {\em Knowledge-Based Systems}, page 105740, 2020.

\bibitem{zhao2019incremental}
Zhongying Zhao, Chao Li, Xuejian Zhang, Francisco Chiclana, and Enrique~Herrera
  Viedma.
\newblock An incremental method to detect communities in dynamic evolving
  social networks.
\newblock {\em Knowledge-Based Systems}, 163:404--415, 2019.

\bibitem{dpr-IM}
David~F. Gleich and Ryan~A. Rossi.
\newblock A dynamical system for pagerank with time-dependent teleportation.
\newblock {\em Internet Mathematics}, 10(1-2):188--217, 2014.

\bibitem{aggarwal2011outlier}
Charu~C Aggarwal, Yuchen Zhao, and S~Yu Philip.
\newblock Outlier detection in graph streams.
\newblock In {\em ICDE}, pages 399--409. IEEE, 2011.

\bibitem{rossi2013modeling}
Ryan~A. Rossi, Brian Gallagher, Jennifer Neville, and Keith Henderson.
\newblock Modeling dynamic behavior in large evolving graphs.
\newblock In {\em WSDM}, pages 667--676, 2013.

\bibitem{Bharadhwaj2018}
Homanga Bharadhwaj and Shruti Joshi.
\newblock Explanations for temporal recommendations.
\newblock {\em KI}, 32(4):267--272, Nov 2018.

\bibitem{grindrod2011communicability}
Peter Grindrod, Mark~C Parsons, Desmond~J Higham, and Ernesto Estrada.
\newblock Communicability across evolving networks.
\newblock {\em Phy. Rev. E}, 83(4):046120, 2011.

\bibitem{taylor2017eigenvector}
Dane Taylor, Sean~A Myers, Aaron Clauset, Mason~A Porter, and Peter~J Mucha.
\newblock Eigenvector-based centrality measures for temporal networks.
\newblock {\em Multiscale Modeling \& Simulation}, 15(1):537--574, 2017.

\bibitem{from-comm-to-structural-role-embeddings}
Ryan~A. Rossi, Di~Jin, Sungchul Kim, Nesreen~K. Ahmed, Danai Koutra, and
  John~Boaz Lee.
\newblock From community to role-based graph embeddings.
\newblock In {\em arXiv:1908.08572}, 2019.

\bibitem{roles2015tkde}
Ryan~A. Rossi and Nesreen~K. Ahmed.
\newblock Role discovery in networks.
\newblock {\em IEEE Transactions on Knowledge and Data Engineering (TKDE)},
  27(4):1112--1131, 2015.

\bibitem{lee2019temporal}
John~Boaz Lee, Giang Nguyen, Ryan~A Rossi, Nesreen~K Ahmed, Eunyee Koh, and
  Sungchul Kim.
\newblock Temporal network representation learning.
\newblock {\em arXiv:1904.06449}, 2019.

\bibitem{Liu2019www}
Xi~Liu, Ping-Chun Hsieh, Nick Duffield, Rui Chen, Muhe Xie, and Xidao Wen.
\newblock Real-time streaming graph embedding through local actions.
\newblock In {\em WWW}, pages 285--293, 2019.

\bibitem{mahdavi2019dynamic}
Sedigheh Mahdavi, Shima Khoshraftar, and Aijun An.
\newblock Dynamic joint variational graph autoencoders.
\newblock {\em arXiv:1910.01963}, 2019.

\bibitem{taheri2019predictive}
Aynaz Taheri and Tanya Berger-Wolf.
\newblock Predictive temporal embedding of dynamic graphs.
\newblock {\em ASONAM}, 2019.

\bibitem{hogun-ijcai19}
Hogun Park and Jennifer Neville.
\newblock Exploiting interaction links for node classification with deep graph
  neural networks.
\newblock In {\em IJCAI}, pages 3223--3230, 7 2019.

\bibitem{10.5555/1287369.1287400}
Gurmeet~Singh Manku and Rajeev Motwani.
\newblock Approximate frequency counts over data streams.
\newblock In {\em Proceedings of the 28th International Conference on Very
  Large Data Bases}, VLDB ’02, page 346–357. VLDB Endowment, 2002.

\bibitem{10.5555/314500.315083}
Phillip~B. Gibbons and Yossi Matias.
\newblock Synopsis data structures for massive data sets.
\newblock In {\em Proceedings of the Tenth Annual ACM-SIAM Symposium on
  Discrete Algorithms}, SODA ’99, page 909–910, USA, 1999. Society for
  Industrial and Applied Mathematics.

\bibitem{10.1145/859716.859719}
Cristian Estan and George Varghese.
\newblock New directions in traffic measurement and accounting: Focusing on the
  elephants, ignoring the mice.
\newblock {\em ACM Trans. Comput. Syst.}, 21(3):270–313, August 2003.

\bibitem{10.1145/2783258.2783279}
Edith Cohen.
\newblock Stream sampling for frequency cap statistics.
\newblock In {\em Proceedings of the 21th ACM SIGKDD International Conference
  on Knowledge Discovery and Data Mining}, KDD ’15, page 159–168, New York,
  NY, USA, 2015. Association for Computing Machinery.

\bibitem{10.1145/1064212.1064223}
Ken Keys, David Moore, and Cristian Estan.
\newblock A robust system for accurate real-time summaries of internet traffic.
\newblock In {\em Proceedings of the 2005 ACM SIGMETRICS International
  Conference on Measurement and Modeling of Computer Systems}, SIGMETRICS
  ’05, page 85–96, New York, NY, USA, 2005. Association for Computing
  Machinery.

\end{thebibliography}

\appendix
\section{Proofs of Theorems}
\label{sec:proofs}

\begin{proof}[Proof of Theorem~\ref{thm:nonu:C}]
Although (i) is special case of (ii), we prove (i) first then extend to (ii). We establish that
\be\label{eq:unb:iter}
\E[\hat C_{e,t}|\hat C_{e,t-1},Q]-C_{e,t}=\hat C_{e,t-1}-C_{e,t-1}\ee
for all members $Q$ of set of disjoint events whose union is identically true. Since $\hat
  C_{e,t_e-1}=C_{e,t_e-1}=0$ we then conclude that $\E[\hat C_{e,t}]=C_{e,t}$.
For $t_e\le s \le s'$, let $A\up 1_e(s)=\{e\notin \hat K_{s-1}\}$ (note
$A\up1_e(t_e)$ is identically true), let $A\up 2_e(s,s')$ denote the event $\{e\in\hat K_s\dots,\hat
  K_{s'}\}$, \ie, that $e$ is in sample at all times in $[s,s']$. Then $A\up1_e(s)A\up2_e(s,t-1)$ is the event that $e$ was sampled at time $s\le t-1$ and has remained in the reservoir up to and including time $t-1$.
  For each $t\ge t_e$ the union of the collection of events formed by
  $\{A\up1_e(s)A\up2_e(s,t-1):\ s\in[t_e,t-1]\}$, and $A\up1_e(t)$ is identically true.

(a) \textsl{Conditioning on $A\up1_e(t)$.}\, On $A\up1_e(t)$, $e_t\ne e$
implies $\hat C_{e,t}=\hat C_{e,t-1} = 0 = C_{e,t}-C_{e,t-1}$. On the
other hand $e_t=e$ implies $t\in \Omega$ since the arriving edge $e$ is not in the current sample. Further conditioning on $z_{e,t}=\min_{j\in
  \hat K_{j,t-1}}r_{j,t-1}$ then (\ref{eq:22}) 
  tells us
\be
\Pr[e\in \hat K_t | A\up1_e(t),z_{e,t}]=\Pr[u_e < w_{e,t}/z_{e,t}]=p_{e,t}
\ee
and hence regardless of $z_{e,t}$ we have
\be \E[\hat C_{e,t} | C_{e,t-1},A\up1_e(t),z_{e,t}]=\hat C_{e,t-1} + C_{e,t}-C_{e,t-1}
\ee 

(b) \textsl{Conditioning on $A\up1_e(s)A\up2_e(s,t-1)$ any
  $s\in[t_e,t-1]$.}\,  Under this condition $e\in \hat K_{t-1}$ and if
furthermore $e_t\in \hat K_{t-1}$ then $t\notin \Omega$ and the first
line in (\ref{eq:22}) holds.
Suppose instead $e_t\notin \hat K_{t-1}$ so that $t\in \Omega$. Let
$\cZ_{e}(s,t)=\{z_{e,s'}: s'\in [s,t]\cap\Omega\}$. Observing that
{\be
\Pr[A\up 2_e(s,t)|A\up1_e(s),\cZ_e(s,t)]=
\Pr[\bigcap_{s'\in [s,t]\cap\Omega}\{u_e< \frac{w_{e,s'}}{z_{e,s'}}\}]=p_{e,t}
\nonumber
\ee
}
then
\be
\Pr[e\in\hat K_t| A\up2_e(t-1,s)A\up1_e(s),\cZ_e(s,t)]
=
\frac{\Pr[A\up 2_e(s,t) | A\up1_e(s),\cZ_e(s,t)]}
{\Pr[A\up 2_e(s,t-1) | A\up1_e(s),\cZ_e(s,t-1)]}
=\frac{p_{e,t}}{p_{e,\omega(t)}} = q_{e,t} 
\ee
and hence
\be\label{eq:cond:AB}
\E[\hat C_{e,t}| \hat C_{e,t-1}, A\up1_e(s),\cZ_e(s,t)] =\hat C_{e,t-1}
\ee
independently of the conditions on the LHS of (\ref{eq:cond:AB}). As
noted above, $e\in \hat K_{t-1}$ when $A\up 2_e(s,t-1)$ is true in which case
$C_{e,t}=C_{e,t-1}$ and we recover (\ref{eq:unb:iter}). 

(ii) The proof employs a conditioning argument that generalizes a the property of Priority Sampling, namely, that inverse probability estimators of item samples are independent when conditioned on the priorities of other items; see \cite{duffield2007priority}. Our generalization establishes that link-product form estimators of subgraph multiplicities are unbiased. Let $z_{J,t}=\min_{j\in \hat K'_t\setminus J_t}r_{j,t}$. Then $J_t\subset \hat K_t$ iff $u_i\le w_{i,t}/z_{J,t}$ for all $i\in J_t$, in which case $z_{J,t}=z_{j,t}$ for all $j\in J_t$. A sufficient condition for Theorem~\ref{thm:nonu:C}(ii) is that
\be\label{eq:cond4} \E[\prod_{j\in J_t}\left(\hat C_{j,t}-C_{j,t}\right)|z_{J,t},J_{t-1}\subset \hat K_{t-1}]=0
\ee Since $ \hat C_{j,t}=\left(\hat C_{j,t-1} + c_{j,t}\right){I(u_j< w_{j,t}/z_{j,t})}/{q_{j,t}}$, then conditioning on $z_{J,t}$ and $\{J_{t-1}\subset \hat K_{t-1}\}$ fixes $\hat C_{j,t-1}, z_{j,t}$ and $q_{j,t}$ for $j\in J_t$. If we can show furthermore that the  $\{u_j: j\in J_t\}$ are independent 
under the same conditioning, then the conditional expectation (\ref{eq:cond4}) will factorize over $j\in J_t$ (the expectation and product may be interchanged) and the result follows from (i). 

We establish conditional independence by an inductive argument.
Denote $\cZ_t=\{z_{J,s}: s\in [t_J,t]\cap \Omega\}$
and assume conditional on $Z_{J,t-1}, J_{t-1}\subset \hat K_t$ that the $u_j: j\in J_{t-1}$ and mutually independent with each uniformly distributed on $(0,p_{j,t-1})$. Note the weights $w_{i,t}: i\in J_t$ determined by $J_{t-1}\subset \hat K_t$ since arrivals are non-random. Further conditioning on $J_t\in \hat K_t$ result in each $i$ being uniform on $(0, \min\{p_{i,t-1}, w_{i,t}/z_{J,t}\}] = (0, p_{i,t}]$ so completing the induction. The property is trivial at the time $t_i$ of first arrival of each edge. The form (iii) then follows inductively on the size of the subgraph $J$ on expanding the product and taking expectations.
\end{proof}

\begin{proof}[Proof of Theorem~\ref{thm:var-c}]
Here we specify $\hat V_{e,t}$ being commutable from the first $t$ arrivals to mean that it is $\cF_{t}$-measurable, where $\cF_{t}$ is set of random variables $\{u_{e_t}: t\in\Omega\}$ generated up to time $t$. 
By the Law of Total Variance

\bea \kern-10pt
\var(\hat C_{e,t})&=& \E[\var(\hat C_{e,t})|\cF_{t-1}]+
\var(\E[\hat C_{e,t}|\cF_{t-1}])\\
&=&
\E[\left(\frac{\hat C_{e,t-1}+c_{e,t}}{q_{e,t}}\right)^2 \var(I(B_e(z_t))| \cF_{t-1}]
+
\var(\hat C_{e,t-1} +c_{e,t}) \\
&=&
\E[\left(\frac{\hat C_{e,t-1}+c_{e,t}}{q_{e,t}}\right)^2 q_t(1-q_t)] + \var(\hat C_{e,t-1}) 
\eea

Since $\tilde V_{e,t} := \left(\frac{\hat C_{e,t-1}+c_{e,t}}{q_{e,t}}\right)^2 q_t(1-q_t)$ is 
$\cF_{t-1}$=measurable, then $\tilde V_{e,t} I(B_e(z_t))/q_{e,t}$ is $\cF_{t}$-measurable, and 
\be\E[\frac{I(B_e(z_t))}{q_{e,t}} \tilde V_{e,t}]=
\E[\E[\frac{I(B_e(z_t))}{q_{e,t}}|\cF_{e,t}] \tilde V_{e,t}]=
\E[\tilde V_{e,t}]
\ee
and similarly by assumption on $\hat V_{e,t-1}$,
\be
\E[\frac{I(B_e(z_t))}{q_{e,t}} \hat V_{e,t-1}]=\E[\hat V_{e,t-1}]=\var(\hat C_{e,t-1})
\ee
\end{proof}

\begin{proof}[Proof of Theorem~\ref{thm:decay}]
(i) follows by linearity of expectation, while (ii) follows by substitution in (\ref{eq:var-c}).
\end{proof}

\end{document}